\documentclass{lmcs}
\pdfoutput=1

\lmcsdoi{16}{2}{11}
\lmcsheading{}{\pageref{LastPage}}{}{}%
{Mar.~15,~2019}{Jun.~04,~2020}{}


\usepackage{amssymb}
\usepackage{amsfonts}
\usepackage{mathtools}
\usepackage{calc}
\usepackage{multicol}
\usepackage{multirow}
\usepackage{graphicx}
\usepackage{needspace}
\usepackage{amsmath}
\usepackage{stmaryrd}
\usepackage{xspace}
\usepackage{tikz}
\usepackage{pgfplots}
\usepackage{nicefrac}
\usetikzlibrary{decorations.markings}
\usetikzlibrary{matrix,backgrounds,folding}
\usetikzlibrary{shapes.geometric}
\pgfdeclarelayer{edgelayer}
\pgfdeclarelayer{nodelayer}
\pgfsetlayers{background,edgelayer,nodelayer,main}
\tikzstyle{every picture}=[baseline=-0.25em]
\tikzstyle{none}=[inner sep=0mm]
\tikzstyle{zxnode}=[shape=circle, minimum width=.25cm, inner sep=0.5pt, font=\footnotesize, draw=black]
\tikzstyle{gn}=[zxnode,fill=green]
\tikzstyle{rn}=[zxnode,fill=red]
\tikzstyle{H box}=[rectangle,fill=yellow,draw=black,xscale=1,yscale=1,font=\footnotesize,inner sep=1.2pt,minimum width=0.15cm,minimum height=0.15cm]
\tikzstyle{ug}=[regular polygon, regular polygon sides=3, fill=red,draw=black,inner sep = 0pt,minimum width=0.8em]
\tikzstyle{black dot}=[inner sep=0.6mm,minimum width=0pt,minimum height=0pt,fill=black,draw=black,shape=circle]
\tikzstyle{dot}=[black dot]
\tikzstyle{white dot}=[dot,fill=white]
\tikzstyle{zwcross}=[diamond, draw, fill=gray, minimum width=0em, inner sep=1.5pt]
\tikzstyle{arrow}=[decoration={markings,mark=at position 1 with
    {\arrow[scale=1.2,>=stealth]{>}}},postaction={decorate}]

\tikzstyle{st}=[star,star points = 5, fill=white,draw=black,inner sep = 1.2pt,line width=1.2pt]
\tikzstyle{gnlabel}=[rounded corners=0.2em,fill=green!20,inner sep=0.1em,font=\scriptsize, anchor=west, xshift=-0.2em, yshift=0,opacity=1]
\tikzstyle{uglabel}=[rounded corners=0.2em,fill=red!20,inner sep=0.1em,font=\scriptsize, anchor=west, xshift=0.1em, yshift=-0.2em,opacity=1]%,text opacity=1]

\pgfdeclarelayer{edgelayer}
\pgfdeclarelayer{nodelayer}
\pgfsetlayers{background,edgelayer,nodelayer,main}
\tikzstyle{none}=[inner sep=0mm]
\tikzstyle{every loop}=[]

\def\fig{}

\newcommand{\callrule}[2]{\hyperlink{r:#1}{\textnormal{(#2)}}\xspace}

\newcommand{\s}{\callrule{rules}{S}}
\newcommand{\id}{\callrule{rules}{I}}

\newcommand{\cp}{\callrule{rules}{CP}}
\newcommand{\bialg}{\callrule{rules}{B}}
\newcommand{\hd}{\callrule{rules}{HD}}
\newcommand{\h}{\callrule{rules}{H}}
\newcommand{\picom}{\callrule{rules}{K}}

\newcommand{\iv}{\callrule{rules}{IV}}
\newcommand{\zo}{\callrule{rules}{ZO}}

\newcommand{\supp}{\callrule{cliff-T-rules}{SUP}}
\newcommand{\ccom}{\callrule{cliff-T-rules}{C}}
\newcommand{\bw}{\callrule{cliff-T-rules}{BW}}
\newcommand{\e}{\callrule{cliff-T-rules}{E}}

\newcommand{\add}{\callrule{rule-A}{A}}

\newcommand{\td}{\callrule{param-triangle-rules}{TD}}
\newcommand{\ta}{\callrule{param-triangle-rules}{TA}}

\newcommand{\eq}[2][~~]{
#1
\underset{\substack{#2}}{=}
#1
}

\newcommand{\equi}[2][\quad]{
#1
\underset{\substack{#2}}{\iff}
#1
}

\newcommand{\fit}[1]{\resizebox{\columnwidth}{!}{#1}}
\newcommand{\interp}[1]{\left\llbracket#1\right\rrbracket}
\newcommand{\frag}[1]{$\frac{\pi}{#1}$-frag\-ment}
\newcommand{\titlerule}[1]{\noindent
\begin{center}
\raisebox{0.5ex}{\rule{(\textwidth-\widthof{#1})/2}{0.5pt}}#1\raisebox{0.5ex}{\rule{(\textwidth-\widthof{#1})/2}{0.5pt}}
\end{center}}
\newcommand{\annoted}[3]{{\scriptstyle#1}\left\lbrace\mathrlap{\phantom{#3}}\right.\overbrace{#3}^{#2}}

\newcommand{\piq}[1]{\ensuremath{e^{i\frac{#1\pi}{4}}}}
\newcommand{\bra}[1]{\ensuremath{\left\langle #1 \right|}}
\newcommand{\ket}[1]{\ensuremath{\left|  #1 \right\rangle}}

\def\zx{\textnormal{ZX}\xspace}
\def\zxc{\textnormal{ZX$_{\frac{\pi}{2}}$}\xspace}
\def\zxct{\textnormal{ZX$_{\frac{\pi}{4}}$}\xspace}
\def\zxt{\textnormal{ZX$_{\text{T}}$}\xspace}
\def\zw{\textnormal{ZW}\xspace}
\def\zwc{\textnormal{ZW$_\mathbb{C}$}\xspace}
\def\zwh{\textnormal{ZW\texorpdfstring{\!$_\frac{1}{2}$}{(1/2)}}\xspace}

\newcommand{\half}{%
    \begin{tikzpicture}
		\node [style=st] (0) at (0,0) {};
    \end{tikzpicture}
}
\newcommand{\two}{%
    \begin{tikzpicture}
		\draw (0,0) circle (0.25) ;
    \end{tikzpicture}
}

\newcommand{\nmcrossI}[2]{%
    \scalebox{0.8}{\begin{tikzpicture}
	\begin{pgfonlayer}{nodelayer}
		\node [style=none] (0) at (0.25, 0.5) {};
		\node [style=none] (1) at (0.25, -0.5) {};
		\node [style=none] (2) at (0.5, 0.5) {};
		\node [style=none] (3) at (0.5, -0.5) {};
		\node [style=none] (4) at (0, -0.5) {};
		\node [style=none] (5) at (-0.5, -0.5) {};
		\node [style=none] (6) at (-1, -0.5) {};
		\node [style=none] (7) at (-1.5, -0.5) {};
		\node [style=none] (8) at (0, 0.5) {};
		\node [style=none] (9) at (-0.5, 0.5) {};
		\node [style=none] (10) at (-1, 0.5) {};
		\node [style=none] (11) at (-1.5, 0.5) {};
		\node [style=none] (12) at (-1.25, -0.5) {$~\cdots~$};
		\node [style=none] (13) at (-1.25, 0.5) {$~\cdots~$};
		\node [style=none] (14) at (-0.2500002, 0.5) {$~\cdots~$};
		\node [style=none] (15) at (-0.2500002, -0.5) {$~\cdots~$};
		\node [style=none] (16) at (-1.25, 0.75) {$#1$};
		\node [style=none] (17) at (-0.2500002, 0.75) {$#2$};
	\end{pgfonlayer}
	\begin{pgfonlayer}{edgelayer}
		\draw (0.center) to (1.center);
		\draw (2.center) to (3.center);
		\draw [in=90, out=-90, looseness=0.75] (8.center) to (6.center);
		\draw [in=90, out=-90, looseness=0.75] (9.center) to (7.center);
		\draw [in=90, out=-90, looseness=0.75] (11.center) to (5.center);
		\draw [in=90, out=-90, looseness=0.75] (10.center) to (4.center);
	\end{pgfonlayer}
\end{tikzpicture}
}}

\newcommand{\rz}[1]{~\begin{tikzpicture}
	\begin{pgfonlayer}{nodelayer}
		\node [style=gn] (0) at (0, 0.25) {#1};
		\node [style=none] (1) at (0, -0.25) {};
	\end{pgfonlayer}
	\begin{pgfonlayer}{edgelayer}
		\draw [style=none] (0) to (1.center);
	\end{pgfonlayer}
\end{tikzpicture}%
~}

\newcommand{\rx}[1]{~\begin{tikzpicture}
	\begin{pgfonlayer}{nodelayer}
		\node [style=rn] (0) at (0, 0.25) {#1};
		\node [style=none] (1) at (0, -0.25) {};
	\end{pgfonlayer}
	\begin{pgfonlayer}{edgelayer}
		\draw [style=none] (0) to (1.center);
	\end{pgfonlayer}
\end{tikzpicture}%
~}
  \usepackage{lastpage}
  \usepackage{hyperref}
\usepackage[utf8]{inputenc}

\allowdisplaybreaks%

\pgfplotsset{compat=1.14}

\begin{document}

\bibliographystyle{alpha}

\title[ZX-Calculus Completeness]{Completeness of the ZX-Calculus}
%\titlecomment{Merge of two papers}

\author[E.~Jeandel]{Emmanuel Jeandel}	%required
\address{Universit\'e de Lorraine, CNRS, Inria, LORIA, F 54000 Nancy, France}
\email{\{emmanuel.jeandel, simon.perdrix, renaud.vilmart\}@loria.fr}

\author[S.~Perdrix]{Simon Perdrix}	%optional
\author[R.~Vilmart]{Renaud Vilmart}	%optional

\newcommand{\ropen}[1]{[#1)} % chktex 9
\newcommand{\lopen}[1]{(#1]} % chktex 9

\begin{abstract}
  The ZX-Calculus is a graphical language for diagrammatic reasoning in quantum mechanics and quantum information theory. It comes equipped with an equational presentation. We focus here on a very important property of the language: completeness, which roughly ensures the equational theory captures all of quantum mechanics. We first improve on the known-to-be-complete presentation for the so-called Clifford fragment of the language --- a restriction that is not universal --- by adding some axioms. Thanks to a system of back-and-forth translation between the ZX-Calculus and a third-party complete graphical language, we prove that the provided axiomatisation is complete for the first approximately universal fragment of the language, namely Clifford+T.

  We then prove that the expressive power of this presentation, though aimed at achieving completeness for the aforementioned restriction, extends beyond Clifford+T, to a class of diagrams that we call linear with Clifford+T constants. We use another version of the third-party language --- and an adapted system of back-and-forth translation --- to complete the language for the ZX-Calculus as a whole, that is, with no restriction. We briefly discuss the added axioms, and finally, we provide a complete axiomatisation for an altered version of the language which involves an additional generator, making the presentation simpler.
\end{abstract}

\maketitle

\section*{Introduction}%
\label{sec:intro}

The ZX-Calculus is a powerful graphical language for quantum reasoning and quantum computing introduced by Bob Coecke and Ross Duncan~\cite{interacting}.
The language comes with a way of interpreting any ZX-diagram as a matrix --- called the \emph{standard interpretation}. Two diagrams represent the same quantum evolution when they have the same standard interpretation. The language is also equipped with a set of axioms --- transformation rules --- which are sound, i.e.~they preserve the standard interpretation. Their purpose is to explain how a diagram can be transformed into an equivalent one.

The ZX-calculus has several applications in quantum information processing~\cite{picturing-qp} (e.g.~measurement-based quantum computing~\cite{mbqc,horsman2011quantum,duncan2013mbqc}, quantum codes~\cite{verifying-color-code,duncan2014verifying,chancellor2016coherent,de2017zx}, foundations~\cite{toy-model-graph,duncan2016hopf}), and can be used through the interactive theorem prover Quantomatic~\cite{quanto,kissinger2015quantomatic}. However, the main obstacle to wider use of the ZX-calculus was the absence of a \emph{completeness} result for a \emph{universal} fragment of quantum mechanics --- a restriction of the language that can be used to represent any quantum operator. Completeness would guarantee that, in this fragment, any true property is provable using the ZX-calculus. More precisely, the language would be complete if, given any two diagrams representing the same matrix, one could transform one diagram into the other using the axioms of the language.
Completeness is crucial,  it means in particular that all the fundamental properties of quantum mechanics are captured by the graphical rules.

The ZX-Calculus, as it was initially introduced, has been proved to be incomplete in general~\cite{incompleteness}, and despite the necessary axioms that have since been identified~\cite{supplementarity,cyclo}, the language remained incomplete. However, several fragments of the language have been proved to be complete (\frag{2}~\cite{pi_2-complete}; $\pi$-fragment~\cite{pivoting}; single-qubit \frag{4}~\cite{pi_4-single-qubit}), but none of them are \emph{universal} for quantum mechanics, even approximately. In particular all quantum algorithms expressible in these fragments are efficiently simulable on a classical computer.
%Several fragments -- restrictions -- of the language have been proved to be complete (\frag{2}~\cite{pi_2-complete}; $\pi$-fragment~\cite{pivoting}; single-qubit \frag{4}~\cite{pi_4-single-qubit}), but none of them are \emph{universal} for quantum mechanics, even approximately. In particular all quantum algorithms expressible in these fragments are efficiently simulable on a classical computer.

%We tackle this problem here, by giving the first complete axiomatisation for Clifford+T quantum mechanics. We then extend the result of completeness to a larger fragment of the language, without adding new axioms. We then use this
%As a consequence, most of the attention has been paid to find a complete axiomatisation of the \frag{4} of the ZX-Calculus for the many-qubit Clifford+T quantum mechanics, the simplest approximately universal fragment of quantum mechanics, which is widely used in quantum computing.

We tackle this problem, and present here four main results:
\begin{itemize}
\item A complete axiomatisation for Clifford+T quantum mechanics
\item A complete axiomatisation for the so-called linear diagrams with constants in Clifford+T
\item A complete axiomatisation for unrestricted ZX-Calculus
\item A complete axiomatisation for a version of (unrestricted) ZX-Calculus with additional generator
\end{itemize}

\noindent
To do so, we are going to use another graphical language called the ZW-Calculus, based on the interactions of the GHZ and W states~\cite{ghz-w}. The two versions of the language we are going to exploit --- denoted \zw and \zwc --- have been proven to be complete through the use of normal forms~\cite{zw,Amar}.

\textbf{Completeness for Clifford+T.} The first completeness result will use the first version of the ZW-Calculus. However, its diagrams only represent matrices over $\mathbb{Z}$, and hence it is not approximately universal. We introduce the \zwh-calculus, a simple extension of the ZW-Calculus which remains  complete and in which any matrix over the dyadic rational numbers can be represented. % and still keep its completeness property.
We then introduce two interpretations from the ZX-Calculus  to  the \zwh-Calculus and back. Thanks to these interpretations, we derive the completeness of the \frag 4 of the ZX-Calculus from the completeness of the \zwh-Calculus. Notice that the interpretation of ZX-diagrams (which represent complex matrices)  into \zwh-diagrams (which represent dyadic rationals) requires a non trivial encoding.
Notice also that this approach provides a completion procedure. Roughly speaking each axiom of the \zwh-calculus generates an equation in the ZX-calculus: if this equation is not already provable using the existing axioms of the ZX-calculus one can treat this equality as an new axiom. A great part of the work has been to reduce all these equalities to only two additional axioms for the language.

\textbf{Completeness for linear diagrams with constants in Clifford+T.} We extend the result to a restriction of the ZX-Calculus that contains Clifford+T. In some axioms and derivations, it is customary to consider some parameters as variables. Such a derivation can then be seen as a uniform proof, i.e.~a proof that is valid for any valuation of the variables in the diagrams. In linear diagrams, some angles are considered as variables, and some other are constant. The accepted parameters are then linear combinations of the variables and the constant angles. We show that any equation between two linear diagrams with constants in Clifford+T, that is sound for all valuations of its variables, is provable using the axioms of Clifford+T ZX-Calculus. This result proves very useful for the following results.

\textbf{Completeness for unrestricted ZX-Calculus.} To find a complete axiomatisation for the general ZX-Calculus, we use the same technique as for the Clifford+T completeness (back and forth interpretations), but this time using the \zwc-Calculus which is universal, i.e.~its diagrams can represent any matrix over $\mathbb{C}$ with dimension a power of two. Again, the interpretation from \zwc to ZX provide a range of equations that the latter has to be able to prove. Most of them are derivable thanks to the previous result: the completeness for linear diagrams with constants in Clifford+T. The other ones lead to an additional axiom for the ZX-Calculus. This axiom is ``non-linear'', in the sense that that the parameters on one side of the equality are related in a non-linear manner to those on the other side. This is necessary: the argument for incompleteness~\cite{incompleteness} can be adapted to any linear axiomatisation.

\textbf{Completeness for ZX-Calculus with parametrised triangles.} The proof of completeness for Clifford+T heavily relies on a particular Clifford+T diagram, which performs a non-trivial, non-unitary matrix% $\begin{pmatrix}1&1\\0&1\end{pmatrix}$
. To simplify the proofs, we gave it a syntactic sugar: \resizebox{!}{0.9em}{
%\InputIfFileExists{#1.tikz}{}{
\begin{tikzpicture}
	\begin{pgfonlayer}{nodelayer}
		\node [style=ug, yshift=-1pt] (0) at (0, -0) {};
		\node [style=none] (1) at (0, 0.25) {};
		\node [style=none] (2) at (0, -0.25) {};
	\end{pgfonlayer}
	\begin{pgfonlayer}{edgelayer}
		\draw [style=none] (1.center) to (2.center);
	\end{pgfonlayer}
\end{tikzpicture}%} % chktex 27
}. In order to provide a simplified axiomatisation for universal quantum mechanics, we propose in the last section to add to the language a parametrised version of the triangle, motivated by the behaviour of said matrix %the fact that $\begin{pmatrix}1&x\\0&1\end{pmatrix}\begin{pmatrix}1&y\\0&1\end{pmatrix} = \begin{pmatrix}1&x{+}y\\0&1\end{pmatrix}$
. The introduction of this generator greatly simplifies two axioms at the cost of an additional one: its decomposition.

\textbf{Related works.} After the result for Clifford+T came out, another axiomatisation was found for unrestricted ZX-Calculus~\cite{HNW}, but for an altered version of the ZX-Calculus, that is, a version where two generators were added, namely the triangle --- which we have introduced as a syntactic sugar for a diagram of Clifford+T --- and the $\lambda$-box, used to simplify the translation from ZW to ZX\@. This axiomatisation was proven to be complete using the same method of back and forth interpretation with the \zwc-Calculus. It has a significantly higher amount of rules than the one we obtain for unrestricted ZX-Calculus.

This paper condenses two articles~\cite{JPV,JPV-universal}, and extends them with the notion of parametrised triangle. It is structured as follows. Sections~\ref{sec:zx} and~\ref{sec:zw} present respectively the ZX-Calculi and the ZW-Calculi. Section~\ref{sec:cliff-t} is dedicated to the proof of completeness for Clifford+T, while Section~\ref{sec:lin-diag} extends the result to linear diagrams. Section~\ref{sec:appli-lin-diag} provides methods to simplify the use of the latter result and Section~\ref{sec:gen-zx} makes use of them to show the provided axiomatisation for unrestricted ZX-Calculus is complete. We briefly discuss the axioms we have added so far in Section~\ref{sec:new-axioms}, and finally provide an alternative ZX-Calculus with parametrised triangles in Section~\ref{sec:param-triangle}.

\section{The ZX-Calculus}%
\label{sec:zx}

\subsection{Syntax and Semantics}

In the ZX-Calculus, quantum operators are represented as diagrams.
A ZX-diagram $D:k\to l$ with $k$ inputs and $l$ outputs is generated by:\\
%\begin{center}
%\bgroup
%\def\arraystretch{2.5}
%{\begin{tabular}{|cc|cc|}
%\hline
%$R_Z^{(n,m)}(\alpha):n\to m$ & \tikzfig{gn-alpha} & $R_X^{(n,m)}(\alpha):n\to m$ & \tikzfig{rn-alpha}\\[4ex]\hline
%$H:1\to 1$ & \tikzfig{Hadamard} & $e:0\to 0$ & \tikzfig{empty-diagram}\\\hline
%$\mathbb{I}:1\to 1$ & \tikzfig{single-line} & $\sigma:2\to 2$ & \tikzfig{crossing}\\\hline
%$\epsilon:2\to 0$ & \tikzfig{cup} & $\eta:0\to 2$ & \tikzfig{caps}\\\hline
%%$U:1\to 1$ & \tikzfig{ug-node} &&\\\hline
%\end{tabular}}
%\egroup\\
%where $n,m\in \mathbb{N}$ and $\alpha \in \mathbb{R}$. The generator $e$ is the empty diagram.
%\end{center}
\begin{minipage}{\columnwidth}
\titlerule{ ZX-Generators }
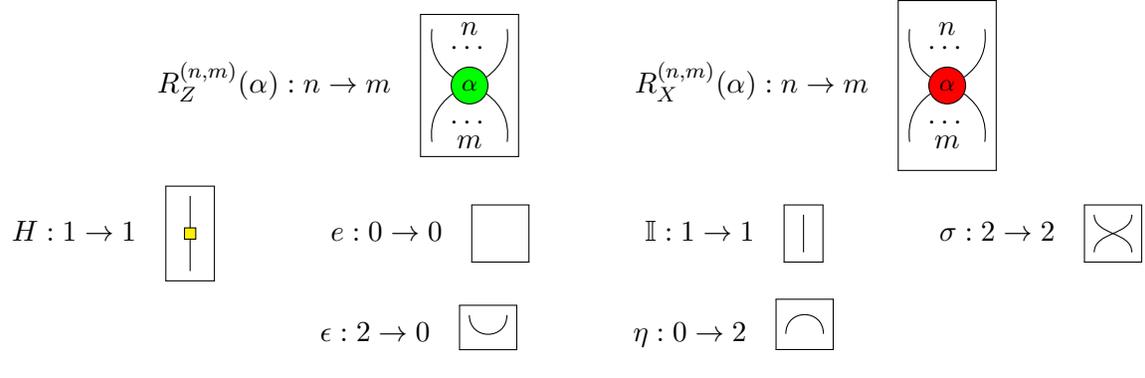
\[R_Z^{(n,m)}(\alpha):n\to m\quad\boxed{
%\InputIfFileExists{#1.tikz}{}{
\input{./figures/gn-alpha.tikz}%} % chktex 27
} \qquad\qquad R_X^{(n,m)}(\alpha):n\to m\quad\boxed{
%\InputIfFileExists{#1.tikz}{}{
\input{./figures/rn-alpha.tikz}%} % chktex 27
}\]
\end{minipage}
\[H:1\to 1\quad\boxed{~~
%\InputIfFileExists{#1.tikz}{}{
\begin{tikzpicture}
	\begin{pgfonlayer}{nodelayer}
		\node [style={H box}] (0) at (0, 0) {};
		\node [style=none] (1) at (0, 0.5) {};
		\node [style=none] (2) at (0, -0.5) {};
	\end{pgfonlayer}
	\begin{pgfonlayer}{edgelayer}
		\draw (2.center) to (1.center);
	\end{pgfonlayer}
\end{tikzpicture}
%} % chktex 27
~~}\qquad\qquad e:0\to 0 \quad\boxed{\phantom{
%\InputIfFileExists{#1.tikz}{}{
\input{./figures/empty-diagram.tikz}%} % chktex 27
}}\qquad\qquad\mathbb{I}:1\to 1\quad\boxed{~~
%\InputIfFileExists{#1.tikz}{}{
\begin{tikzpicture}
	\begin{pgfonlayer}{nodelayer}
		\node [style=none] (0) at (0, 0.2499999) {};
		\node [style=none] (1) at (0, -0.2499999) {};
	\end{pgfonlayer}
	\begin{pgfonlayer}{edgelayer}
		\draw (0.center) to (1.center);
	\end{pgfonlayer}
\end{tikzpicture}%} % chktex 27
~~} \qquad\qquad\sigma:2\to 2\quad\boxed{
%\InputIfFileExists{#1.tikz}{}{
\input{./figures/crossing.tikz}%} % chktex 27
}\]
\begin{minipage}{\columnwidth}
\[\epsilon:2\to 0\quad\boxed{
%\InputIfFileExists{#1.tikz}{}{
\begin{tikzpicture}
	\begin{pgfonlayer}{nodelayer}
		\node [style=none] (0) at (-0.2500001, 0.2500001) {};
		\node [style=none] (1) at (0.2500001, 0.2500001) {};
	\end{pgfonlayer}
	\begin{pgfonlayer}{edgelayer}
		\draw [bend right=90, looseness=1.75] (0.center) to (1.center);
	\end{pgfonlayer}
\end{tikzpicture}%} % chktex 27
}\qquad\qquad\eta:0\to 2\quad\boxed{
%\InputIfFileExists{#1.tikz}{}{
\begin{tikzpicture}
	\begin{pgfonlayer}{nodelayer}
		\node [style=none] (0) at (-0.2500001, -0) {};
		\node [style=none] (1) at (0.2500001, -0) {};
	\end{pgfonlayer}
	\begin{pgfonlayer}{edgelayer}
		\draw [bend left=90, looseness=1.75] (0.center) to (1.center);
	\end{pgfonlayer}
\end{tikzpicture}%} % chktex 27
}\]
\rule{\columnwidth}{0.5pt}
\end{minipage}

\vspace{0.2cm}
and the two compositions:
\begin{itemize}
\item Spacial Composition: for any $D_1:a\to b$ and $D_2:c\to d$, $D_1\otimes D_2:a+c\to b+d$ consists in placing $D_1$ and $D_2$ side by side, $D_2$ on the right of $D_1$.
\item Sequential Composition: for any $D_1:a\to b$ and $D_2:b\to c$, $D_2\circ D_1:a\to c$ consists in placing $D_1$ on the top of $D_2$, connecting the outputs of $D_1$ to the inputs of $D_2$.\footnote{
In order to be simpler, the present article eludes part of the underlying formalism (strict compact closed symmetric monoidal categories), which makes the language unambiguous. For more information, see e.g.~\cite{selinger2010survey}.}
\end{itemize}

\noindent
The standard interpretation of the ZX-diagrams associates to any diagram $D:n\to m$ a linear map $\interp{D}:\mathbb{C}^{2^n}\to\mathbb{C}^{2^m}$ inductively defined as follows:%, using Dirac notations, an equivalent, matrix-based description of the semantics is given in the appendix:\\
\vspace{1em}
\titlerule{$\interp{.}$}
\[ \interp{D_1\otimes D_2}:=\interp{D_1}\otimes\interp{D_2} \qquad
\interp{D_2\circ D_1}:=\interp{D_2}\circ\interp{D_1}\]
\[\interp{
%\InputIfFileExists{#1.tikz}{}{
\input{./figures/empty-diagram.tikz}%} % chktex 27
~}:=\begin{pmatrix}
1
\end{pmatrix} \qquad
\interp{~
%\InputIfFileExists{#1.tikz}{}{
\begin{tikzpicture}
	\begin{pgfonlayer}{nodelayer}
		\node [style=none] (0) at (0, 0.2499999) {};
		\node [style=none] (1) at (0, -0.2499999) {};
	\end{pgfonlayer}
	\begin{pgfonlayer}{edgelayer}
		\draw (0.center) to (1.center);
	\end{pgfonlayer}
\end{tikzpicture}%} % chktex 27
~~}:= \begin{pmatrix}
1 & 0 \\ 0 & 1\end{pmatrix}\qquad
\interp{~
%\InputIfFileExists{#1.tikz}{}{
\begin{tikzpicture}
	\begin{pgfonlayer}{nodelayer}
		\node [style={H box}] (0) at (0, 0) {};
		\node [style=none] (1) at (0, 0.5) {};
		\node [style=none] (2) at (0, -0.5) {};
	\end{pgfonlayer}
	\begin{pgfonlayer}{edgelayer}
		\draw (2.center) to (1.center);
	\end{pgfonlayer}
\end{tikzpicture}
%} % chktex 27
~}:= \frac{1}{\sqrt{2}}\begin{pmatrix}1 & 1\\1 & -1\end{pmatrix}\]
\[\interp{
%\InputIfFileExists{#1.tikz}{}{
\input{./figures/crossing.tikz}%} % chktex 27
}:= \begin{pmatrix}
1&0&0&0\\
0&0&1&0\\
0&1&0&0\\
0&0&0&1
\end{pmatrix} \qquad
\interp{\raisebox{-0.25em}{$
%\InputIfFileExists{#1.tikz}{}{
\begin{tikzpicture}
	\begin{pgfonlayer}{nodelayer}
		\node [style=none] (0) at (-0.2500001, 0.2500001) {};
		\node [style=none] (1) at (0.2500001, 0.2500001) {};
	\end{pgfonlayer}
	\begin{pgfonlayer}{edgelayer}
		\draw [bend right=90, looseness=1.75] (0.center) to (1.center);
	\end{pgfonlayer}
\end{tikzpicture}%} % chktex 27
$}}:= \begin{pmatrix}
1&0&0&1
\end{pmatrix} \qquad
\interp{\raisebox{-0.35em}{$
%\InputIfFileExists{#1.tikz}{}{
\begin{tikzpicture}
	\begin{pgfonlayer}{nodelayer}
		\node [style=none] (0) at (-0.2500001, -0) {};
		\node [style=none] (1) at (0.2500001, -0) {};
	\end{pgfonlayer}
	\begin{pgfonlayer}{edgelayer}
		\draw [bend left=90, looseness=1.75] (0.center) to (1.center);
	\end{pgfonlayer}
\end{tikzpicture}%} % chktex 27
$}}:= \begin{pmatrix}
1\\0\\0\\1
\end{pmatrix}\]
For any $\alpha\in\mathbb{R}$, $\interp{
    \begin{tikzpicture}
	\begin{pgfonlayer}{nodelayer}
		\node [style=gn] (0) at (0, -0) {$\alpha$};
	\end{pgfonlayer}
    \end{tikzpicture}
}
:=\begin{pmatrix}1+e^{i\alpha}\end{pmatrix}$, and  for any $n,m\geq 0$ such that $n+m>0$:
\[
%\interp{\begin{tikzpicture}
%	\begin{pgfonlayer}{nodelayer}
%		\node [style=gn] (0) at (0, -0) {$\alpha$};
%	\end{pgfonlayer}
%\end{tikzpicture}}:=\begin{pmatrix}1+e^{i\alpha}\end{pmatrix} \qquad
\interp{
%\InputIfFileExists{#1.tikz}{}{
\input{./figures/gn-alpha.tikz}%} % chktex 27
}:=
\annoted{2^m}{2^n}{\begin{pmatrix}
  1 & 0 & \cdots & 0 & 0 \\
  0 & 0 & \cdots & 0 & 0 \\
  \vdots & \vdots & \ddots & \vdots & \vdots \\
  0 & 0 & \cdots & 0 & 0 \\
  0 & 0 & \cdots & 0 & e^{i\alpha}
 \end{pmatrix}}
%~~\begin{pmatrix}n+m>0\end{pmatrix}
\]
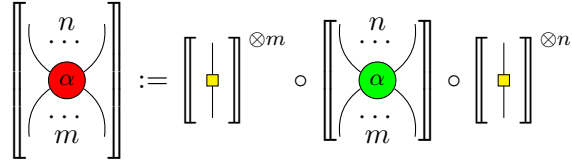
\begin{minipage}{\columnwidth}
\[\interp{
%\InputIfFileExists{#1.tikz}{}{
\input{./figures/rn-alpha.tikz}%} % chktex 27
}:=\interp{~
%\InputIfFileExists{#1.tikz}{}{
\begin{tikzpicture}
	\begin{pgfonlayer}{nodelayer}
		\node [style={H box}] (0) at (0, 0) {};
		\node [style=none] (1) at (0, 0.5) {};
		\node [style=none] (2) at (0, -0.5) {};
	\end{pgfonlayer}
	\begin{pgfonlayer}{edgelayer}
		\draw (2.center) to (1.center);
	\end{pgfonlayer}
\end{tikzpicture}
%} % chktex 27
~}^{\otimes m}\circ \interp{
%\InputIfFileExists{#1.tikz}{}{
\input{./figures/gn-alpha.tikz}%} % chktex 27
}\circ \interp{~
%\InputIfFileExists{#1.tikz}{}{
\begin{tikzpicture}
	\begin{pgfonlayer}{nodelayer}
		\node [style={H box}] (0) at (0, 0) {};
		\node [style=none] (1) at (0, 0.5) {};
		\node [style=none] (2) at (0, -0.5) {};
	\end{pgfonlayer}
	\begin{pgfonlayer}{edgelayer}
		\draw (2.center) to (1.center);
	\end{pgfonlayer}
\end{tikzpicture}
%} % chktex 27
~}^{\otimes n}\]
$\left(\text{where }M^{\otimes 0}=\begin{pmatrix}1\end{pmatrix}\text{ and }M^{\otimes k}=M\otimes M^{\otimes k-1}\text{ for any }k\in \mathbb{N}^*\right)$.
\vspace{0.2em}
\end{minipage}
\rule{\columnwidth}{0.5pt}

To simplify, the red and green nodes will be represented empty when holding a 0 angle:
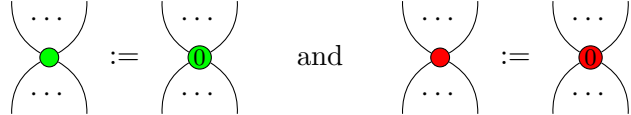
\[ 
%\InputIfFileExists{#1.tikz}{}{
\input{./figures/gn-empty-is-gn-zero.tikz}%} % chktex 27
 \qquad\text{and}\qquad 
%\InputIfFileExists{#1.tikz}{}{
\input{./figures/rn-empty-is-rn-zero.tikz}%} % chktex 27
 \]

The ZX-Calculus is \emph{universal}, i.e.~for any complex matrix (with dimensions some powers of 2), there exists a diagram that represents it~\cite{interacting}.

\subsection{Fragments and Axiomatisations}

The language comes equipped with an equational theory, or axiomatisation: a set of diagram transformations that can be applied locally while preserving the semantics.
When two diagrams $D_1$ and $D_2$ can be turned into one another using the rules of the ZX-Calculus, we will write $\zx\vdash D_1=D_2$. The fact that the rules can be applied locally is formalised by:
\[\zx\vdash D_1=D_2 \implies \left\lbrace \begin{array}{l}
\zx\vdash D_1\circ D = D_2\circ D\\
\zx\vdash D\circ D_1 = D\circ D_2\\
\zx\vdash D_1\otimes D = D_2\otimes D\\
\zx\vdash D\otimes D_1 = D\otimes D_2\\
\end{array}\right.\]
whenever $D$ has the right type.

Of course, the rules are chosen so that they are sound:
\[\zx\vdash D_1=D_2 \implies \interp{D_1}=\interp{D_2}\]
\textbf{Completeness} is the converse of soundness.
%The converse of soundness is called completeness.
The property of completeness is essential: it implies that the language captures all the properties of quantum mechanics, and also that we can do as much with the graphical language as with matrices and linear algebra, for handling qubits.

The problem of completeness is hard to tackle, plus the parameters in the general ZX-Calculus can take their values in $\mathbb{R}$, which implies handling an uncountable set of angles. It can then be preferred to restrict the language to some finite set of values for the parameters, and study the completeness for these specific restrictions. Hence the ZX actually has several axiomatisations, each of them, while sharing some fundamental axioms, are devoted to specific fragments.

A set of fundamental axioms that are present in \emph{all} axiomatisations of the qubit ZX-Calculus are the ones that fall under the paradigm ``only topology matters''~\cite{interacting}. The language is such that any graph isomorphism (with fixed inputs and outputs) preserves the semantics, so we can bend the wires, and move nodes around as we please, so for instance we have the following equalities:

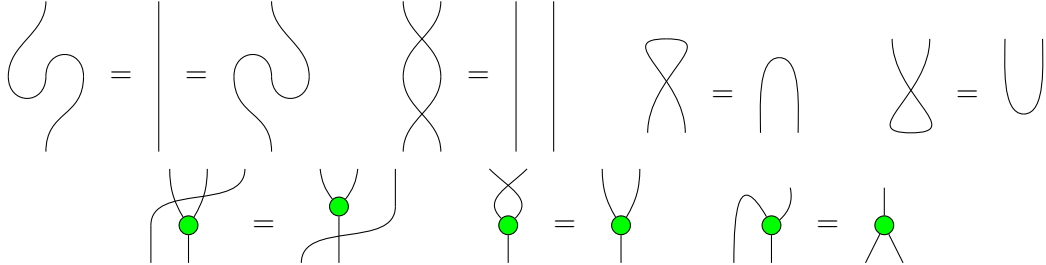
\[{
%\InputIfFileExists{#1.tikz}{}{
\input{./figures/bent-wire.tikz}%} % chktex 27
}\]
\[{
%\InputIfFileExists{#1.tikz}{}{
\input{./figures/bent-wire-2.tikz}%} % chktex 27
}\]

In the following we review some restrictions of the language, and their state-of-the-art axiomatisation.

\subsubsection{Clifford}

We call the Clifford fragment, or the \frag2 of the language, the restriction of ZX-Calculus where the nodes $R_Z$ and $R_X$ have their parameters in $\frac{\pi}{2}\mathbb{Z}$, i.e.~when all the angles in the diagrams are multiples of $\frac{\pi}{2}$. This corresponds to the Clifford fragment of quantum circuits. It is a fragment of prime interest for it is used in error correction schemes~\cite{Gottesman_2006}, as well as a resource for computation in the model of measurement-based quantum computing~\cite{Raussendorf_2003}. However, the fragment is not universal i.e.\ its diagrams cannot approach all quantum evolutions. It even has been shown that Clifford circuits and diagrams can be efficiently simulated by a classical computer~\cite{clifford-not-universal}. However, an axiomatisation is known to be complete for the Clifford fragment of the ZX-Calculus~\cite{pi_2-complete}, and is given in Figure~\ref{fig:ZX-Clifford-rules}. Notice that this axiomatisation is a simplified version of the original one~\cite{simplified-stabilizer} and that one could simplify further by removing the rule \picom, which happens to be derivable from the others in the \frag2. However, since it is not derivable for arbitrary angles, it will be used in the next axiomatisations.

\begin{figure*}[!htb]
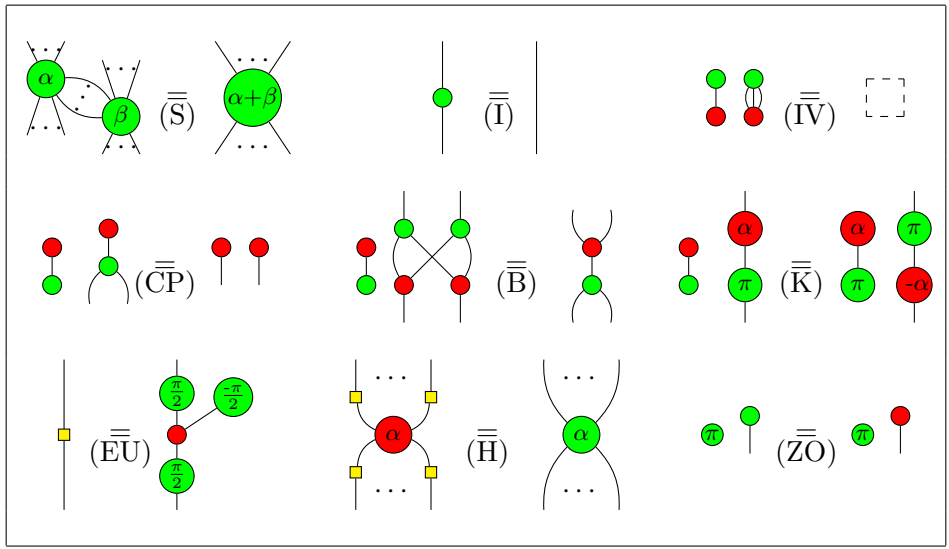

 \centering
 \hypertarget{r:rules}{}
 \begin{tabular}{|c@{$\qquad$}c@{$\qquad$}c|}
   \hline
   && \\
   
%\InputIfFileExists{#1.tikz}{}{
\input{./figures/spider-1.tikz}%} % chktex 27
& 
%\InputIfFileExists{#1.tikz}{}{
\input{./figures/s2-simple.tikz}%} % chktex 27
& 
%\InputIfFileExists{#1.tikz}{}{
\input{./figures/inverse-label.tikz}%} % chktex 27
\\
   && \\
   
%\InputIfFileExists{#1.tikz}{}{
\input{./figures/b1s.tikz}%} % chktex 27
& 
%\InputIfFileExists{#1.tikz}{}{
\input{./figures/b2s.tikz}%} % chktex 27
& 
%\InputIfFileExists{#1.tikz}{}{
\input{./figures/k2s.tikz}%} % chktex 27
\\
   && \\
   
%\InputIfFileExists{#1.tikz}{}{
\input{./figures/euler-decomp-scalar-free.tikz}%} % chktex 27
&  
%\InputIfFileExists{#1.tikz}{}{
\input{./figures/h2.tikz}%} % chktex 27
& 
%\InputIfFileExists{#1.tikz}{}{
\input{./figures/ZO.tikz}%} % chktex 27
\\
   && \\
   \hline
  \end{tabular}
 \caption[]{Set of rules for the Clifford ZX-Calculus with scalars. All of these rules also hold when flipped upside-down, or with the colours red and green swapped. The right-hand side of (E) is an empty diagram. (\ldots) denote zero or more wires, while (\protect\rotatebox{45}{\raisebox{-0.4em}{$\cdots$}}) denote one or more wires.% $\alpha,\beta,\gamma\in\pg$.
 }%
 \label{fig:ZX-Clifford-rules}
\end{figure*}

\subsubsection{Real Stabiliser}

A retriction of the Clifford fragment is the so-called real stabiliser, or the $\pi$-fragment, where all the angles are in $\pi\mathbb{Z}$. This fragment exactly represents graph states, and also has a complete axiomatisation~\cite{pivoting}.

\subsubsection{Clifford+T}

While the Clifford fragment is known to be efficiently simulable and not universal, it actually only needs a little boost to get there. Indeed, the $T$ gate, identified as $R_Z(\frac{\pi}{4})$ in the ZX-Calculus, is enough to bring approximate universality to the Clifford fragment~\cite{toffoli}. This new restriction is called Clifford+T.
%TODO : peut-etre retirer frivoulousness

Although the axiomatisation for Clifford presented in Figure~\ref{fig:ZX-Clifford-rules} is enough to make the 1-qubit Clifford+T fragment complete~\cite{pi_4-single-qubit}, it is not enough to get the result for the many-qubit case. In particular, two equations of the \frag4 --- \supp and \e in the following --- were shown to be unprovable~\cite{supplementarity,cyclo}, and thus were added to the axiomatisation, but the question of the completeness of the resulting set of axioms remained an open problem. One of the main results of the present article is to provide an axiomatisation that builds on the previous one, and to show that it is complete for the fragment. The additional axioms for Clifford+T are given in Figure~\ref{fig:ZX-CliffordT-rules}.

\begin{figure*}[!htb]
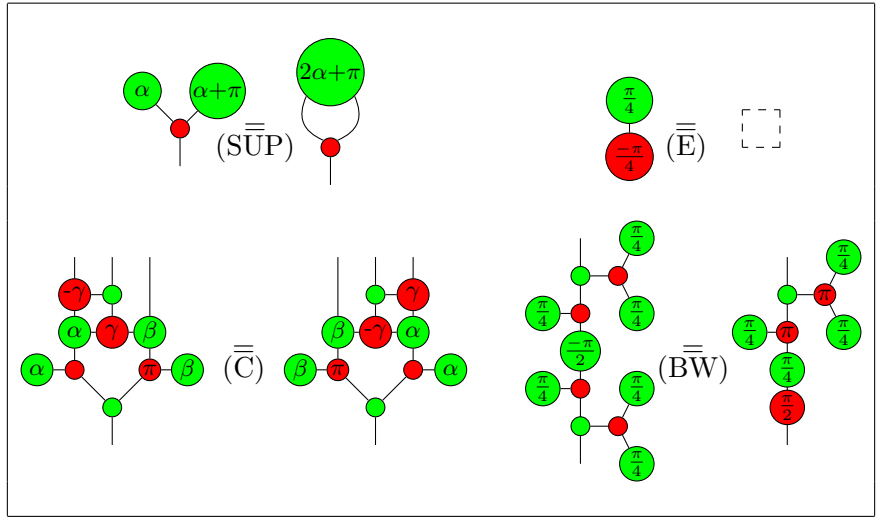

 \centering
 \hypertarget{r:cliff-T-rules}{}
 \begin{tabular}{|c@{$\qquad$}c|}
   \hline
   & \\
   
%\InputIfFileExists{#1.tikz}{}{
\input{./figures/former-supp.tikz}%} % chktex 27
&
%\InputIfFileExists{#1.tikz}{}{
\input{./figures/bicolor_pi_4_eq_empty.tikz}%} % chktex 27
\\
   & \\
   
%\InputIfFileExists{#1.tikz}{}{
\input{./figures/commutation-of-controls-general-simplified.tikz}%} % chktex 27
&
%\InputIfFileExists{#1.tikz}{}{
\input{./figures/BW-simplified.tikz}%} % chktex 27
\\
   & \\
   \hline
  \end{tabular}
 \caption[]{Additional rules for the Clifford+T fragment. Together with $\zxc\setminus\{\iv,\zo\}$, they form the axiomatisation $\zxct$.
 }%
 \label{fig:ZX-CliffordT-rules}
\end{figure*}

\begin{thm}%
\label{thm:cliff-T}
For any diagrams $D_1$ and $D_2$ of the \frag4 of the ZX-Calculus:
\[\interp{D_1}=\interp{D_2} \equi{} \zxct \vdash D_1=D_2\]
\end{thm}

It is also to be noted that a complete axiomatisation has been proposed for the 2-qubit unitaries in Clifford+T~\cite{2-qubits-zx}, although the equality between some pairs of diagrams of the \frag4 can only be achieved by getting out of the fragment. That is, proving $D_1=D_2$ can sometimes only be achieved by proving $D_1=D_3=D_2$, but where $D_3$ is not in the fragment, and has angles that can even be irrational multiples of $\pi$.

\subsubsection{Linear Diagrams}

Linear diagrams are a restriction of the language that stand out from the other fragments, for they are not simply obtained by restricting the angles to be some multiples of a fraction of $\pi$.

In linear diagrams, some angles are considered as variables, and the accepted parameters in green and red nodes are affine combinations of those ($c+\sum n_i\alpha_i$ where $\alpha_i$ are variables and $c$ is constant). Two diagrams are supposed to be equal when they are equal for all valuations of said variables. More details will be provided in Section~\ref{sec:lin-diag}, but we are going to prove the completeness of linear diagrams with constants in $\frac{\pi}{4}\mathbb{Z}$, which will prove very useful for the next completeness result.

\begin{thm}%
\label{thm:lin-diag}
For any diagrams $D_1(\vec\alpha)$ and $D_2(\vec\alpha)$ linear in $\vec\alpha$ with constants in $\frac{\pi}{4}\mathbb{Z}$:
\[\forall \vec{\alpha} \in \mathbb R,~\interp{D_1(\vec\alpha)}=\interp{D_2(\vec\alpha)} \equi{} \forall \vec{\alpha} \in \mathbb R,~\zxct \vdash D_1(\vec\alpha)=D_2(\vec\alpha)\]
\end{thm}

Then, Theorem~\ref{thm:cliff-T} for Clifford+T diagrams can be seen as a particular case of Theorem~\ref{thm:lin-diag}, where the number of variables is $0$.

\subsubsection{General ZX-Calculus}

Even though the axiomatisation $\zxct$ makes an approxiamtely universal fragment of the ZX-Calculus complete, we know that the said axiomatisation does not make the general ZX-Calculus complete. Indeed, an argument provided in~\cite{incompleteness} holds for all the axiomatisations we have provided thus far. However, the last main result of the paper shows that only one additional axiom is needed to get the property for the language with no restriction.

\begin{figure*}[!htb]
 \centering
 \hypertarget{r:rule-A}{}
 \begin{tabular}{|c|}
   \hline\\
   
%\InputIfFileExists{#1.tikz}{}{
\input{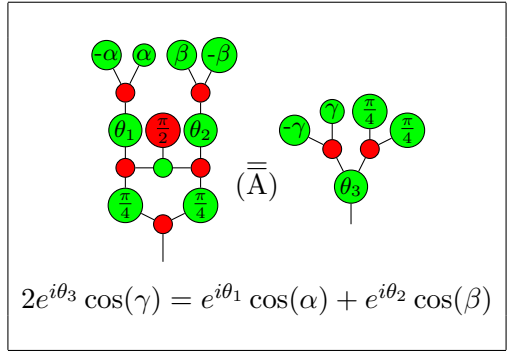}%} % chktex 27
\\\\
   \hline
  \end{tabular}
 \caption[]{Additional rule for the general ZX-Calculus.}%
 \label{fig:ZX-non-linear-rule}
\end{figure*}

\begin{thm}%
\label{thm:gen-zx}
For any diagrams $D_1$ and $D_2$ of the ZX-Calculus:
\[\interp{D_1}=\interp{D_2} \equi{} \left(\zxct+\add\right) \vdash D_1=D_2\]
%TODO : Completeness of the general ZX
\end{thm}

\subsubsection{ZX-Calculus with Parametrised Triangles}

This language, presented in Section~\ref{sec:param-triangle} shows an altered version of the ZX-Calculus, an additional generator with a parameter is introduced. It builds from the realisation that a particular diagram of the Clifford+T fragment proved so useful (and expressive) that it got a syntactic sugar: the so-called triangle. This node was introduced as a generator in~\cite{HNW} and later on in~\cite{zx-toffoli}. We propose here to make it even more expressive by giving it a parameter. This leads to a new complete axiomatisation for universal quantum computing.

\section{The ZW-Calculi}%
\label{sec:zw}

The proofs of completeness provided in the following heavily rely on the completeness of a third-party language called the ZW-Calculus. It was introduced as the GHZ/W calculus by Coecke and Kissinger in~\cite{ghz-w}, as a graphical language devoted to describing the interactions of the GHZ and W states, which constitute the only two classes of entanglement between three qubits, and has then been shown to be well suited for describing fermionic quantum computation.

It has been completed in~\cite{zw}, the diagrams allowing then to represent matrices over the integers $\mathbb{Z}$. We will denote this first version of the language \zw, but we will not directly use it as it is. Instead, we will use a slightly altered version, that we denote \zwh, where it can represent any matrix over dyadic rationals. The language has then been extended to allow the completeness for matrices over any ring in~\cite{HNW}, hence a fortiori, over $\mathbb{C}$. We will denote this last version of the language \zwc.

\subsection{For Integer Matrices}

We will present here the expanded version of the ZW-Calculus, i.e.~the version where all the black and white nodes have a degree $\leq 3$. To stay consistent with the previous definition of the ZX-Calculus, we will assume that the diagrams are to be read from top to bottom. The ZW-Calculus has the following finite set of generators:
\[T_e=\left\lbrace

%\InputIfFileExists{#1.tikz}{}{
\begin{tikzpicture}
	\begin{pgfonlayer}{nodelayer}
		\node [style={white dot}] (0) at (0, -0) {};
		\node [style=none] (1) at (0, 0.5) {};
		\node [style=none] (2) at (0, -0.5) {};
	\end{pgfonlayer}
	\begin{pgfonlayer}{edgelayer}
		\draw (1.center) to (2.center);
	\end{pgfonlayer}
\end{tikzpicture}%} % chktex 27
~,~
%\InputIfFileExists{#1.tikz}{}{
\input{./figures/Z-2-1.tikz}%} % chktex 27
~,~
%\InputIfFileExists{#1.tikz}{}{
\begin{tikzpicture}
	\begin{pgfonlayer}{nodelayer}
		\node [style=dot] (0) at (0, -0) {};
		\node [style=none] (1) at (0, 0.5) {};
		\node [style=none] (2) at (0, -0.5) {};
	\end{pgfonlayer}
	\begin{pgfonlayer}{edgelayer}
		\draw (1.center) to (2.center);
	\end{pgfonlayer}
\end{tikzpicture}%} % chktex 27
~,~
%\InputIfFileExists{#1.tikz}{}{
\begin{tikzpicture}
	\begin{pgfonlayer}{nodelayer}
		\node [style=dot] (0) at (0, 0) {};
		\node [style=none] (1) at (0, 0.5000001) {};
		\node [style=none] (2) at (-0.25, -0.5) {};
		\node [style=none] (3) at (0.25, -0.5) {};
	\end{pgfonlayer}
	\begin{pgfonlayer}{edgelayer}
		\draw (0) to (2.center);
		\draw (3.center) to (0);
		\draw (0) to (1.center);
	\end{pgfonlayer}
\end{tikzpicture}%} % chktex 27
~,\scalebox{1}{
%\InputIfFileExists{#1.tikz}{}{
\begin{tikzpicture}
	\begin{pgfonlayer}{nodelayer}
		\node [style=none] (0) at (0, 0.2499999) {};
		\node [style=none] (1) at (0, -0.2499999) {};
	\end{pgfonlayer}
	\begin{pgfonlayer}{edgelayer}
		\draw (0.center) to (1.center);
	\end{pgfonlayer}
\end{tikzpicture}%} % chktex 27
}~,\raisebox{-0.3em}{
%\InputIfFileExists{#1.tikz}{}{
\begin{tikzpicture}
	\begin{pgfonlayer}{nodelayer}
		\node [style=none] (0) at (-0.2500001, 0.2500001) {};
		\node [style=none] (1) at (0.2500001, 0.2500001) {};
	\end{pgfonlayer}
	\begin{pgfonlayer}{edgelayer}
		\draw [bend right=90, looseness=1.75] (0.center) to (1.center);
	\end{pgfonlayer}
\end{tikzpicture}%} % chktex 27
}~,\raisebox{-0.4em}{
%\InputIfFileExists{#1.tikz}{}{
\begin{tikzpicture}
	\begin{pgfonlayer}{nodelayer}
		\node [style=none] (0) at (-0.2500001, -0) {};
		\node [style=none] (1) at (0.2500001, -0) {};
	\end{pgfonlayer}
	\begin{pgfonlayer}{edgelayer}
		\draw [bend left=90, looseness=1.75] (0.center) to (1.center);
	\end{pgfonlayer}
\end{tikzpicture}%} % chktex 27
}~,~
%\InputIfFileExists{#1.tikz}{}{
\input{./figures/crossing.tikz}%} % chktex 27
~,~
%\InputIfFileExists{#1.tikz}{}{
\input{./figures/zw-cross.tikz}%} % chktex 27
~,~
%\InputIfFileExists{#1.tikz}{}{
\input{./figures/empty-diagram.tikz}%} % chktex 27
~
\right\rbrace\]
and diagrams are created thanks to the same two --- spacial and sequential --- compositions as ZX\@.

As for the ZX-Calculus, we define a standard interpretation, that associates to any diagram of the ZW-Calculus $D$ with $n$ inputs and $m$ outputs, a linear map $\interp{D}:\mathbb{Z}^{2^n}\to\mathbb{Z}^{2^m}$, inductively defined as:\\
\begin{minipage}{\columnwidth}
\titlerule{$\interp{.}$}
\[ \interp{D_1\otimes D_2}:=\interp{D_1}\otimes\interp{D_2} \qquad
\interp{D_2\circ D_1}:=\interp{D_2}\circ\interp{D_1}\]
\end{minipage}
\[\interp{
%\InputIfFileExists{#1.tikz}{}{
\input{./figures/empty-diagram.tikz}%} % chktex 27
~}:=\begin{pmatrix}1\end{pmatrix} \quad
\interp{~
%\InputIfFileExists{#1.tikz}{}{
\begin{tikzpicture}
	\begin{pgfonlayer}{nodelayer}
		\node [style=none] (0) at (0, 0.2499999) {};
		\node [style=none] (1) at (0, -0.2499999) {};
	\end{pgfonlayer}
	\begin{pgfonlayer}{edgelayer}
		\draw (0.center) to (1.center);
	\end{pgfonlayer}
\end{tikzpicture}%} % chktex 27
~~}:= \begin{pmatrix}
1 & 0 \\ 0 & 1\end{pmatrix}\qquad
\interp{\raisebox{-0.4em}{$
%\InputIfFileExists{#1.tikz}{}{
\begin{tikzpicture}
	\begin{pgfonlayer}{nodelayer}
		\node [style=none] (0) at (-0.2500001, -0) {};
		\node [style=none] (1) at (0.2500001, -0) {};
	\end{pgfonlayer}
	\begin{pgfonlayer}{edgelayer}
		\draw [bend left=90, looseness=1.75] (0.center) to (1.center);
	\end{pgfonlayer}
\end{tikzpicture}%} % chktex 27
$}}:= \begin{pmatrix}
1\\0\\0\\1
\end{pmatrix}\qquad
\interp{\raisebox{-0.3em}{$
%\InputIfFileExists{#1.tikz}{}{
\begin{tikzpicture}
	\begin{pgfonlayer}{nodelayer}
		\node [style=none] (0) at (-0.2500001, 0.2500001) {};
		\node [style=none] (1) at (0.2500001, 0.2500001) {};
	\end{pgfonlayer}
	\begin{pgfonlayer}{edgelayer}
		\draw [bend right=90, looseness=1.75] (0.center) to (1.center);
	\end{pgfonlayer}
\end{tikzpicture}%} % chktex 27
$}}:= \begin{pmatrix}
1&0&0&1
\end{pmatrix}\]
\[
\interp{
%\InputIfFileExists{#1.tikz}{}{
\input{./figures/crossing.tikz}%} % chktex 27
}:= \begin{pmatrix}
1&0&0&0\\
0&0&1&0\\
0&1&0&0\\
0&0&0&1
\end{pmatrix} \qquad
\interp{
%\InputIfFileExists{#1.tikz}{}{
\input{./figures/zw-cross.tikz}%} % chktex 27
}:= \begin{pmatrix}
1&0&0&0\\
0&0&1&0\\
0&1&0&0\\
0&0&0&-1
\end{pmatrix}\qquad
\interp{
%\InputIfFileExists{#1.tikz}{}{
\begin{tikzpicture}
	\begin{pgfonlayer}{nodelayer}
		\node [style=dot] (0) at (0, -0) {};
		\node [style=none] (1) at (0, 0.5) {};
		\node [style=none] (2) at (0, -0.5) {};
	\end{pgfonlayer}
	\begin{pgfonlayer}{edgelayer}
		\draw (1.center) to (2.center);
	\end{pgfonlayer}
\end{tikzpicture}%} % chktex 27
}:= \begin{pmatrix}
0&1\\1&0
\end{pmatrix}\]
\begin{minipage}{\columnwidth}
\[
\interp{
%\InputIfFileExists{#1.tikz}{}{
\begin{tikzpicture}
	\begin{pgfonlayer}{nodelayer}
		\node [style=dot] (0) at (0, 0) {};
		\node [style=none] (1) at (0, 0.5000001) {};
		\node [style=none] (2) at (-0.25, -0.5) {};
		\node [style=none] (3) at (0.25, -0.5) {};
	\end{pgfonlayer}
	\begin{pgfonlayer}{edgelayer}
		\draw (0) to (2.center);
		\draw (3.center) to (0);
		\draw (0) to (1.center);
	\end{pgfonlayer}
\end{tikzpicture}%} % chktex 27
}:= \begin{pmatrix}
0&1\\1&0\\1&0\\0&0
\end{pmatrix}\qquad
\interp{
%\InputIfFileExists{#1.tikz}{}{
\begin{tikzpicture}
	\begin{pgfonlayer}{nodelayer}
		\node [style={white dot}] (0) at (0, -0) {};
		\node [style=none] (1) at (0, 0.5) {};
		\node [style=none] (2) at (0, -0.5) {};
	\end{pgfonlayer}
	\begin{pgfonlayer}{edgelayer}
		\draw (1.center) to (2.center);
	\end{pgfonlayer}
\end{tikzpicture}%} % chktex 27
}:= \begin{pmatrix}
1&0\\0&-1
\end{pmatrix} \qquad
\interp{
%\InputIfFileExists{#1.tikz}{}{
\input{./figures/Z-2-1.tikz}%} % chktex 27
}:= \begin{pmatrix}
1&0&0&0\\0&0&0&-1
\end{pmatrix}
\]
\rule{\columnwidth}{0.5pt}
\end{minipage}
\vspace{0.2em}

\noindent This map is obviously different from the one of the ZX-Calculus --- the domain is different --- but we will use the same notation.

%\begin{rem}
%The symbols used for the generators have be altered from the original ZW-Calculus in order to make it more compatible with the ZX-Calculus.
%\end{rem}

\begin{lemC}[\cite{zw}]%
\label{lem:zw-universal}
ZW-Diagrams are universal for matrices of $\mathbb{Z}^{2^n}\times\mathbb{Z}^{2^m}$:
\[\forall A\in \mathbb{Z}^{2^n}\times\mathbb{Z}^{2^m},~~\exists D:n\to m,~~ \interp{D}=A\]
\end{lemC}

%\subsection{Calculus}

%\newcommand{\tikzfigcd}[1]{\tikzfig{#1\iftoggle{braids}{}{-no-braid}}}
\begin{figure*}[!htb]
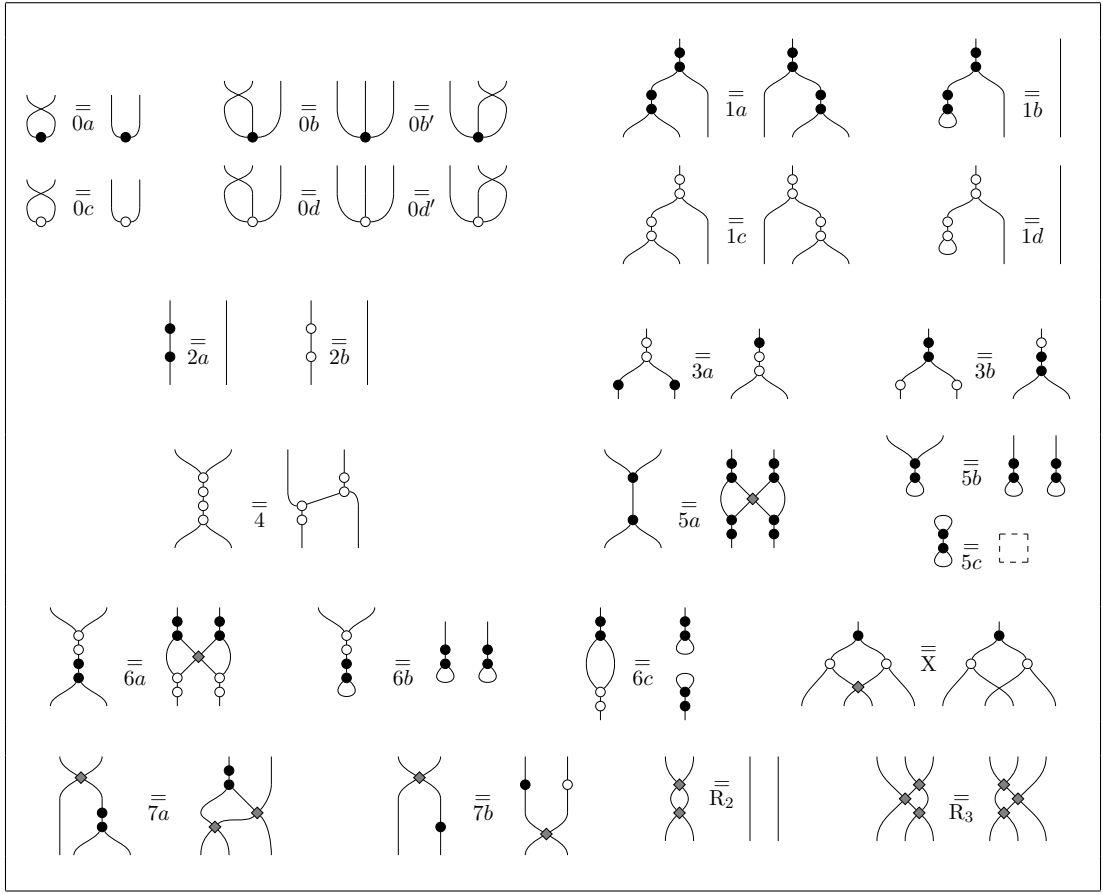

\def\scale{0.75}
\centering
\begin{tabular}{|ccc|}
\hline
&&\\
\scalebox{\scale}{
%\InputIfFileExists{#1.tikz}{}{
\input{./figures/ZW-rule-0-no-param.tikz}%} % chktex 27
} & $\quad$ & \scalebox{\scale}{
%\InputIfFileExists{#1.tikz}{}{
\input{./figures/ZW-rule-1-no-param.tikz}%} % chktex 27
} \\
&&\\
\scalebox{\scale}{
%\InputIfFileExists{#1.tikz}{}{
\input{./figures/ZW-rule-2-no-param.tikz}%} % chktex 27
} && \scalebox{\scale}{
%\InputIfFileExists{#1.tikz}{}{
\input{./figures/ZW-rule-3-no-param.tikz}%} % chktex 27
} \\
&&\\
 \scalebox{\scale}{
%\InputIfFileExists{#1.tikz}{}{
\input{./figures/ZW-rule-4-no-param.tikz}%} % chktex 27
} && \scalebox{\scale}{
%\InputIfFileExists{#1.tikz}{}{
\input{./figures/ZW-rule-5-no-braid.tikz}%} % chktex 27
} \\
 &&\\
 \multicolumn{3}{|c|}{\scalebox{\scale}{
%\InputIfFileExists{#1.tikz}{}{
\input{./figures/ZW-rule-6-no-param.tikz}%} % chktex 27
} $\qquad$ \scalebox{\scale}{
%\InputIfFileExists{#1.tikz}{}{
\input{./figures/ZW-rule-X-no-braid.tikz}%} % chktex 27
}}\\
 &&\\
 \multicolumn{3}{|c|}{
 \scalebox{\scale}{
%\InputIfFileExists{#1.tikz}{}{
\input{./figures/ZW-rule-7-no-param.tikz}%} % chktex 27
} $\qquad$ \scalebox{\scale}{
%\InputIfFileExists{#1.tikz}{}{
\input{./figures/reidmeister-3.tikz}%} % chktex 27
}} \\
 &&\\
\hline
\end{tabular}
\caption{Set of rules for the ZW-Calculus.}%
 \label{fig:ZW_rules}
\end{figure*}

The ZW-Calculus comes with a \emph{complete} set of rules ZW that is given in Figure~\ref{fig:ZW_rules}.
Here again, the paradigm \emph{Only Topology Matters} applies except for 
%\InputIfFileExists{#1.tikz}{}{
\input{./figures/zw-cross.tikz}%} % chktex 27
. It gives sense to nodes that are not directly given in $T_e$, e.g.:
\[
%\InputIfFileExists{#1.tikz}{}{
\input{./figures/zw-bent-wire-example.tikz}%} % chktex 27
\]
For 
%\InputIfFileExists{#1.tikz}{}{
\input{./figures/zw-cross.tikz}%} % chktex 27
, the order of inputs and outputs is important. However, it is invariant under cyclic permutations, as suggested by the shape of the node: 
%\InputIfFileExists{#1.tikz}{}{
\input{./figures/zw-cross-cycle.tikz}%} % chktex 27
.

All these rules are sound. We use the same notation $\vdash$ as defined in Section~\ref{sec:zx}, and we can still apply the rewrite rules to subdiagrams. %still have:
%\[ (\zw\vdash D_1=D_2)\implies \left\lbrace\begin{array}{ccc}
%(\zw\vdash D_1\circ D = D_2\circ D) & \land & (\zw\vdash D\circ D_1 = D\circ D_2)\\
%(\zw\vdash D_1\otimes D = D_2\otimes D) & \land & (\zw\vdash D\otimes D_1 = D\otimes D_2)
%\end{array}\right. \]
In the following we may use the shortcuts:
\begin{center}

%\InputIfFileExists{#1.tikz}{}{
\input{./figures/black-dot-0-1.tikz}%} % chktex 27
 and\hspace{1em} 
%\InputIfFileExists{#1.tikz}{}{
\input{./figures/white-dot-0-1.tikz}%} % chktex 27
 % chktex 8
\end{center}

The main interest of this language is that it  has been proved to be complete.
\begin{thmC}[\cite{zw}]
For any two \zw-diagrams $D_1$ and $D_2$:
\[\interp{D_1}=\interp{D_2} \iff \zw\vdash D_1=D_2\]
\end{thmC}

\subsection{For Dyadic Matrices}

We define an extension of the ZW-Calculus by adding a new node that represents the $1\times1$ matrix $\begin{pmatrix}\frac{1}{2}\end{pmatrix}$ and binding it to the calculus with an additional rule.

\begin{defi}
We define the \zwh-Calculus as the extension of the ZW-Calculus such as:
\[\left\lbrace
\begin{array}{l}
 T_{1/2}=T_e\cup\{\half\}\\
\zwh=\zw\cup\left\lbrace
%\InputIfFileExists{#1.tikz}{}{
\input{./figures/additional-ZW-rule.tikz}%} % chktex 27
~\right\rbrace
\end{array}
\right.\]
The standard interpretation of a diagram $D:n\to m$ is now a matrix
$\interp{D}:\mathbb{D}^{2^n}\to\mathbb{D}^{2^m}$ over the
ring $\mathbb{D} = \mathbb{Z}[1/2]$ of dyadic rationals and is given by
the standard interpretation of the ZW-Calculus extended with $\interp{\half}:=\begin{pmatrix}\frac{1}{2}\end{pmatrix}$.
\end{defi}

\begin{prop}%
  \label{prop:zwcomplete}
  The \zwh is sound and complete:
  For two diagrams $D_1, D_2$ of the \zwh-Calculus,
  \[\interp{D_1} = \interp{D_2}\iff\zwh \vdash D_1 = D_2\]
\end{prop}
\begin{proof}
  Soundness is obvious.

Now let $D_1$ and $D_2$ be two diagrams of the \zwh-Calculus such that $\interp{D_1}=\interp{D_2}$.
We can rewrite $D_1$ and $D_2$ as $D_i = d_i\otimes(\half)^{\otimes   n_i}$ for some integers $n_i$ and diagrams $d_i$ of the ZW-Calculus that do not use the $\half$ symbol.

From the new introduced rule, we get that  $\zwh\vdash d_i=D_i\otimes \left(\two\right)^{\otimes n_i}$. W.l.o.g.\ assume $n_1\leq n_2$. Then $\interp{d_1\otimes \left(\two\right)^{\otimes n_2-n_1}}=2^{n_2-n_1}\interp{d_1}=2^{n_2}\interp{D_1} = \interp{d_2}$. Since $d_1$ and $d_2$ are ZW-diagrams and have the same interpretation, thanks to the completeness of the ZW-Calculus, $\zwh\vdash d_1\otimes \left(\two\right)^{\otimes n_2-n_1}=d_2$, which means $\zwh\vdash D_1=D_2$ by applying $n_2$ times the new rule on both sides of the equality.
\end{proof}

We can also precisely characterise the expressive power of this new extension.

\begin{prop}%
\label{prop:zwh-universal}
\zwh-Diagrams are universal for matrices of $\mathbb{D}^{2^n}\times\mathbb{D}^{2^m}$:
\[\forall A\in \mathbb{D}^{2^n}\times\mathbb{D}^{2^m},~~\exists D\in \zwh,~~ \interp{D}=A\]
\end{prop}

\begin{proof}
Let $A\in \mathbb{D}^{2^n}\times\mathbb{D}^{2^m}$. There exists $n\in\mathbb{N}$ and $A'\in \mathbb{Z}^{2^n}\times\mathbb{Z}^{2^m}$ such that $A=\frac{1}{2^n}A'$. By Lemma~\ref{lem:zw-universal}, there exists a ZW-diagram $D_{A'}$ such that $\interp{D_{A'}}=A'$. $D_{A'}$ is also a \zwh-diagram. Hence, we can define $D_A:=D_{A'}\otimes(\half)^{\otimes n}$ and $\interp{D_A} = \interp{\half}^n\interp{D_{A'}} = \frac{1}{2^n}A'=A$.
\end{proof}

\subsection{For Complex Matrices}

The \zwc-Calculus~\cite{HNW} differs from the initial language in that the white nodes now bear a parameter: 
%\InputIfFileExists{#1.tikz}{}{
\input{./figures/Z-param.tikz}%} % chktex 27
 where $r$ is a complex number:
\[\interp{
%\InputIfFileExists{#1.tikz}{}{
\input{./figures/Z-param.tikz}%} % chktex 27
} := \annoted{2^m}{2^n}{\begin{pmatrix}
  1 & 0 & \cdots & 0 & 0 \\
  0 & 0 & \cdots & 0 & 0 \\
  \vdots & \vdots & \ddots & \vdots & \vdots \\
  0 & 0 & \cdots & 0 & 0 \\
  0 & 0 & \cdots & 0 & r
 \end{pmatrix}}\]
It is important to notice that the white node in the initial ZW-Calculus would now be represented by a white node with parameter $-1$. This more expressive language comes equipped with a \emph{complete} set of axioms, some of which can already be found in the presentation of the ZW-Calculus. The axiomatisation is given in Figure~\ref{fig:ZWC_rules}.

\begin{thmC}[\cite{HNW}]
For any two \zwc-diagrams $D_1$ and $D_2$:
\[\interp{D_1}=\interp{D_2} \iff \zwc\vdash D_1=D_2\]
\end{thmC}

\begin{figure*}[!htb]
\def\scale{0.75}
\centering
\begin{tabular}{|ccc|}
\hline
&&\\
\scalebox{\scale}{
%\InputIfFileExists{#1.tikz}{}{
\input{./figures/ZW-rule-0-no-braid.tikz}%} % chktex 27
} & $\quad$ & \scalebox{\scale}{
%\InputIfFileExists{#1.tikz}{}{
\input{./figures/ZW-rule-1.tikz}%} % chktex 27
} \\
&&\\
\scalebox{\scale}{
%\InputIfFileExists{#1.tikz}{}{
\input{./figures/ZW-rule-2-no-braid.tikz}%} % chktex 27
} && \scalebox{\scale}{
%\InputIfFileExists{#1.tikz}{}{
\input{./figures/ZW-rule-3.tikz}%} % chktex 27
} \\
&&\\
 \scalebox{\scale}{
%\InputIfFileExists{#1.tikz}{}{
\input{./figures/ZW-rule-4.tikz}%} % chktex 27
} && \scalebox{\scale}{
%\InputIfFileExists{#1.tikz}{}{
\input{./figures/ZW-rule-5-no-braid.tikz}%} % chktex 27
} \\
 &&\\
 \multicolumn{3}{|c|}{\scalebox{\scale}{
%\InputIfFileExists{#1.tikz}{}{
\input{./figures/ZW-rule-6-no-braid.tikz}%} % chktex 27
} $\qquad$ \scalebox{\scale}{
%\InputIfFileExists{#1.tikz}{}{
\input{./figures/ZW-rule-X-no-braid.tikz}%} % chktex 27
}}\\
 &&\\
 \multicolumn{3}{|c|}{
 \scalebox{\scale}{
%\InputIfFileExists{#1.tikz}{}{
\input{./figures/ZW-rule-7-no-braid.tikz}%} % chktex 27
} $\qquad$ \scalebox{\scale}{
%\InputIfFileExists{#1.tikz}{}{
\input{./figures/reidmeister-3.tikz}%} % chktex 27
}} \\
 &&\\
\hline
\end{tabular}
\caption{Set of rules for the \zwc-Calculus.}%
 \label{fig:ZWC_rules}
\end{figure*}

\section{Completeness for Clifford+T ZX-Calculus}%
\label{sec:cliff-t}

This section is devoted to proving Theorem~\ref{thm:cliff-T}, that is, the completeness of the ZX-Calculus for the \frag4 when equipped with the rules $\zxct$. The proof method is the following:
\begin{itemize}
\item Provide an interpretation $[.]_W$ from the \frag4 of the ZX-Calculus to the \zwh-Calculus, where the semantics undergoes a well-defined homomorphism with left inverse.
\item Provide an interpretation $[.]_X$ from the \zwh-Calculus to the \frag4 of the ZX-Calculus that respects the semantics.
\item Prove that all the axioms of the \zwh-Calculus can be derived using $\zxct$ after application of $[.]_X$.
\item Prove that any diagram $D$ of the \frag4 can be recovered from $\left[[D]_W\right]_X$.
\end{itemize}

\subsection{Encoding the Clifford+T ZX-Calculus into ZW\texorpdfstring{\!$\mathbf{_{\frac{1}{2}}}$}{(1/2)}}

Our goal here is to send ZX-diagrams in the \frag4 into the \zwh-Calculus. The main obstacle is that the former represent matrices over $\mathbb{D}\left[\piq{}\right]$, while the latter only represent matrices over $\mathbb{D}$. %So, we aim at encoding the

\subsubsection{From ${\mathbb{Q}[\piq{}]}$ to ${\mathbb{Q}}$}
All results used in the next two sections are standard in field
theory, see e.g.~\cite{roman}.
Let $R \subseteq \mathbb{C}$ be a (commutative) ring and $\alpha \in \mathbb{C}$.
By $R[\alpha]$ we denote the smallest subring of $\mathbb{C}$ that contains both $R$ and $\alpha$.

Of primary importance will be the ring $\mathbb{Q}[\piq{}]$, as all terms of the $\pi/4$ fragment of the ZX-Calculus have interpretations as matrices in this ring.
This is clear for all terms except possibly for $\sqrt{2}$, but $\sqrt{2} = \piq{} - \piq{3}$.

If $\alpha$ is algebraic, it is well known that $\mathbb{Q}[\alpha]$ is a field. When $F \subseteq F'$ are two fields, $F'$ can be seen as a vector space (actually an algebra) over $F$. Its dimension is denoted $[F':F]$ and we say that $F'$ is an extension of $F$ of degree $[F':F]$.
In the specific case of $\mathbb{Q}[\alpha]$, its dimension over $\mathbb{Q}$ is exactly the degree of the minimal polynomial over $\mathbb{Q}$ of $\alpha$.
Notice that the minimal polynomial of a $n-$th primitive root of the
unity is $\phi(n)$ where $\phi$ is Euler's totient function.

In our case, $\piq{}$ is a eighth primitive root of the unity, so that $\mathbb{Q}[\piq{}]$ is a vector space of dimension $4$, one basis being given by $1,\piq{},\piq2, \piq3$.
In particular:
\begin{prop}
  Every element of $\mathbb{Q}[\piq{}]$ can be written in a unique way
 $a + b \piq{} + c \piq{2} + d \piq{3}$ for some rationals numbers $a,b,c,d$.
\end{prop}

For $x \in \mathbb{Q}[\piq{}]$, let $\psi(x)$ be the function defined by  $\psi(x) = y \mapsto x y$. For each $x$, $\psi(x)$ is a linear map and therefore can be given by a $4\times 4$ matrix in the basis $\{\piq3,\piq2, \piq{}, 1\}$.
$\psi(1)$ is of course the identity matrix and
\[
\psi(\piq{}) = M = \begin{pmatrix}
  0 & 1 & 0 & 0 \\
  0 & 0 & 1 & 0 \\
  0 & 0 & 0 & 1 \\
  -1 & 0 & 0 & 0 \\
  \end{pmatrix}
\]
Notice that $M^t$ is the companion matrix of the polynomial $X^4+1$ which characterises $\piq{}$ as an algebraic number.

\begin{prop}%
  \label{prop:homomorphism}
  The map:
  \[  \psi: a+b\piq{} +c\piq2 + d\piq3 \mapsto a I_4 + b M + c M^2 +d M ^3\]
  is a  homomorphism of $\mathbb{Q}$-algebras from  $\mathbb{Q}[\piq{}]$ to $M_4(\mathbb{Q})$
\end{prop}
This homomorphism has a left-inverse. Indeed, let
\[
\omega = \begin{pmatrix}
  1 \\
  \piq{} \\
  \piq2 \\
  \piq3 \\
\end{pmatrix}
\]
Then $\psi(x)\omega = x\omega$ for all $x\in \mathbb{Q}[\piq{}]$, so the left inverse of $\psi$ is given by $x\mapsto x\circ\omega$.

With this morphism, we can see elements of $\mathbb{Q}[\piq{}]$ as matrices over $\mathbb{Q}$.

Of course we can do the same with matrices over $\mathbb{Q}[\piq{}]$.

\begin{defi}
  Define:
  \[  \psi: A+B\piq{} +C\piq2 + D\piq3 \mapsto A\otimes  I_4 + B\otimes M + C\otimes M^2 +D\otimes M ^3\]
  $\psi$ is injective and maps a matrix over $\mathbb{Q}[\piq{}]$ of dimension $n\times m$ to a matrix over $\mathbb{Q}$ of dimension $4n\times 4m$.
\end{defi}
We use the same notation  $\psi$ as before, as the definitions are equivalent for one-by-one matrices (i.e.\ scalars).

It is easy to see that Proposition~\ref{prop:homomorphism} holds for the extended $\psi$ in the sense that $\psi(qA) = q\psi(A)$ for $q$ rational, $\psi(A+B) = \psi(A) + \psi(B)$, $\psi(AB) = \psi(A)\psi(B)$ whenever this makes sense.

Notice however that $\psi(A\otimes B)$ is not $\psi(A) \otimes \psi(B)$.

As before, $\psi$ has a left inverse, as evidenced by:
\begin{prop}%
\label{prop:left-inverse}
  For all matrices $X$ of dimension $n \times m$, $ \psi(X)(I_m \otimes \omega) = X \otimes \omega$
\end{prop}

While it is true that all coefficients of the standard interpretation of the $\pi/4$ fragment are in $\mathbb{Q}[\piq{}]$, we can be more precise.

Let $\mathbb{D} = \mathbb{Z}[1/2]$ be the set of all dyadic rational numbers, i.e.\ rational numbers of the form $p/2^n$.

It is easy to see that any element of $\mathbb{D}[\piq{}]$ can be written in a unique way
$a + b \piq{} + c \piq2 + d \piq3$ for some dyadic rational numbers $a,b,c,d$. (It is NOT a consequence of the similar statement for $\mathbb{Q}$. We have to use here the additional property that $\piq{}$ is not only an algebraic number, but also an algebraic \emph{integer}).

Then it is clear that actually all coefficients of the $\pi/4$ fragment of the ZX-Calculus are in $\mathbb{D}[\piq{}]$.
As $\mathbb{D} \subset \mathbb{Q}$ all we said before still holds, and we actually obtain with $\psi$ a map from matrices over $\mathbb{D}[\piq{}]$ to matrices over $\mathbb{D}$.

We have now provided a way to see a matrix over $\mathbb{D}[\piq{}]$ as a matrix over $\mathbb{D}$. This paves the way to an interpretation from $\zxct$ to $\zwh$.

\begin{rem}
After all the hassle of the current section, and of the one to come, one can legitimately ask: Why not directly use a version of the ZW-Calculus for the ring $\mathbb{D}[\piq{}]$ ($\zw_{\mathbb{D}[\piq{}]}$), since the ZW-Calculus can be adapted to any ring? Indeed, we could have used this version of the ZW-Calculus. Basically, doing so, the interpretation $[.]_W$ becomes easier, while $[.]_X$ becomes a bit less straightforward.
However, our aim is to provide an axiomatisation that is the simplest possible for the Clifford+T ZX-Calculus. Hence, since the hard part is really to prove that all the axioms of the \zwh are derivable in ZX, it is natural to begin with the simplest ZW axiomatisation possible, hence \zwh and not $\zw_{\mathbb{D}[\piq{}]}$.
\end{rem}

\subsubsection{The Interpretation ${[.]_W}$}

As announced, the interpretation $[.]_W$ maps any ZX-diagram of the \frag4 to a \zwh-diagram, while the semantics undergoes the homomorphism $\psi$ defined above. It is recursively defined as:\\
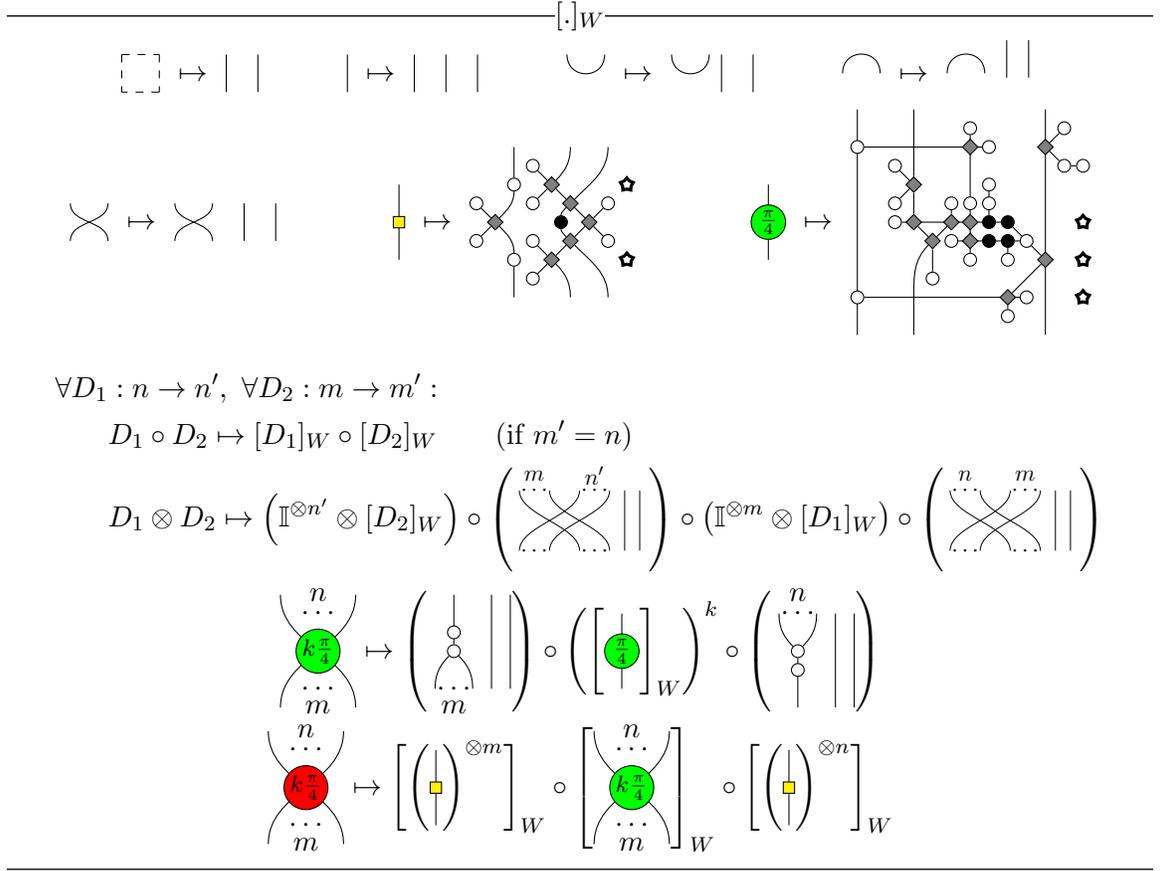
\begin{minipage}{\columnwidth}
\titlerule{$[.]_W$}
%\begin{multicols}{3}
\[ 
%\InputIfFileExists{#1.tikz}{}{
\input{./figures/empty-diagram.tikz}%} % chktex 27
~~\mapsto~~ 
%\InputIfFileExists{#1.tikz}{}{
\begin{tikzpicture}
	\begin{pgfonlayer}{nodelayer}
		\node [style=none] (0) at (0, 0.2499999) {};
		\node [style=none] (1) at (0, -0.2499999) {};
	\end{pgfonlayer}
	\begin{pgfonlayer}{edgelayer}
		\draw (0.center) to (1.center);
	\end{pgfonlayer}
\end{tikzpicture}%} % chktex 27
\quad
%\InputIfFileExists{#1.tikz}{}{
\begin{tikzpicture}
	\begin{pgfonlayer}{nodelayer}
		\node [style=none] (0) at (0, 0.2499999) {};
		\node [style=none] (1) at (0, -0.2499999) {};
	\end{pgfonlayer}
	\begin{pgfonlayer}{edgelayer}
		\draw (0.center) to (1.center);
	\end{pgfonlayer}
\end{tikzpicture}%} % chktex 27
 \qquad\quad
 
%\InputIfFileExists{#1.tikz}{}{
\begin{tikzpicture}
	\begin{pgfonlayer}{nodelayer}
		\node [style=none] (0) at (0, 0.2499999) {};
		\node [style=none] (1) at (0, -0.2499999) {};
	\end{pgfonlayer}
	\begin{pgfonlayer}{edgelayer}
		\draw (0.center) to (1.center);
	\end{pgfonlayer}
\end{tikzpicture}%} % chktex 27
~~\mapsto~~ 
%\InputIfFileExists{#1.tikz}{}{
\begin{tikzpicture}
	\begin{pgfonlayer}{nodelayer}
		\node [style=none] (0) at (0, 0.2499999) {};
		\node [style=none] (1) at (0, -0.2499999) {};
	\end{pgfonlayer}
	\begin{pgfonlayer}{edgelayer}
		\draw (0.center) to (1.center);
	\end{pgfonlayer}
\end{tikzpicture}%} % chktex 27
\quad
%\InputIfFileExists{#1.tikz}{}{
\begin{tikzpicture}
	\begin{pgfonlayer}{nodelayer}
		\node [style=none] (0) at (0, 0.2499999) {};
		\node [style=none] (1) at (0, -0.2499999) {};
	\end{pgfonlayer}
	\begin{pgfonlayer}{edgelayer}
		\draw (0.center) to (1.center);
	\end{pgfonlayer}
\end{tikzpicture}%} % chktex 27
\quad
%\InputIfFileExists{#1.tikz}{}{
\begin{tikzpicture}
	\begin{pgfonlayer}{nodelayer}
		\node [style=none] (0) at (0, 0.2499999) {};
		\node [style=none] (1) at (0, -0.2499999) {};
	\end{pgfonlayer}
	\begin{pgfonlayer}{edgelayer}
		\draw (0.center) to (1.center);
	\end{pgfonlayer}
\end{tikzpicture}%} % chktex 27
 \qquad\quad
 
%\InputIfFileExists{#1.tikz}{}{
\begin{tikzpicture}
	\begin{pgfonlayer}{nodelayer}
		\node [style=none] (0) at (-0.2500001, 0.2500001) {};
		\node [style=none] (1) at (0.2500001, 0.2500001) {};
	\end{pgfonlayer}
	\begin{pgfonlayer}{edgelayer}
		\draw [bend right=90, looseness=1.75] (0.center) to (1.center);
	\end{pgfonlayer}
\end{tikzpicture}%} % chktex 27
~~\mapsto~~ 
%\InputIfFileExists{#1.tikz}{}{
\begin{tikzpicture}
	\begin{pgfonlayer}{nodelayer}
		\node [style=none] (0) at (-0.2500001, 0.2500001) {};
		\node [style=none] (1) at (0.2500001, 0.2500001) {};
	\end{pgfonlayer}
	\begin{pgfonlayer}{edgelayer}
		\draw [bend right=90, looseness=1.75] (0.center) to (1.center);
	\end{pgfonlayer}
\end{tikzpicture}%} % chktex 27
~~
%\InputIfFileExists{#1.tikz}{}{
\begin{tikzpicture}
	\begin{pgfonlayer}{nodelayer}
		\node [style=none] (0) at (0, 0.2499999) {};
		\node [style=none] (1) at (0, -0.2499999) {};
	\end{pgfonlayer}
	\begin{pgfonlayer}{edgelayer}
		\draw (0.center) to (1.center);
	\end{pgfonlayer}
\end{tikzpicture}%} % chktex 27
\quad
%\InputIfFileExists{#1.tikz}{}{
\begin{tikzpicture}
	\begin{pgfonlayer}{nodelayer}
		\node [style=none] (0) at (0, 0.2499999) {};
		\node [style=none] (1) at (0, -0.2499999) {};
	\end{pgfonlayer}
	\begin{pgfonlayer}{edgelayer}
		\draw (0.center) to (1.center);
	\end{pgfonlayer}
\end{tikzpicture}%} % chktex 27
\qquad\quad
 
%\InputIfFileExists{#1.tikz}{}{
\begin{tikzpicture}
	\begin{pgfonlayer}{nodelayer}
		\node [style=none] (0) at (-0.2500001, -0) {};
		\node [style=none] (1) at (0.2500001, -0) {};
	\end{pgfonlayer}
	\begin{pgfonlayer}{edgelayer}
		\draw [bend left=90, looseness=1.75] (0.center) to (1.center);
	\end{pgfonlayer}
\end{tikzpicture}%} % chktex 27
~~\mapsto~~ 
%\InputIfFileExists{#1.tikz}{}{
\begin{tikzpicture}
	\begin{pgfonlayer}{nodelayer}
		\node [style=none] (0) at (-0.2500001, -0) {};
		\node [style=none] (1) at (0.2500001, -0) {};
	\end{pgfonlayer}
	\begin{pgfonlayer}{edgelayer}
		\draw [bend left=90, looseness=1.75] (0.center) to (1.center);
	\end{pgfonlayer}
\end{tikzpicture}%} % chktex 27
~~\raisebox{0.5em}{
%\InputIfFileExists{#1.tikz}{}{
\begin{tikzpicture}
	\begin{pgfonlayer}{nodelayer}
		\node [style=none] (0) at (0, 0.2499999) {};
		\node [style=none] (1) at (0, -0.2499999) {};
	\end{pgfonlayer}
	\begin{pgfonlayer}{edgelayer}
		\draw (0.center) to (1.center);
	\end{pgfonlayer}
\end{tikzpicture}%} % chktex 27

%\InputIfFileExists{#1.tikz}{}{
\begin{tikzpicture}
	\begin{pgfonlayer}{nodelayer}
		\node [style=none] (0) at (0, 0.2499999) {};
		\node [style=none] (1) at (0, -0.2499999) {};
	\end{pgfonlayer}
	\begin{pgfonlayer}{edgelayer}
		\draw (0.center) to (1.center);
	\end{pgfonlayer}
\end{tikzpicture}%} % chktex 27
} \]
\end{minipage}
\[
%\InputIfFileExists{#1.tikz}{}{
\input{./figures/crossing.tikz}%} % chktex 27
~~\mapsto~~ 
%\InputIfFileExists{#1.tikz}{}{
\input{./figures/crossing.tikz}%} % chktex 27
\quad
%\InputIfFileExists{#1.tikz}{}{
\begin{tikzpicture}
	\begin{pgfonlayer}{nodelayer}
		\node [style=none] (0) at (0, 0.2499999) {};
		\node [style=none] (1) at (0, -0.2499999) {};
	\end{pgfonlayer}
	\begin{pgfonlayer}{edgelayer}
		\draw (0.center) to (1.center);
	\end{pgfonlayer}
\end{tikzpicture}%} % chktex 27
\quad
%\InputIfFileExists{#1.tikz}{}{
\begin{tikzpicture}
	\begin{pgfonlayer}{nodelayer}
		\node [style=none] (0) at (0, 0.2499999) {};
		\node [style=none] (1) at (0, -0.2499999) {};
	\end{pgfonlayer}
	\begin{pgfonlayer}{edgelayer}
		\draw (0.center) to (1.center);
	\end{pgfonlayer}
\end{tikzpicture}%} % chktex 27
\qquad\qquad

%\InputIfFileExists{#1.tikz}{}{
\begin{tikzpicture}
	\begin{pgfonlayer}{nodelayer}
		\node [style={H box}] (0) at (0, 0) {};
		\node [style=none] (1) at (0, 0.5) {};
		\node [style=none] (2) at (0, -0.5) {};
	\end{pgfonlayer}
	\begin{pgfonlayer}{edgelayer}
		\draw (2.center) to (1.center);
	\end{pgfonlayer}
\end{tikzpicture}
%} % chktex 27
~~\mapsto~~
%\InputIfFileExists{#1.tikz}{}{
\input{./figures/hadamard-interpretation-2-no-braid.tikz}%} % chktex 27
\qquad\qquad

%\InputIfFileExists{#1.tikz}{}{
\begin{tikzpicture}
	\begin{pgfonlayer}{nodelayer}
		\node [style=gn] (0) at (0, -0) {$\frac{\pi}{4}$};
		\node [style=none] (1) at (0, 0.5000001) {};
		\node [style=none] (2) at (0, -0.5000001) {};
	\end{pgfonlayer}
	\begin{pgfonlayer}{edgelayer}
		\draw (1.center) to (2.center);
	\end{pgfonlayer}
\end{tikzpicture}%} % chktex 27
~~\mapsto~~
%\InputIfFileExists{#1.tikz}{}{
\input{./figures/gn-pi_4-interpretation-2-no-braid.tikz}%} % chktex 27
\]
%\end{multicols}

\begin{align*}
\forall D_1&:n\to n',~\forall D_2:m\to m':\\
& D_1\circ D_2 \mapsto [D_1]_W\circ [D_2]_W\qquad(\text{if } m'=n)\\
& D_1\otimes D_2 \mapsto\left( \mathbb{I}^{\otimes n'}\otimes [D_2]_W\right)\circ\left(\nmcrossI{m}{n'}\right)\circ \left(\mathbb{I}^{\otimes m}\otimes [D_1]_W\right)\circ\left(\nmcrossI{n}{m}\right)
\end{align*}

\[
%\InputIfFileExists{#1.tikz}{}{
\input{./figures/gn-kpi_4.tikz}%} % chktex 27
\mapsto \left(
%\InputIfFileExists{#1.tikz}{}{
\input{./figures/gn-0-1-m-interpretation.tikz}%} % chktex 27
\right)\circ\left(\left[
%\InputIfFileExists{#1.tikz}{}{
\begin{tikzpicture}
	\begin{pgfonlayer}{nodelayer}
		\node [style=gn] (0) at (0, -0) {$\frac{\pi}{4}$};
		\node [style=none] (1) at (0, 0.5000001) {};
		\node [style=none] (2) at (0, -0.5000001) {};
	\end{pgfonlayer}
	\begin{pgfonlayer}{edgelayer}
		\draw (1.center) to (2.center);
	\end{pgfonlayer}
\end{tikzpicture}%} % chktex 27
\right]_W\right)^k\circ \left(
%\InputIfFileExists{#1.tikz}{}{
\input{./figures/gn-0-n-1-interpretation.tikz}%} % chktex 27
\right)\]
\begin{minipage}{\columnwidth}
\[
%\InputIfFileExists{#1.tikz}{}{
\input{./figures/rn-kpi_4.tikz}%} % chktex 27
\mapsto \left[\left(
%\InputIfFileExists{#1.tikz}{}{
\begin{tikzpicture}
	\begin{pgfonlayer}{nodelayer}
		\node [style={H box}] (0) at (0, 0) {};
		\node [style=none] (1) at (0, 0.5) {};
		\node [style=none] (2) at (0, -0.5) {};
	\end{pgfonlayer}
	\begin{pgfonlayer}{edgelayer}
		\draw (2.center) to (1.center);
	\end{pgfonlayer}
\end{tikzpicture}
%} % chktex 27
\right)^{\otimes m}\right]_W\circ \left[
%\InputIfFileExists{#1.tikz}{}{
\input{./figures/gn-kpi_4.tikz}%} % chktex 27
\right]_W \circ \left[\left(
%\InputIfFileExists{#1.tikz}{}{
\begin{tikzpicture}
	\begin{pgfonlayer}{nodelayer}
		\node [style={H box}] (0) at (0, 0) {};
		\node [style=none] (1) at (0, 0.5) {};
		\node [style=none] (2) at (0, -0.5) {};
	\end{pgfonlayer}
	\begin{pgfonlayer}{edgelayer}
		\draw (2.center) to (1.center);
	\end{pgfonlayer}
\end{tikzpicture}
%} % chktex 27
\right)^{\otimes n}\right]_W\]
\vspace{0.2em}
\end{minipage}
\rule{\columnwidth}{0.5pt}

\noindent The interpretation is made so that the two wires on the right act as ``control wires''. They are the ones that bear the encoding $\psi$, and as Proposition~\ref{prop:left-inverse-ZX}, they are the ones the ``decoder'' will be applied to. In particular, the tensor product can be understood graphically as:
%The interpretation of the spacial composition $\otimes$ might seem a tad cryptical. It is in fact a way of putting ``side-by-side'' the interpretations of $D_1$ and $D_2$, while at the same time making them share the two wires on the right. We can see it as:
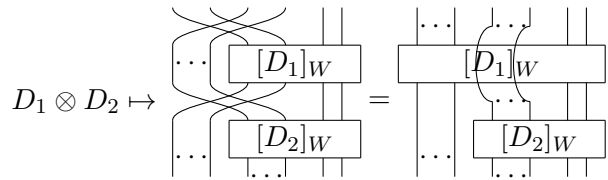
\[ D_1\otimes D_2 \mapsto 
%\InputIfFileExists{#1.tikz}{}{
\input{./figures/interp-of-tensor-product-visualised-2.tikz}%} % chktex 27
 \]
%The two wires on the right act as ``control wires''. They are the ones that bear the encoding $\psi$, and as Proposition~\ref{prop:left-inverse-ZX}, they are the ones the ``decoder'' will be applied to.
It order for the tensor product to make sense, $D_1$ and $D_2$ should be able to commute on the control wires. This property is provided by the completeness of the ZW-Calculus, since it is semantically true.

One can check that $\interp{\left[
%\InputIfFileExists{#1.tikz}{}{
\begin{tikzpicture}
	\begin{pgfonlayer}{nodelayer}
		\node [style={H box}] (0) at (0, 0) {};
		\node [style=none] (1) at (0, 0.5) {};
		\node [style=none] (2) at (0, -0.5) {};
	\end{pgfonlayer}
	\begin{pgfonlayer}{edgelayer}
		\draw (2.center) to (1.center);
	\end{pgfonlayer}
\end{tikzpicture}
%} % chktex 27
\right]_W} = \psi\left( \interp{
%\InputIfFileExists{#1.tikz}{}{
\begin{tikzpicture}
	\begin{pgfonlayer}{nodelayer}
		\node [style={H box}] (0) at (0, 0) {};
		\node [style=none] (1) at (0, 0.5) {};
		\node [style=none] (2) at (0, -0.5) {};
	\end{pgfonlayer}
	\begin{pgfonlayer}{edgelayer}
		\draw (2.center) to (1.center);
	\end{pgfonlayer}
\end{tikzpicture}
%} % chktex 27
} \right) = \frac{1}{2}\begin{pmatrix}1&1\\1&-1\end{pmatrix}\otimes (M-M^3)$ and $\interp{\left[
%\InputIfFileExists{#1.tikz}{}{
\begin{tikzpicture}
	\begin{pgfonlayer}{nodelayer}
		\node [style=gn] (0) at (0, -0) {$\frac{\pi}{4}$};
		\node [style=none] (1) at (0, 0.5000001) {};
		\node [style=none] (2) at (0, -0.5000001) {};
	\end{pgfonlayer}
	\begin{pgfonlayer}{edgelayer}
		\draw (1.center) to (2.center);
	\end{pgfonlayer}
\end{tikzpicture}%} % chktex 27
\right]_W} = \psi\left( \interp{
%\InputIfFileExists{#1.tikz}{}{
\begin{tikzpicture}
	\begin{pgfonlayer}{nodelayer}
		\node [style=gn] (0) at (0, -0) {$\frac{\pi}{4}$};
		\node [style=none] (1) at (0, 0.5000001) {};
		\node [style=none] (2) at (0, -0.5000001) {};
	\end{pgfonlayer}
	\begin{pgfonlayer}{edgelayer}
		\draw (1.center) to (2.center);
	\end{pgfonlayer}
\end{tikzpicture}%} % chktex 27
} \right) = \begin{pmatrix}I_4 & 0 \\ 0 & M\end{pmatrix}$. More generally:

    \begin{prop}%
      \label{prop:pxwpsi}
      Let $D$ be a diagram of the Clifford+T ZX-Calculus. Then
      \[\interp{[D]_W}=\psi(\interp{D})\]

      In particular, if $\interp{D_1} = \interp{D_2}$ then
      $\interp{[D_1]_W} = \interp{[D_2]_W}$
    \end{prop}
    The proof is a straightforward induction using the fact
    that $\psi$ is an homomorphism. Slight care has to be taken to
    treat the case of $D_1 \otimes D_2$:\\
Suppose $\interp{D_1}=\sum\limits_{k=0}^3A_k e^{i\frac{k\pi}{4}}$ and $\interp{D_2}=\sum\limits_{k=0}^3B_k e^{i\frac{k\pi}{4}}$ are their \emph{unique} decomposition, and that $\interp{[D_i]_W}=\psi(\interp{D_i})$. Then:
\begin{align*}
\interp{[D_1\otimes D_2]_W} &= \left(I\otimes \psi(\interp{D_1})\right)\circ \interp{\scalebox{0.7}{\nmcrossI{}{}~}}\circ \left(I\otimes \sum\limits_{k=0}^3A_k\otimes M^k\right)\circ \interp{\scalebox{0.7}{\nmcrossI{}{}~}}\\
&= \left(\sum\limits_{l=0}^3I\otimes B_l \otimes M^l\right)\circ \left(\sum\limits_{k=0}^3A_k \otimes I\otimes M^k\right)
= \sum_{k,l}A_k\otimes B_l \otimes M^{k+l}\\
&= \psi\left(\sum_{k,l}(A_k\otimes B_l) e^{i\frac{(k+l)\pi}{4}}\right) = \psi(\interp{D_1\otimes D_2})
\end{align*}

\subsection{From ZW\texorpdfstring{\!$_\frac{1}{2}$}{(1/2)} to Clifford+T ZX-Diagrams}%
\label{subsec:zwh-to-clifft}

We define here an interpretation $[.]_X$ that transforms any diagram of the \zwh-Calculus into a Clifford+T ZX-diagram, which is easy to do since $\mathbb{D}\subset\mathbb{D}[\piq{}]$:

\noindent
\begin{minipage}{\columnwidth}
\titlerule{$[.]_X$}
%$\hfill \tikzfig{empty-diagram}~~\mapsto~~\tikzfig{empty-diagram}\hfill \half~~\mapsto~~\tikzfig{half-ZX}\hfill$\\
%\vspace{-2em}
%\begin{multicols}{3}
\[ 
%\InputIfFileExists{#1.tikz}{}{
\input{./figures/empty-diagram.tikz}%} % chktex 27
 \quad\mapsto\quad 
%\InputIfFileExists{#1.tikz}{}{
\input{./figures/empty-diagram.tikz}%} % chktex 27
\qquad\qquad

%\InputIfFileExists{#1.tikz}{}{
\begin{tikzpicture}
	\begin{pgfonlayer}{nodelayer}
		\node [style=none] (0) at (0, 0.2499999) {};
		\node [style=none] (1) at (0, -0.2499999) {};
	\end{pgfonlayer}
	\begin{pgfonlayer}{edgelayer}
		\draw (0.center) to (1.center);
	\end{pgfonlayer}
\end{tikzpicture}%} % chktex 27
 \quad\mapsto\quad 
%\InputIfFileExists{#1.tikz}{}{
\begin{tikzpicture}
	\begin{pgfonlayer}{nodelayer}
		\node [style=none] (0) at (0, 0.2499999) {};
		\node [style=none] (1) at (0, -0.2499999) {};
	\end{pgfonlayer}
	\begin{pgfonlayer}{edgelayer}
		\draw (0.center) to (1.center);
	\end{pgfonlayer}
\end{tikzpicture}%} % chktex 27
\qquad\qquad

%\InputIfFileExists{#1.tikz}{}{
\begin{tikzpicture}
	\begin{pgfonlayer}{nodelayer}
		\node [style=none] (0) at (-0.2500001, -0) {};
		\node [style=none] (1) at (0.2500001, -0) {};
	\end{pgfonlayer}
	\begin{pgfonlayer}{edgelayer}
		\draw [bend left=90, looseness=1.75] (0.center) to (1.center);
	\end{pgfonlayer}
\end{tikzpicture}%} % chktex 27
 \quad \raisebox{0.3em}{$\mapsto$} \quad 
%\InputIfFileExists{#1.tikz}{}{
\begin{tikzpicture}
	\begin{pgfonlayer}{nodelayer}
		\node [style=none] (0) at (-0.2500001, -0) {};
		\node [style=none] (1) at (0.2500001, -0) {};
	\end{pgfonlayer}
	\begin{pgfonlayer}{edgelayer}
		\draw [bend left=90, looseness=1.75] (0.center) to (1.center);
	\end{pgfonlayer}
\end{tikzpicture}%} % chktex 27
\qquad\qquad

%\InputIfFileExists{#1.tikz}{}{
\begin{tikzpicture}
	\begin{pgfonlayer}{nodelayer}
		\node [style=none] (0) at (-0.2500001, 0.2500001) {};
		\node [style=none] (1) at (0.2500001, 0.2500001) {};
	\end{pgfonlayer}
	\begin{pgfonlayer}{edgelayer}
		\draw [bend right=90, looseness=1.75] (0.center) to (1.center);
	\end{pgfonlayer}
\end{tikzpicture}%} % chktex 27
 \quad\raisebox{0.3em}{$\mapsto$}\quad 
%\InputIfFileExists{#1.tikz}{}{
\begin{tikzpicture}
	\begin{pgfonlayer}{nodelayer}
		\node [style=none] (0) at (-0.2500001, 0.2500001) {};
		\node [style=none] (1) at (0.2500001, 0.2500001) {};
	\end{pgfonlayer}
	\begin{pgfonlayer}{edgelayer}
		\draw [bend right=90, looseness=1.75] (0.center) to (1.center);
	\end{pgfonlayer}
\end{tikzpicture}%} % chktex 27
\]

\end{minipage}
\[
%\InputIfFileExists{#1.tikz}{}{
\input{./figures/ZW-to-ZX-braid-no-braid.tikz}%} % chktex 27
\qquad\qquad

%\InputIfFileExists{#1.tikz}{}{
\input{./figures/ZW-to-ZX-cross-no-braid.tikz}%} % chktex 27
\qquad\qquad
 \half \quad\mapsto\quad 
%\InputIfFileExists{#1.tikz}{}{
\input{./figures/half-ZX.tikz}%} % chktex 27
 \qquad\qquad

%\InputIfFileExists{#1.tikz}{}{
\input{./figures/ZW-to-ZX-white-dot-1-1.tikz}%} % chktex 27
\]
\[
%\InputIfFileExists{#1.tikz}{}{
\input{./figures/ZW-to-ZX-white-dot-2-1.tikz}%} % chktex 27
\qquad\qquad

%\InputIfFileExists{#1.tikz}{}{
\input{./figures/ZW-to-ZX-dot-1-1.tikz}%} % chktex 27
\qquad\qquad

%\InputIfFileExists{#1.tikz}{}{
\input{./figures/ZW-to-ZX-dot-1-2.tikz}%} % chktex 27
\]
%\end{multicols}
%\vspace{-3em}

\noindent\begin{minipage}{\columnwidth}
\[D_1\circ D_2\mapsto [D_1]_X\circ[D_2]_X\qquad D_1\otimes D_2\mapsto [D_1]_X\otimes [D_2]_X\]
\vspace{0.2em}
\end{minipage}
\rule{\columnwidth}{0.5pt}

\begin{prop}%
  \label{prop:interpwx}
  Let $D$ be a diagram of the \zwh calculus.
  Then $\interp{[D]_X}=\interp{D}$
\end{prop}
The proof is by induction on $D$.

This interpretation $[.]_X$ from the ZW-Calculus to the ZX-Calculus is pretty straightforward, except for the three-legged black node. Indeed, the two languages express two different kinds of interactions, but at the same time they share a family of generators (the ZW white node and the ZX green node are essentially the same).
This is where some syntactic sugar can come in handy.
\begin{defi}%
\label{def:triangle}
We define the ``triangle node'' as:
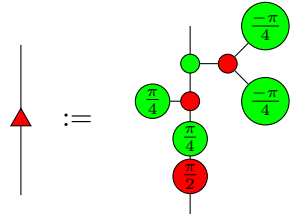
\[
%\InputIfFileExists{#1.tikz}{}{
\input{./figures/ug-decomp.tikz}%} % chktex 27
\]
\end{defi}
One can check that $\interp{~
%\InputIfFileExists{#1.tikz}{}{
\begin{tikzpicture}
	\begin{pgfonlayer}{nodelayer}
		\node [style=ug] (0) at (0, -0) {};
		\node [style=none] (1) at (0, 0.5000001) {};
		\node [style=none] (2) at (0, -0.5000001) {};
		\node [style=none] (3) at (0, -0.7499999) {};
		\node [style=none] (4) at (0, 0.7500001) {};
	\end{pgfonlayer}
	\begin{pgfonlayer}{edgelayer}
		\draw (1.center) to (2.center);
	\end{pgfonlayer}
\end{tikzpicture}%} % chktex 27
~} = \begin{pmatrix}1&1\\0&1\end{pmatrix}$. Then the interpretation of the three-legged black dot is simplified using \s:
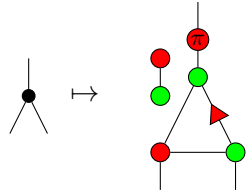
\[
%\InputIfFileExists{#1.tikz}{}{
\input{./figures/ZW-to-ZX-dot-1-2-simplified.tikz}%} % chktex 27
\]
%as is the rule \bw (see Lemma~\ref{lem:not-ug-is-symmetrical}):
%%\[\tikzfig{ug-and-not-W-commute-simplified}\quad \textnormal{(C2')}\]
%\[\tikzfig{not-ug-is-symmetrical}\]

This shortcut will be very useful in the technical proof of the completeness of the language for Clifford+T with the set of rules of $\zxct$.

\begin{prop}%
\label{prop:rules-preserved}
The interpretation $[.]_X$ preserves all the rules of the \zwh-Calculus:
\[ \zwh\vdash D_1=D_2 \quad\implies\quad \zxct\vdash [D_1]_X= [D_2]_X \]
\end{prop}
The proof, quite technical, is in Appendix at Section~\ref{prf:rules-preserved}.

\subsection{Composing the Interpretations}

To finish the proof it remains to compose the two interpretations and check that we can recover the initial diagram after the decoding part (Proposition~\ref{prop:left-inverse}), where:
\[\omega = \begin{pmatrix}1\\\piq{}\\ \piq2\\ \piq3\end{pmatrix} =\interp{\scalebox{0.8}{
%\InputIfFileExists{#1.tikz}{}{
\input{./figures/theta.tikz}%} % chktex 27
}}\qquad\text{and}\qquad
e_1=\begin{pmatrix}1&0&0&0\end{pmatrix} = \interp{\scalebox{0.8}{
%\InputIfFileExists{#1.tikz}{}{
\input{./figures/projector-1-0-0-0.tikz}%} % chktex 27
}}\]

\begin{prop}%
\label{prop:left-inverse-ZX}
We can recover any Clifford+T ZX-diagram $D$ from its image under the composition of the two interpretations:
\begin{align*}
\zxct\vdash D =\left(\hspace{-0.2em}\scalebox{0.8}{
%\InputIfFileExists{#1.tikz}{}{
\input{./figures/bottom-composition.tikz}%} % chktex 27
}\hspace{-0.3em}\right)\circ\left[[D]_W\right]_X\circ\left(\hspace{-0.2em}\scalebox{0.8}{
%\InputIfFileExists{#1.tikz}{}{
\input{./figures/top-composition.tikz}%} % chktex 27
}\hspace{-0.3em}\right)
\end{align*}
\end{prop}
\noindent The proof is in Appendix at Section~\ref{prf:left-inverse-ZX}.

\begin{cor}%
\label{cor:composition}
  If $\zxct \vdash \left[[D_1]_W\right]_X = \left[[D_2]_W\right]_X$
  then $\zxct \vdash D_1 = D_2$.
\end{cor}

We can now prove Theorem~\ref{thm:cliff-T}:

\begin{proof}[Proof of Thm.~\ref{thm:cliff-T}]
Let $D_1$ and $D_2$ be two Clifford+T ZX-diagrams, such that $\interp{D_1}=\interp{D_2}$. By Proposition~\ref{prop:pxwpsi}, $\interp{[D_1]_W}=\interp{[D_2]_W}$. Since $[D_i]_W$ are \zwh-diagrams, and since the \zwh-Calculus is complete, $\zwh\vdash [D_1]_W=[D_2]_W$. By Proposition~\ref{prop:rules-preserved}, $\zxct\vdash \left[[D_1]_W\right]_X = \left[[D_2]_W\right]_X$. Finally, by Corollary~\ref{cor:composition}, $\zxct \vdash D_1 = D_2$.
\end{proof}

\subsection{Expressive Power of the Clifford+T ZX-Diagrams}%
\label{section:exp}

Since the unitary matrices over $\mathbb{D}[\piq{}]$ are representable with Clifford+T circuits~\cite{SelingerGiles}, so are they with \frag4 diagrams. Using the fact that the $\zwh$-Diagrams can represent any matrix of $\mathbb{D}^{2^n}\times\mathbb{D}^{2^m}$ (Proposition~\ref{prop:zwh-universal}), we actually show that any matrix over $\mathbb{D}[\piq{}]$ can be represented by a \zxct-diagram:

\begin{prop}%
\label{prop:universality}
The \frag{4} of the ZX-Calculus represents exactly  matrices over $\mathbb{D}[\piq{}]$:
\[\forall A\in \mathbb{D}[\piq{}]^{2^n\times 2^m},~~\exists D\in \zxct,~~ \interp{D}=A\]
\end{prop}
\begin{proof}
Let $A\in \mathbb{D}[\piq{}]^{2^n\times 2^m}$. We define $A'=\psi(A)\in\mathbb{D}^{2^{n+2}\times 2^{m+2}}$. Since \zwh-diagrams are universal for matrices over dyadic rationals: $\exists D\in \zwh,~\interp{D}=A'$. Since $[.]_X$ preserves the semantics, we can define a ZX-diagram of the \frag{4} $D'=[D]_X$ such that $\interp{D'}=A'$.\\
%Now, notice that $\omega = \begin{pmatrix}1\\\piq{}\\ \piq2\\ \piq3\end{pmatrix} =\interp{\scalebox{0.8}{\tikzfig{theta}}}$, and $e_1=\begin{pmatrix}1&0&0&0\end{pmatrix} = \interp{\scalebox{0.8}{\tikzfig{projector-1-0-0-0}}}$, so if we apply the second diagram at the two bottom right wires, and the first state on the two top right wires of $D'$, we end up with $D''$ such that $\interp{D''}=A$. Indeed:
%\begin{align*}
%\interp{D''} &= (I\otimes
%e_1)\circ\interp{D'}\circ\left(I\otimes\omega\right)
%= (I\otimes e_1)\circ\psi(A)\circ\left(I\otimes\omega\right)\\
%&=(I\otimes e_1)\circ(A\otimes\omega)
%=A\otimes\left(e_1\circ\omega\right)
%= A
%\end{align*}
Now, by plugging \scalebox{0.8}{
%\InputIfFileExists{#1.tikz}{}{
\input{./figures/theta.tikz}%} % chktex 27
} at the top and \scalebox{0.8}{
%\InputIfFileExists{#1.tikz}{}{
\input{./figures/projector-1-0-0-0.tikz}%} % chktex 27
} at the bottom of the two rightmost wires of $D'$, we get the diagram $D''$ which is such that $\interp{D''}=A$. % chktex 8
\end{proof}

\section{Completeness for Linear Diagrams with Clifford+T Constants}%
\label{sec:lin-diag}

\subsection{Variables and Constants}

It is customary to view some angles in the ZX-diagrams as variables, in order to prove families of equalities. For instance, the rule \s displays two variables $\alpha$ and $\beta$, and potentially gives an infinite number of equalities. Notice that in the axioms of the ZX-calculus, the variables are used in a linear way, reflecting the phase group structure.

\begin{defi} %Given $C\subseteq [0, 2\pi)$,
A ZX-diagram is linear in $\alpha_1, \ldots, \alpha_k$ with constants in $C\subseteq \mathbb R$, if it is generated by $R_Z^{(n,m)}(E)$, $R_X^{(n,m)}(E)$, $H$, $e$, $\mathbb I$, $\sigma$, $\epsilon$, $\eta$, and the spacial and sequential compositions, where $n,m\in \mathbb  N$, and $E$ is an affine combination of $\alpha_i$ with coefficients in $\mathbb{Z}$ and constants in $C$, i.e.~of the form $\sum_{i} n_i \alpha_i+c$, with $n_i\in \mathbb Z$ and $c\in C$.
\end{defi}

Notice that all the diagrams in Figures~\ref{fig:ZX-Clifford-rules} and~\ref{fig:ZX-CliffordT-rules} are linear in $\alpha, \beta, \gamma$ with constants in $\frac \pi 4 \mathbb Z$. A diagram linear in $\alpha_1, \ldots,  \alpha_k$ is denoted $D(\alpha_1, \ldots, \alpha_k)$, or more compactly $D(\vec{\alpha})$ with $\vec{\alpha} = \alpha_1, \ldots, \alpha_k$. Obviously, if $D(\alpha)$ is a diagram linear in $\alpha$, $D(\pi/2)$ denotes the ZX-diagram where all occurrences of $\alpha$ are replaced by $\pi/2$.

%We now want to prove Theorem~\ref{thm:lin-diag}

\subsection{From variables to inputs}
%The variables involved in a

We show in this section that, given an equation involving diagrams linear in some variable $\alpha$, the variables can be \emph{extracted}, splitting  the diagrams into two parts: a collection of points (nodes with parameter $\alpha$) and a constant diagram independent of the variables. % as described in Proposition \label{prop:var2inp}

First we define the multiplicity of a variable in an equation:

\begin{defi}
For any $D_1(\vec{\alpha}), D_2(\vec{\alpha}): n\to m$ two ZX-diagrams linear in $\vec{\alpha}$, the multiplicity of $\alpha_1$ in the equation $D_1(\vec{\alpha}) = D_2(\vec{\alpha})$ is defined as:
\[\mu_{\alpha_1} = \max_{i\in \{1,2\}}\left(\mu^+_{\alpha_1}(D_i(\vec{\alpha}))\right)  + \max_{i\in \{1,2\}}\left(\mu^-_{\alpha_1}(D_i(\vec{\alpha}))\right)\]%\mu_\alpha^-(D_i(\alpha))+\mu^-_\alpha(D_{1-i}(\alpha))\right)$$
 where
$\mu^+_{\alpha_1}(D)$ (resp. $\mu^-_{\alpha_1}(D)$) is the number of occurrences of $\alpha_1$ (resp. -$\alpha_1$) in $D$,  inductively defined as   \\
$\mu^+_{\alpha_1}(R_Z^{(n,m)}(\ell\alpha_1+E(\alpha_2\cdots \alpha_n)))=\mu^+_{\alpha_1}(R_X^{(n,m)}(\ell\alpha_1+E(\alpha_2\cdots \alpha_n)))=\begin{cases} \ell &\text{if $\ell>0$}\\0&\text{otherwise}\end{cases}$\\
$\mu^-_{\alpha_1}(R_Z^{(n,m)}(\ell\alpha_1+E(\alpha_2\cdots \alpha_n)))=\mu^-_{\alpha_1}(R_X^{(n,m)}(\ell\alpha_1+E(\alpha_2\cdots \alpha_n)))=\begin{cases} -\ell&\text{if $\ell<0$}\\0&\text{otherwise}\end{cases}$\\
$\mu^\pm_{\alpha_1}(D\otimes D') = \mu^\pm_{\alpha_1}(D\circ D') = \mu^\pm_{\alpha_1}(D)+\mu^\pm_{\alpha_1}(D')$\\
$\mu^\pm_{\alpha_1}(H)=\mu^\pm_{\alpha_1}(e)=\mu^\pm_{\alpha_1}(\mathbb I)=\mu^\pm_{\alpha_1}(\sigma)=\mu^\pm_{\alpha_1}(\epsilon)=\mu^\pm_{\alpha_1}(\eta)=0$

%$\mu^+_\alpha(D\otimes D') = \mu^+_\alpha(D\circ D') = \mu^+_\alpha(D)+\mu^+_\alpha(D')$\\
%$\mu^+_\alpha(R_Z^{(n,m)}(n\alpha+c))=\mu^+_\alpha(R_X^{(n,m)}(n\alpha+c))=\begin{cases} n_i&\text{if $n_i>0$}\\0&\text{otherwise}\end{cases}$\\
%$\mu^+_\alpha(H)=\mu^+_\alpha(e)=\mu^+_\alpha(\mathbb I)=\mu^+_\alpha(\sigma)=\mu^+_\alpha(\epsilon)=\mu^+_\alpha(\eta)=0$\\
%$\mu^-_\alpha(D\otimes D') = \mu^-_\alpha(D\circ D') = \mu^-_\alpha(D)+\mu^-_\alpha(D')$\\
%$\mu^-_\alpha(R_Z^{(n,m)}(n\alpha+c))=\mu^-_\alpha(R_X^{(n,m)}(n\alpha+c))=\begin{cases} -n_i&\text{if $n_i<0$}\\0&\text{otherwise}\end{cases}$\\
%$\mu^-_\alpha(H)=\mu^-_\alpha(e)=\mu^-_\alpha(\mathbb I)=\mu^-_\alpha(\sigma)=\mu^-_\alpha(\epsilon)=\mu^-_\alpha(\eta)=0$\\
\end{defi}

For instance, consider the following equation:
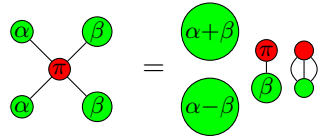
\[
%\InputIfFileExists{#1.tikz}{}{
\input{./figures/big-scalar-equation.tikz}%} % chktex 27
\]
The multiplicity of $\alpha$ is $\mu_\alpha = 2$ and $\beta$'s is $\mu_\beta = 3$.

\begin{prop}\label{prop:var2inp}
For any $D_1(\alpha), D_2(\alpha): n\to m$ two ZX-diagrams linear in $\alpha$ with constants in $\frac \pi 4$ $\mathbb Z$, there exist  $D_1', D'_2:r\to n+m$ two ZX-diagrams with angles multiple of $\frac \pi 4$ such that, for any $\alpha \in \mathbb R$, the equivalence
\begin{equation}
D_1(\alpha)= D_2(\alpha) \equi{} D_1'\circ \theta_r(\alpha) =  D_2'\circ \theta_r(\alpha)
\end{equation}
is provable using the axioms of $\zxct$, where $r$ is the multiplicity of $\alpha$ in $D_1(\alpha) = D_2(\alpha)$, and $\theta_r(\alpha)= \left(R_Z^{(0,1)}(\alpha)\right)^{\otimes r}$.\\%=\tikzfig{theta_r-alpha}$. \\
Pictorially:
\def\fig{theta_r-alpha-on-diagrams}
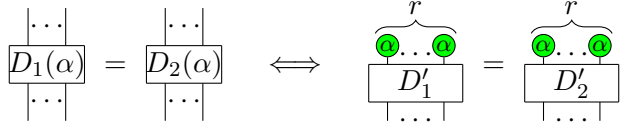
\[\input{./figures/\fig/\fig_00.tikz}\eq{}\input{./figures/\fig/\fig_01.tikz}\equi{}\input{./figures/\fig/\fig_02.tikz}\eq{}\input{./figures/\fig/\fig_03.tikz}\]
\end{prop}

\begin{proof}The proof consists in transforming the equation $D_1(\alpha)=D_2(\alpha)$ into the equivalent equation $D_1'\circ \theta_r(\alpha) =  D_1'\circ \theta_r(\alpha)$ using axioms of the ZX-calculus. This transformation involves 6 steps:
\begin{itemize}[label={--}]
    \item
    \emph{Turn inputs into outputs.} First, each input can be bent to an output using $\eta$:
\def\fig{thm1-equivalence}\[\input{./figures/\fig/\fig_00.tikz}\eq{}\input{./figures/\fig/\fig_01.tikz}\equi{}\input{./figures/\fig/\fig_02.tikz}\eq{}\input{./figures/\fig/\fig_03.tikz}\]

    \item
    \emph{Make the red spiders green.} All red spiders $R^{(k,l)}_X(n\alpha+c)$ are transformed into green spiders using the axioms \s and \h:
\def\fig{h-on-red-spiders}
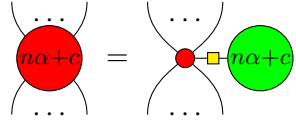
\[\input{./figures/\fig/\fig_00.tikz}\eq{}\input{./figures/\fig/\fig_01.tikz}\]

    \item
    \emph{Expending spiders.} All spiders $R_Z(n\alpha +c)$ are expended using \s so that all the occurrences of $\alpha$ are either \rz{$\alpha$} or \rz{-$\alpha$}:%E.g. TBC ($-2\alpha+\pi/4$)\[[[DESSIN]]\]
\def\fig{S1-on-n-alpha-plus-c-arxiv}
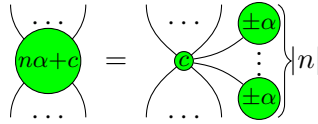
\[\input{./figures/\fig/\fig_00.tikz}\eq{}\input{./figures/\fig/\fig_01.tikz}\]

    \item
    \emph{Changing the sign.} Using \picom all occurrences of \rz{-$\alpha$}
 are replaced as follows: $\rz{-$\alpha$} \mapsto
%\InputIfFileExists{#1.tikz}{}{
\input{./figures/Rz-0-1-minus-alpha.tikz}%} % chktex 27
$. Notice that this rule is not applied recursively, which would not terminate. After this step all the original $-\alpha$ have been replaced by an $\alpha$ and as many scalars 
%\InputIfFileExists{#1.tikz}{}{
\input{./figures/scalar-e-pow-minus-i-alpha.tikz}%} % chktex 27
 have been created. So far, we have shown:
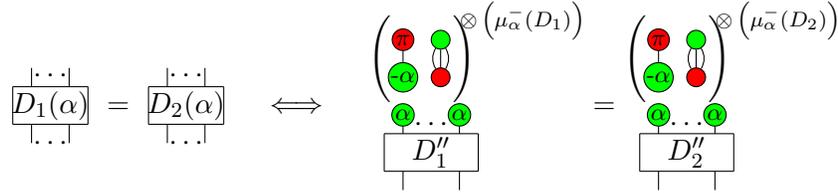
\def\fig{thm1-equivalence-one-var}
\begin{align*}
\input{./figures/\fig/\fig_00.tikz}\eq{}\input{./figures/\fig/\fig_01.tikz}\equi{}
\input{./figures/\fig/\fig_06.tikz}\eq[]{}\input{./figures/\fig/\fig_07.tikz}
\end{align*}

    \item
    \emph{(Re)moving scalars.} The scalar 
%\InputIfFileExists{#1.tikz}{}{
\input{./figures/scalar-e-pow-i-alpha.tikz}%} % chktex 27
 has an inverse for $\otimes$, which is 
%\InputIfFileExists{#1.tikz}{}{
\input{./figures/scalar-e-pow-minus-i-alpha.tikz}%} % chktex 27
 (see Lemmas~\ref{lem:multiplying-global-phases},~\ref{lem:bicolor-0-alpha} and~\ref{lem:inverse}). This has for consequence: % chktex 36
\begin{itemize}
\item $\zxct\vdash \scalebox{0.8}{
%\InputIfFileExists{#1.tikz}{}{
\input{./figures/scalar-e-pow-minus-i-alpha.tikz}%} % chktex 27
}D_1=D_2 \equi{}\zxct\vdash D_1=\scalebox{0.8}{
%\InputIfFileExists{#1.tikz}{}{
\input{./figures/scalar-e-pow-i-alpha.tikz}%} % chktex 27
}D_2$
\item $\zxct\vdash \scalebox{0.8}{
%\InputIfFileExists{#1.tikz}{}{
\input{./figures/scalar-e-pow-i-alpha.tikz}%} % chktex 27
}D_1=\scalebox{0.8}{
%\InputIfFileExists{#1.tikz}{}{
\input{./figures/scalar-e-pow-i-alpha.tikz}%} % chktex 27
}D_2 \equi{} \zxct\vdash D_1=D_2$
\end{itemize}
The scalars 
%\InputIfFileExists{#1.tikz}{}{
\input{./figures/scalar-e-pow-minus-i-alpha.tikz}%} % chktex 27
 are eliminated by adding $\overset{-}{\mu}\substack{\max\\\alpha\ \ }:=\max\left(\mu^-_\alpha(D_1),\mu^-_\alpha(D_2)\right)$ times the scalar 
%\InputIfFileExists{#1.tikz}{}{
\input{./figures/scalar-e-pow-i-alpha.tikz}%} % chktex 27
 on both sides, then simplifying when we have a scalar and its inverse.
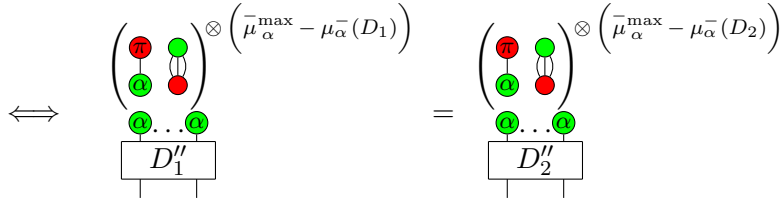
\begin{align*}
\equi{}\input{./figures/\fig/\fig_08.tikz}\eq{}\input{./figures/\fig/\fig_09.tikz}
\end{align*}

    \item
    \emph{Balancing the variables.} At this step the number of occurrences of $\alpha$ might be different on both sides of the equation. Indeed, one can check that the side of $D_i$ has $\mu^+_\alpha(D_i)+\overset{-}{\mu}\substack{\max\\\alpha\ \ }$ occurrences of $\alpha$. One can then use the simple equation 
%\InputIfFileExists{#1.tikz}{}{
\input{./figures/inverse-alpha.tikz}%} % chktex 27
 (whose proof uses Lemmas~\ref{lem:bicolor-0-alpha} and~\ref{lem:inverse}) $\overset{+}{\mu}\substack{\max\\\alpha\ \ }-\mu^+_\alpha(D_i)$ times on the side of $D_i$, where $\overset{+}{\mu}\substack{\max\\\alpha\ \ }:=\max\left(\mu^+_\alpha(D_1),\mu^+_\alpha(D_2)\right)$. We hence end up with $\mu_\alpha =\overset{+}{\mu}\substack{\max\\\alpha\ \ }+\overset{-}{\mu}\substack{\max\\\alpha\ \ }$ occurrences of $\alpha$ on both sides. Formally, $D_i'$ is defined as:
\begin{align*}
\input{./figures/\fig/\fig_10.tikz}:=\input{./figures/\fig/\fig_11.tikz}
\\[-\normalbaselineskip]\tag*{\qedhere}
\end{align*}
\end{itemize}
\end{proof}

\noindent
Proposition~\ref{prop:var2inp}  implies in particular that if the equation $D_1'\circ \theta_r(\alpha) =  D_2'\circ \theta_r(\alpha)$ is provable  using the axioms of the ZX-calculus, then so is $D_1(\alpha)= D_2(\alpha)$.
 Proposition~\ref{prop:var2inp} also implies that if $\interp{D_1(\alpha)}= \interp{D_2(\alpha)}$, then $ \interp{D_1'\circ \theta_r(\alpha)} = \interp{ D_2'\circ \theta_r(\alpha)}$, thanks to the soundness of the ZX-calculus.

%
%\begin{defi}The number of occurences of $\alpha$ in the equation $D_1(\alpha)=D_2(\alpha)$ is defined as the
% minimal $r$ for which equation (\ref{eq:varasinputs}) holds.
% \end{defi}
%
 \subsection{Removing the variables}
% \subsubsection{Multiplicity $1$}

 Given $D_1(\alpha)$ and $D_2(\alpha)$  linear in $\alpha$ with constants in $\frac \pi 4\mathbb Z$, if $\alpha$ has multiplicity $1$ in  ${D_1(\alpha)} = {D_2(\alpha)}$, then according to Proposition~\ref{prop:var2inp}, the equation  can be transformed into the following equivalent equation involving a single occurrence of $\alpha$:
\begin{equation}\label{eq:single}

%\InputIfFileExists{#1.tikz}{}{
\input{./figures/thm1-single-occurrence.tikz}%} % chktex 27

\end{equation}
 %$\alpha$ occurs only once in ${D_1(\alpha)} = {D_2(\alpha)}$ then according to Prop.~\ref{prop}, the equation is equivalent to [Dessin]
 where $D_1'$ and $D_2'$ are in the $\frac \pi 4$-fragment.
Notice that equation (\ref{eq:single}) holds
 if and only if $\interp{D_1'}=\interp{D_2'}$, since $\left(\rz{}, \rz{$\pi$}\right)$ forms a basis of the input space. Thus, a variable of multiplicity $1$ can easily be removed, leading to an equivalent equation in the complete $\frac \pi 4$-fragment of the ZX-calculus.

%Given an equation  ${D_1(\alpha)} = {D_2(\alpha)}$ where $D_1(\alpha)$ and $D_2(\alpha)$ are linear in $\alpha$ with constants in $\frac \pi 4$$\mathbb Z$, if $\alpha$ occurs only once, then according to Prop.~\ref{prop}, the equation is equivalent to [Dessin] where $D'_1$ and $D'_2$ are in the $\frac \pi 4$-fragment. In this particular case, notice that $\forall \alpha\in \mathbb R, \interp{D_1(\alpha)} = \interp {D_2(\alpha)}$ if and only if $\interp{D_1'}=\interp{D_2'}$, [poiint 0], [point$\pi$] form a basis,  leading to an equation in the complete $\frac \pi 4$-fragment of the ZX-calculus.
%
% the original equation holds for any alpha if and only if $\interp{D_1'}=\interp{D_2'}$.indeed [poiint 0], [point$\pi$ form a basis thus equation~\ref{} is satisfied iff $\interp{D_1'}=\interp{D_2'}$, leading to an equation in the complete $\frac \pi 4$-fragment of the ZX-calculus.

 %Thus, a variable can be removed when it appears only once.
 When a variable has a  multiplicity $r>1$ in an equation, the variable cannot be removed similarly as $\left(\rz{$\alpha$}\right)^{\otimes r}$ does not generate a basis of the $2^r$ dimensional space when $r>1$. However these dots can be replaced by an appropriate projector on the subspace generated by these dots, as described in the following.

% Let $S_n$ be the set of permutation on $\{1,\ldots,n\}$. For any permutation $\tau\in S_n$, let $Q_\tau:\mathbb C^{2^n}\to \mathbb C^{2^n}$ be the unique morphism such that $\forall x_1,\ldots x_n\in \{0,1\}, Q_\tau\ket{x_1\ldots x_n}=\ket{x_{\tau(1)}\ldots x_{\tau(n)}}$.

%
% \begin{defi}
%For any $r\ge 0$, the set of symmetric states on $r$ qubits be defined as \[\mathcal S_r := \{\varphi \in \mathbb C^{2^r}~|~Ê\forall \tau \in \textsf S_r, Q_\tau (\varphi) = \varphi\}\] where $\textsf S_r$ is the set of permutation on $\{1,\ldots,r\}$ and given $\tau\in \textsf S_r$,  $Q_\tau:\mathbb C^{2^r}\to \mathbb C^{2^r}$ is the unique morphism such that $\forall \varphi_1,\ldots, \varphi_r\in \mathbb C^{2}$, $Q_\tau(\varphi_1 \otimes \ldots \otimes \varphi_r)=\varphi_{\tau(1)} \otimes \ldots \otimes \varphi_{\tau(r)}$.
% \end{defi}
%

 %However these dots can be replaced by an appropriate projector on the subspace generated by these dots, as described in the following.

%
%
%Once the diagrams involved in an equation have been split into the variable part and the constant part, notice that if there is a single occurence of the variable, e.g.
%
%
% In this section we show that once the diagrams involved in an equation have been split into the variable part and the constant part, the variables can be replaced by a constant diagram.
%
%

%
% \begin{lem}
% For any $r\ge 0$, $\mathcal S_r$ is of dimension $r+1$ and is generated by $\{\theta_r(\alpha) ~|~\alpha \in \mathbb R\}$.
% \end{lem}
%
%
 \subsubsection{When the multiplicity is 2}

%First we consider the c

%
%
%
%We will first prove the theorem in the special case of two diagrams with only two occurrences of one variable as a warm-up before going to the general case.
%%We will then give an example on how this theorem can be used.

%\subsection{Diagrams with two inputs}
%We aim to find a complete set of rules for the general ZX-Calculus, but the set of rules ZX (Figure~\ref{fig:ZX_rules}) is not, as evidenced by~\cite{incompleteness,cyclo}. However, we can develop a procedure that can help us know if an equality can be derived with this set of rules.

%We will use in the following a shortcut introduced in~\cite{JPV}, the triangle:
%
%\begin{defi}
%\label{def:triangle}
%We define the triangle node as:
%\[\tikzfig{ug-decomp}\qquad\qquad\qquad\text{with interpretation}\interp{\tikzfig{ug-node}}=\begin{pmatrix}1&1\\0&1\end{pmatrix}\]
%\end{defi}

Consider the following diagram $R$:
\def\varone{$R$}
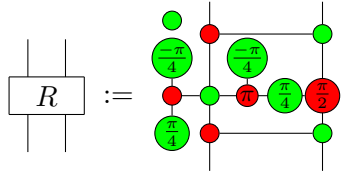
\[
%\InputIfFileExists{#1.tikz}{}{
\input{./figures/2-in-2-out-box-var.tikz}%} % chktex 27
~~:=~~
%\InputIfFileExists{#1.tikz}{}{
\input{./figures/matrix-X.tikz}%} % chktex 27
\]
One can check that
\[\interp R=\begin{pmatrix}
1&0&0&0\\
0&\frac{1}{2}&\frac{1}{2}&0\\
0&\frac{1}{2}&\frac{1}{2}&0\\
0&0&0&1
\end{pmatrix}.\]
This matrix basically mixes the second and third elements of any size-4 vector. We can then show:

\begin{lem} For any $\alpha \in \mathbb R$, $\zxct\vdash R\circ \theta_2(\alpha) = \theta_2(\alpha)$, i.e.\ pictorially:%
\label{lem:alphas-on-X}
\[\forall\alpha\in\mathbb{R},~~\zxct\vdash~ 
%\InputIfFileExists{#1.tikz}{}{
\input{./figures/2-gn-alpha-to-X.tikz}%} % chktex 27
\]
\end{lem}

The proof is given in appendix.

\begin{lem}%
\label{lem:equivalence-X}
For any two ZX-diagrams $D_1,D_2 : 2 \to n$,\\
$(\forall \alpha\in \mathbb{R}, \interp{D_1\circ \theta_2(\alpha)} =  \interp{D_2\circ \theta_2(\alpha)}) \Leftrightarrow \interp{D_1\circ R}=\interp{D_2\circ R}$ i.e.,
\[\left(\forall\alpha\in\mathbb{R},~~ \interp{
%\InputIfFileExists{#1.tikz}{}{
\input{./figures/2-gn-alpha-to-D_1.tikz}%} % chktex 27
}  = \interp{
%\InputIfFileExists{#1.tikz}{}{
\input{./figures/2-gn-alpha-to-D_2.tikz}%} % chktex 27
}\right) \Leftrightarrow \interp{
%\InputIfFileExists{#1.tikz}{}{
\input{./figures/2-gn-D_1-X.tikz}%} % chktex 27
} = \interp{
%\InputIfFileExists{#1.tikz}{}{
\input{./figures/2-gn-D_2-X.tikz}%} % chktex 27
}\]
where $\alpha$ does not appear in $D_1$ or $D_2$.
 %[[TODO: adapter le dessin n+1 -> m]]
\end{lem}

\begin{proof}
The proof consists in showing that $\interp R$ is a projector onto $S= \mathop{\mathrm{span}} \{ \interp{\theta_2(\alpha)}~|~\alpha \in \mathbb{R}\}$. According to Lemma~\ref{lem:alphas-on-X}, $\interp R$ is the identity on $S$, moreover it is easy to show that $\interp R$ is a matrix of rank $3$ and that $\interp {\theta_2(0)},\interp {\theta_2(\pi/2)}, \interp {\theta_2(\pi)}$  are three linearly independent vectors in the image of $\interp R$.
%
%Let $\theta(\alpha) := \begin{pmatrix}
%  1\\e^{i\alpha}\\e^{i\alpha}\\e^{2i\alpha}\end{pmatrix} = \begin{pmatrix}1 \\e^{i\alpha}\end{pmatrix} \otimes \begin{pmatrix}1 \\e^{i\alpha}\end{pmatrix}$.
%
%  $X$ is a projector onto    $\mathop{\mathrm{span}} \{ \theta(\alpha), \alpha \in \mathbb{R}\}$. That is:
%  \begin{itemize}
%  \item $X \theta(\alpha) = \theta(\alpha)$ for all $\alpha$
%  \item For every vector $v$, $Xv$ is a linear combination of the $\theta(\alpha)$.
%  \end{itemize}
%
%  In particular, if $A$ and $B$ are $p \times 4$ matrices, then
%  \[(\forall\alpha,~ A\theta(\alpha)=B\theta(\alpha)) \Leftrightarrow (AX=BX)\]
%
%%\end{lem}
%%  \begin{proof}
%    The first point is a consequence of the previous lemma by looking at interpretations of the left and right terms.
%
%    For the second, notice that $X$ is a matrix of rank $3$ and that $\theta(0), \theta(\pi/2), \theta(\pi)$ are three linearly independant vectors in the image of $X$.
  \end{proof}

\subsubsection{Arbitrary multiplicity}%
\label{sec:multiple-variables}
%
%\begin{prop}
%  \label{prop:meta-2}
%  Let $D_1$ and $D_2$ be two diagrams of the Clifford+T fragment, and suppose that
%$\forall\alpha\in\mathbb{R},~~ \interp{\tikzfig{2-gn-alpha-to-D_1}}  = \interp{\tikzfig{2-gn-alpha-to-D_2}}$.
%\\  Then $\forall\alpha\in\mathbb{R},~~ ZX \vdash \tikzfig{2-gn-alpha-to-D_1}  = \tikzfig{2-gn-alpha-to-D_2}$
%\end{prop}
%
%\begin{proof}%[Proposition~\ref{prop:meta-2}]
%From Lemma~\ref{lem:equivalence-X} we get:
%\[  \interp{\tikzfig{2-gn-D_1-X}} = \interp{\tikzfig{2-gn-D_2-X}} \implies \zx \vdash ~~\tikzfig{2-gn-D_1-X} ~=~ \tikzfig{2-gn-D_2-X}\]
%by completeness of the Clifford+T fragment. We then conclude by Lemma~\ref{lem:alphas-on-X}.
%\end{proof}
%
%\subsection{The general case}
We now want to generalise Lemma~\ref{lem:equivalence-X} to any multiplicity $r$ of $\alpha$.
It turns out that there is no obvious generalization for $r$ wires of the matrix $\interp R$ expressible using angles multiple of $\frac \pi 4$, so we will rather use the following family  $(P_r)_{r\ge 1}$ of diagrams:

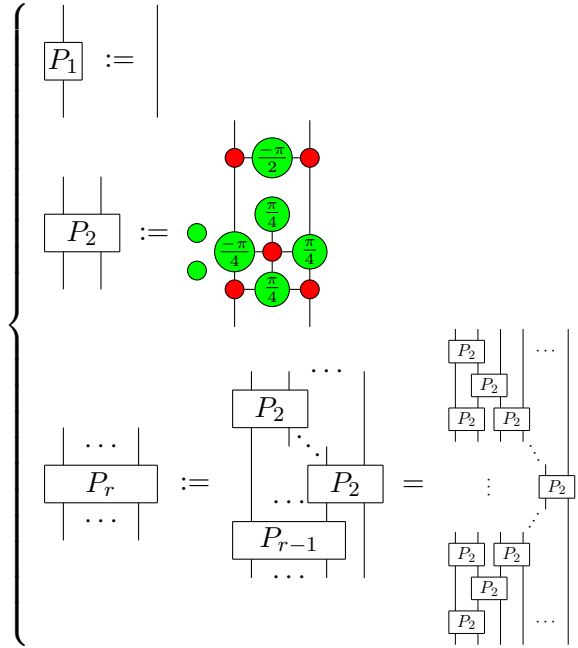
\[\left\lbrace\begin{array}{l}

%\InputIfFileExists{#1.tikz}{}{
\input{./figures/P1.tikz}%} % chktex 27
\\

%\InputIfFileExists{#1.tikz}{}{
\input{./figures/2-in-2-out-box-M_2.tikz}%} % chktex 27
~~:=~~ 
%\InputIfFileExists{#1.tikz}{}{
\input{./figures/matrix-M2-2-arxiv.tikz}%} % chktex 27
\\

%\InputIfFileExists{#1.tikz}{}{
\input{./figures/matrix-M_n-def.tikz}%} % chktex 27
~~=~~\scalebox{0.6}{
%\InputIfFileExists{#1.tikz}{}{
\input{./figures/matrix-M_n-form.tikz}%} % chktex 27
}
\end{array}\right.\]

%One has to be aware of the facts that these diagrams are oriented (they are not symmetrical), and that the inputs cannot be swapped.

For the reader convenience, here are the interpretations of $P_2$ and $P_3$:
\[
\interp{P_2} = \begin{pmatrix}
  1 & 0 & 0 & 0 \\
  0 & 0 & 1 & 0 \\
  0 & 0 & 1 & 0 \\
  0 & 0 & 0 & 1 \\
  \end{pmatrix}
\qquad \interp{P_3} = \begin{pmatrix}
  1 & 0 & 0 & 0 & 0 & 0 & 0 & 0  \\%000
  0 & 0 & 0 & 0 & 1 & 0 & 0 & 0 \\%001
  0 & 0 & 0 & 0 & 1 & 0 & 0 & 0 \\%010
  0 & 0 & 0 & 0 & 0 & 0 & 1 & 0  \\%011
  0 & 0 & 0 & 0 & 1 & 0 & 0 & 0  \\%100
  0 & 0 & 0 & 0 & 0 & 0 & 1 & 0 \\%101
  0 & 0 & 0 & 0 & 0 & 0 & 1 & 0 \\%110
  0 & 0 & 0 & 0 & 0 & 0 & 0 & 1  \\%111
  \end{pmatrix}
\]

We can characterise the interpretation of $P_r$ for any $r$.
\begin{prop}
For any word $\vec x\in\{0,1\}^r$, $\interp{P_r}^t\ket{\vec x}=\ket{1^{|\vec x|_1}0^{r-|\vec x|_1}}$ where $|\vec x|_1$ is the Hamming weight of $x$ i.e.~the number of symbol $1$ in the word $\vec x$.
\end{prop}

Informally, $\interp{P_r}^t$ sends all the words of the same Hamming weight to the word of the same weight where all the $1$s are on the left.

\begin{proof}
First of all, notice that the result is true for $P_2$:
\[\interp{P_2}^t\ket{00}=\ket{00},\qquad \interp{P_2}^t\ket{01}=\interp{P_2}^t\ket{10}=\ket{10},\qquad \interp{P_2}^t\ket{11}=\ket{11}\]
Let us denote $\operatorname{Op}^{[i_1,\ldots,i_k]}$ the application of the $k$-qubit operator $\operatorname{Op}$ on the wires $i_1,\ldots,i_k$. With this notation, $\interp{P_r}^t=\interp{P_2}^{t[1,2]}\circ \interp{P_2}^{t[2,3]}\circ\cdots\circ \interp{P_2}^{t[r-1,r]}\circ \interp{P_{r-1}}^{t[1,\ldots,r-1]}$.
We then prove the result by induction on $r$. Let $\vec x\in\{0,1\}^r$ be a word. Then:
\begin{align*}
\interp{P_{r+1}}\ket{\vec x0}
&= \interp{P_2}^{t[1,2]}\circ \interp{P_2}^{t[2,3]}\circ\cdots\circ \interp{P_2}^{t[r,r+1]}\circ \interp{P_r}^{t[1,\ldots,r]}\ket{\vec x0}\\
&= \interp{P_2}^{t[1,2]}\circ \cdots\circ \interp{P_2}^{t[r,r+1]} \ket{1^{|\vec x|_1}0^{r-|\vec x|_1}0}\\
%&= \interp{P_2}^{t[1,2]}\circ \cdots\circ \ket{1^{|\vec x|_1}0^{r-|\vec x|_1}0}\\
&= \cdots\\
&= \ket{1^{|\vec x|_1}0^{r+1-|\vec x|_1}}
\end{align*}
and
\begin{align*}
\interp{P_{r+1}}\ket{\vec x1}
&= \interp{P_2}^{t[1,2]}\circ \interp{P_2}^{t[2,3]}\circ\cdots\circ \interp{P_2}^{t[r,r+1]}\circ \interp{P_r}^{t[1,\ldots,r]}\ket{\vec x1}\\
&= \interp{P_2}^{t[1,2]}\circ \cdots\circ \interp{P_2}^{t[r,r+1]} \ket{1^{|\vec x|_1}0^{r-|\vec x|_1}1}\\
&= \interp{P_2}^{t[1,2]}\circ \cdots\circ \ket{1^{|\vec x|_1}0^{r-1-|\vec x|_1}10}\\
&= \cdots\\
&= \interp{P_2}^{t[1,2]}\circ \cdots\circ\interp{P_2}^{t[|\vec x|_1,|\vec x|_1+1]} \circ \ket{1^{|\vec x|_1}10^{r-|\vec x|_1}}\\
&= \cdots\\
&= \ket{1^{|\vec x|_1+1}0^{r-|\vec x|_1}}
\qedhere
\end{align*}
\end{proof}

\begin{cor}%
\label{lem:rank}
The rank of $\interp{P_r}$ is exactly $r+1$.
\end{cor}

\begin{lem}%
\label{lem:alphas-on-M}For any $r\ge 2$ and any $\alpha \in \mathbb R$, $\zxct\vdash P_r\circ \theta_r(\alpha) = \theta_r(\alpha)$ i.e.,
\[\zxct\vdash~ 
%\InputIfFileExists{#1.tikz}{}{
\input{./figures/n-gn-alpha-to-M_n.tikz}%} % chktex 27
\]
\end{lem}

\begin{proof}
Notice that $\interp{P_2\circ R} = \interp R$, so by completeness of the ZX-calculus for the $\frac \pi 4$ fragment, $\zxct\vdash P_2\circ R = R$, so  $\zx\vdash P_2\circ R \circ \theta_2(\alpha) = R\circ \theta_2(\alpha)$. According to Lemma~\ref{lem:alphas-on-X}, it implies $\zxct\vdash P_2\circ \theta_2(\alpha) = \theta_2(\alpha)$. The proof for $r>2$ is by induction on $r$.
%The result for $r\geq3$ is an easy recurrence on $r$. The base case $P_2$ is solved using proposition~\ref{prop:meta-2} since $P_2$ is a diagram with constant in $\frac \pi 4$$\mathbb Z$.
\end{proof}

We can now prove a similar statement as in lemma~\ref{lem:equivalence-X}:

\begin{lem}%
\label{lem:equivalence-Pk}
For any $r\ge 2$ and any $D_1,D_2 : r \to n$, \\
$(\forall \alpha\in \mathbb R, \interp{D_1\circ \theta_r(\alpha)} =  \interp{D_2\circ \theta_r(\alpha)}) \Leftrightarrow \interp{D_1\circ P_r}=\interp{D_2\circ P_r}$ i.e.,
%\begin{align*}
%\left(\forall\alpha\in\mathbb{R},~~ \interp{\tikzfig{r-gn-alpha-to-D_1}} = \interp{\tikzfig{r-gn-alpha-to-D_2}}\right)\\ \tag*{$\equi{} \interp{\tikzfig{P-to-D1}} = \interp{\tikzfig{P-to-D2}}$}
%\end{align*}
\begin{align*}
\left(\forall\alpha\in\mathbb{R},~~ \interp{
%\InputIfFileExists{#1.tikz}{}{
\input{./figures/r-gn-alpha-to-D_1.tikz}%} % chktex 27
} = \interp{
%\InputIfFileExists{#1.tikz}{}{
\input{./figures/r-gn-alpha-to-D_2.tikz}%} % chktex 27
}\right) \Leftrightarrow \interp{
%\InputIfFileExists{#1.tikz}{}{
\input{./figures/P-to-D1.tikz}%} % chktex 27
} = \interp{
%\InputIfFileExists{#1.tikz}{}{
\input{./figures/P-to-D2.tikz}%} % chktex 27
}
\end{align*}
where $\alpha$ does not appear in $D_1$ nor $D_2$.
\end{lem}

%
%\begin{prop}
%\label{prop:equivalence-n}
%Let $\theta_n(\alpha):=\begin{pmatrix}
%1\\e^{i\alpha}
%\end{pmatrix}^{\otimes n}$.
%
%
%  $M_n$ is a projector onto    $\mathop{\mathrm{span}} \{ \theta_n(\alpha), \alpha \in \mathbb{R}\}$. That is:
%  \begin{itemize}
%  \item $M_n \theta_n(\alpha) = \theta_n(\alpha)$ for all $\alpha$
%  \item For every vector $v$, $M_nv$ is a linear combination of the $\theta_n(\alpha)$.
%  \end{itemize}
%
%  In particular, if $A$ and $B$ are $p \times 2^n$ matrices, then
%  \[(\forall\alpha,~ A\theta(\alpha)=B\theta(\alpha)) \Leftrightarrow (AM_n=BM_n)\]
%
%
%\end{prop}

\begin{proof} The proof consists in showing that $\interp {P_r}$ is a projector onto $S_r= \mathop{\mathrm{span}} \{ \interp{\theta_r(\alpha)}~|~\alpha \in \mathbb{R}\}$. According to  Lemma~\ref{lem:alphas-on-M}, $\interp{P_r}$ is the identity on $S_r$, and $\interp {P_r}$ is of rank at most $r+1$ according to Corollary~\ref{lem:rank}, thus to finish the proof, it is sufficient to prove that the $r+1$ vectors $(\theta_r(\alpha^{(j)}))_{j=0\ldots r}$ are linearly independent, where $\alpha^{(j)} = j\pi/r$.

  Let $\lambda_0,\ldots,\lambda_r$ be scalars such that $\sum_j\lambda_j\theta_r(\alpha^{(j)})=0$.
  Notice that the $2^p$-th row (when rows are labeled from $1$ to $2^r$) of $\theta_r(\alpha^{(j)})$ is exactly $e^{ip\alpha^{(j)}}$.
  Therefore, if we look at all $2^p$-th rows of the equations, we obtain
\[ \begin{pmatrix}
1&1&\cdots&1\\
e^{i\alpha^{(0)}}&e^{i\alpha^{(1)}}&\cdots&e^{i\alpha^{(r)}}\\
\vdots&\vdots&\ddots&\vdots\\
e^{in\alpha^{(0)}}&e^{in\alpha^{(1)}}&\cdots&e^{in\alpha^{(r)}}
\end{pmatrix}\begin{pmatrix}
%\lambda_0&&&\\
%&\lambda_1&&\\
%&&\ddots&\\
%&&&\lambda_n
\lambda_0\\
\lambda_1\\
\vdots\\
\lambda_r
\end{pmatrix}=0 \]
However, the first matrix is a Vandermonde matrix (actually its transpose), with
$e^{i\alpha^{(j)}}=e^{i\alpha^{(l)}}$ iff $j=l$, which is enough to state that
this matrix is invertible. Therefore all $\lambda_j$ are equal to $0$
and the vectors $\theta_r(\alpha^{(j)})$ are linearly independent.
\end{proof}

We are now ready to prove the main theorem in the particular case of a single variable:

\begin{prop}\label{prop:1var}
For any $D_1(\alpha), D_2(\alpha)$ ZX-diagrams linear in $\alpha$ with constants in $\frac \pi 4 \mathbb Z$,
\[ \forall  \alpha \in \mathbb R, \interp{D_1( \alpha)}= \interp{D_2( \alpha)}      \,\,\,\,\,\Leftrightarrow\,\,\,\,\,  \forall  \alpha \in \mathbb R, \zxct\vdash D_1( \alpha) = D_2( \alpha) \]
%\[\left( \forall  \alpha \in \mathbb R, \interp{D_1( \alpha)}= \interp{D_2( \alpha)} \right)      \Leftrightarrow \left(  \forall  \alpha \in \mathbb R, \zx\vdash D_1( \alpha) = D_2( \alpha) \right)      \]
\end{prop}
\begin{proof}
{[$\Leftarrow$]} is a direct consequence of the soundness of the ZX-calculus.
[$\Rightarrow$] Assume $\forall  \alpha \in \mathbb R, \interp{D_1( \alpha)}= \interp{D_2( \alpha)}$. According to Proposition~\ref{prop:var2inp}, $\forall  \alpha \in \mathbb R, \interp{D'_1\circ \theta_r(\alpha)}= \interp{D'_2\circ \theta_r( \alpha)}$ where $D'_i$ are in the $\frac \pi 4$-fragment of the ZX-calculus. It implies, according to Lemma~\ref{lem:equivalence-Pk}, that  $\interp{D_1'\circ P_r} = \interp{D_2'\circ P_r}$. Thanks to the completeness of the ZX-calculus for the $\frac \pi 4$-fragment, $\zxct\vdash D_1'\circ P_r= D_2'\circ P_r$, so $\forall \alpha\in \mathbb R, \zxct\vdash D_1'\circ P_r \circ \theta_r(\alpha)= D_2'\circ P_r\circ \theta_r(\alpha)$. Thus, by Lemma~\ref{lem:alphas-on-M}, $\forall \alpha\in \mathbb R,~\zxct\vdash D_1' \circ \theta_r(\alpha)= D_2'\circ  \theta_r(\alpha)$, which is equivalent to $\forall  \alpha \in \mathbb R,~\zxct\vdash D_1( \alpha) = D_2( \alpha)$ according to Proposition~\ref{prop:var2inp}.
\end{proof}

\subsection{Multiple variables}

Proposition~\ref{prop:var2inp} can be straighforwardly extended to multiple variables:

\begin{prop}\label{prop:vars2inp}
For any $D_1(\vec{\alpha}), D_2(\vec{\alpha}): n\to m$ two ZX-diagrams linear in $\vec{\alpha}=\alpha_1,\ldots, \alpha_k$ with constants in $\frac \pi 4 \mathbb Z$, there exist  $D_1', D'_2:(\sum_{i=1}^k r_i)\to n+m$ two ZX-diagrams with angles multiple of $\frac \pi 4$ such that, for any $\vec{\alpha} \in \mathbb R^k$,
\begin{equation}
D_1(\vec{\alpha})= D_2(\vec{\alpha}) \quad\Leftrightarrow\quad D_1'\circ \theta_{\vec r}(\vec{\alpha}) =  D_2'\circ \theta_{\vec r}(\vec{\alpha})
\end{equation}
is provable using the axioms of the ZX-calculus, where $r_i$ is the multiplicity of $\alpha_i$ in $D_1(\vec{\alpha})=D_2(\vec{\alpha})$, $\vec r:=r_1, \ldots, r_k$, and $\theta_{\vec{r}}(\vec{\alpha}):=\theta_{r_1}(\alpha_1)  \otimes \ldots\otimes \theta_{r_k}(\alpha_k)$.

Pictorially:
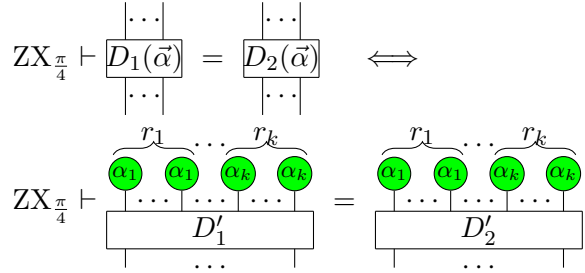
\def\fig{theta_r-alpha-on-diagrams-gen-arxiv}
\begin{align*}
&\zxct\vdash\input{./figures/\fig/\fig_00.tikz}\eq{}\input{./figures/\fig/\fig_01.tikz}\equi{}\\
&\zxct\vdash\input{./figures/\fig/\fig_02.tikz}\eq{}\input{./figures/\fig/\fig_03.tikz}
\end{align*}
\end{prop}

 Similarly Lemma~\ref{lem:equivalence-Pk} can also be extended to multiple variables:
\begin{lem}%
\label{lem:equivalence-Pk-vars}
For any $k\ge  0$, any $\vec r=r_1,\ldots, r_k\in \mathbb N^k$ and any $D_1,D_2 : (\sum_i r_i) \to n$, \\
$(\forall \vec{\alpha}\in \mathbb R^k\!, \interp{D_1\circ \theta_{\vec r}(\vec{\alpha})} =  \interp{D_2\circ \theta_{\vec r}(\vec{\alpha})}) \Leftrightarrow \interp{D_1\circ P_{\vec r}}=\interp{D_2\circ P_{\vec r}}$
%\[\left(\forall\alpha\in\mathbb{R},~~ \interp{\tikzfig{2-gn-alpha-to-D_1}}  = \interp{\tikzfig{2-gn-alpha-to-D_2}}\right) \Leftrightarrow \interp{\tikzfig{2-gn-D_1-X}} = \interp{\tikzfig{2-gn-D_2-X}}\]
where no $\alpha_i$ appear in $D_1$ or $D_2$, and $P_{\vec r} = P_{r_1} \otimes \ldots \otimes P_{r_k}$.
% [[TODO: adapter le dessin n+1 -> m]]
\end{lem}

Using Proposition~\ref{prop:vars2inp} and Lemma~\ref{lem:equivalence-Pk-vars} (whose proofs are similar to those of~\ref{prop:var2inp} and~\ref{lem:equivalence-Pk}), the proof of Theorem~\ref{thm:lin-diag} is similar to the single variable case (Proposition~\ref{prop:1var}) by induction.

Notice that Theorem~\ref{thm:lin-diag} implies that if $\forall \vec{\alpha} \in \mathbb R^k,  \interp{D_1(\vec{\alpha})}= \interp{D_2(\vec{\alpha})} $ then $D_1(\vec{\alpha})=D_2(\vec{\alpha})$ has a \emph{uniform} proof in the ZX-calculus in the sense that the structure of the proof is the same for all the values of $\vec{\alpha} \in \mathbb R^k$. Indeed, following the proof of Theorem~\ref{thm:lin-diag}, the sequence of axioms which leads to a proof of $D_1(\vec{\alpha})=D_2(\vec{\alpha})$ is independent of the particular values of $\vec{\alpha}$.  Notice, however, that Theorem~\ref{thm:lin-diag} is non constructive. %The proof strongly relies on the completeness of ZX-calculus for Clifford+T

\section{Applications of Linear Diagrams}%
\label{sec:appli-lin-diag}

In order to prove that $\forall \vec{\alpha} \in \mathbb R^k, \zxct\vdash D_1(\vec{\alpha})=D_2(\vec{\alpha})$ using  Theorem~\ref{thm:lin-diag}, one has to double check the semantic condition $\interp{D_1(\vec{\alpha})}= \interp{D_2(\vec{\alpha})}$ for all $\vec{\alpha} \in \mathbb R^k$, which might not be easy in practice. We show in the following alternative ways to prove $\forall \vec{\alpha} \in \mathbb R^k, \zxct\vdash D_1(\vec{\alpha})=D_2(\vec{\alpha})$, the two first based on a finite case-based reasoning in the ZX-calculus, and the last one by diagram substitution.

\subsection{Considering a basis}

\begin{thm}%
\label{thm:basis}
For any ZX-diagrams $D_1(\vec{\alpha}),D_2(\vec{\alpha}):1\to m$ linear in $\vec{\alpha}=\alpha_1, \ldots, \alpha_k$ with constants in $\frac \pi 4 \mathbb Z$, if %for any $j\in \{0,1\}$ and any $\vec{\alpha} \in \mathbb R^k$,
 \[\forall j\in \{0,1\},\forall \vec{\alpha} \in \mathbb R^k,  \zxct\vdash {D_1(\vec{\alpha})\circ R_X(j\pi)}= {D_2(\vec{\alpha})\circ R_X(j\pi)}\] % and $ \zx\vdash{D_1(\vec{\alpha})\circ (R_Z(\pi) \otimes \mathbb I^{\otimes n})}= {D_2(\vec{\alpha})\circ (R_Z(\pi) \otimes \mathbb I^{\otimes n})} $\\
 then \[\forall \vec{\alpha} \in \mathbb R^k, \zxct\vdash D_1(\vec{\alpha}) = D_2(\vec{\alpha})\]
\end{thm}

\begin{proof}
Assume $\zxct\vdash {D_1(\vec{\alpha})\circ R_X(j\pi)}= {D_2(\vec{\alpha})\circ R_X(j\pi)}$ for any $j\in \{0,1\}$ and any $\vec{\alpha} \in \mathbb R^k$.
It implies that for $x\in \left\lbrace \left(\begin{array}{c}1\\0\end{array}\right),\left(\begin{array}{c}0\\1\end{array}\right)\right\rbrace$, $\interp{D_1(\vec{\alpha})} x= \interp{D_2(\vec{\alpha})}x$, so $\interp{D_1(\vec{\alpha})} = \interp{D_2(\vec{\alpha})}$, which implies according to Theorem~\ref{thm:lin-diag} $\forall \vec{\alpha} \in \mathbb{R}^k, \zxct\vdash D_1(\vec{\alpha}) = D_2(\vec{\alpha})$. % chktex 1
\end{proof}

Notice that the Theorem~\ref{thm:basis} can be applied recursively: in order to prove the equality between two diagrams  with  $n$ inputs, $m$ outputs, and constants in $\frac{\pi}{4} \mathbb{Z}$, one can consider the $2^{n+m}$ ways to fix these inputs/outputs in a standard  basis states. It reduces the existence of a proof between two diagrams with constants in $\frac{\pi}{4} \mathbb{Z}$ to the existence of proofs on scalar diagrams (diagrams with no input and no output).   %It implies

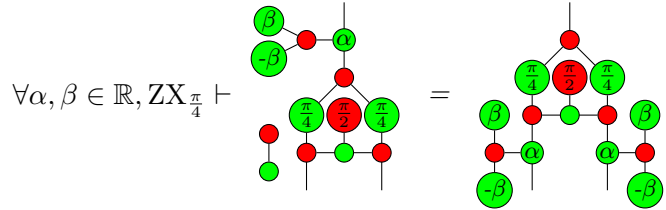
\begin{cor}
\label{cor:distribution}
\[\forall \alpha, \beta\in \mathbb R, \zxct\vdash~~ 
%\InputIfFileExists{#1.tikz}{}{
\input{./figures/add-axiom-2.tikz}%} % chktex 27
\]
\end{cor}
\begin{proof}
We can prove that this equality is derivable by plugging our basis $\left(\rx{},\rx{$\pi$}\right)$ on the input and one of the outputs. The detail is given in the appendix at Section~\ref{prf:distribution}.
\end{proof}

\subsection{Considering a finite set of angles}

%\begin{thm}
%For any ZX-diagrams $D_1(\vec{\alpha}),D_2(\vec{\alpha}):n\to m$ linear in $\alpha_1, \ldots, \alpha_k$ with constants in $\frac \pi 4 \mathbb Z$, if %for any $j\in \{0,1\}$ and any $\vec{\alpha} \in \mathbb R^k$,
% \[\forall \alpha_1\in K,\forall (\alpha_2,\ldots \alpha_k) \in \mathbb R^{k-1},  \zx\vdash {D_1(\vec{\alpha})}= {D_2(\vec{\alpha})}\] % and $ \zx\vdash{D_1(\vec{\alpha})\circ (R_Z(\pi) \otimes \mathbb I^{\otimes n})}= {D_2(\vec{\alpha})\circ (R_Z(\pi) \otimes \mathbb I^{\otimes n})} $\\
% then \[\forall \vec{\alpha} \in \mathbb R^k, \zx\vdash D_1(\vec{\alpha}) = D_2(\vec{\alpha})\]
% with $K$  a set of $\mu+1$  'independent' angles where $\mu$ is the multiplicity of $\alpha$ in $D_1(\vec{\alpha}) = D_2(\vec{\alpha})$.
%\end{thm}

\begin{thm}%
\label{thm:valuations}
For any ZX-diagrams $D_1(\vec{\alpha}),D_2(\vec{\alpha}):n\to m$ linear in $\vec{\alpha} = \alpha_1, \ldots, \alpha_k$ with constants in $\frac \pi 4 \mathbb Z$, if %for any $j\in \{0,1\}$ and any $\vec{\alpha} \in \mathbb R^k$,
 \[\forall \vec{\alpha} \in T_1\times \cdots \times T_k,  \zxct\vdash {D_1(\vec{\alpha})}= {D_2(\vec{\alpha})}\] % and $ \zx\vdash{D_1(\vec{\alpha})\circ (R_Z(\pi) \otimes \mathbb I^{\otimes n})}= {D_2(\vec{\alpha})\circ (R_Z(\pi) \otimes \mathbb I^{\otimes n})} $\\
 then \[\forall \vec{\alpha} \in \mathbb R^k, \zxct\vdash D_1(\vec{\alpha}) = D_2(\vec{\alpha})\]
 with $T_i$  a set of $\mu_i+1$ % 'independent' ({\bf TODO}) angles
 distinct angles in $\mathbb{R}/2\pi\mathbb{Z}$
  where $\mu_i$ is the multiplicity of $\alpha_i$ in $D_1(\vec{\alpha}) = D_2(\vec{\alpha})$.
\end{thm}

\begin{proof}%{\bf TODO?}
In the proof of Lemma~\ref{lem:equivalence-Pk}, we actually only used $\mu_{\alpha}+1$ values of $\alpha$, that constitute a basis of $S_{\mu_{\alpha}}$. This extends naturally to several variables: the dimension of $S_{\mu_{\alpha_1}}\times\cdots\times S_{\mu_{\alpha_k}}$ is $(\mu_{\alpha_1}+1)\times\cdots\times(\mu_{\alpha_k}+1)$, and taking $\vec{\alpha} \in T_1\times \cdots \times T_k$ gives as many linearly independent vectors in (hence a basis of) $S_{\mu_{\alpha_1}}\times\cdots\times S_{\mu_{\alpha_k}}$.
\end{proof}

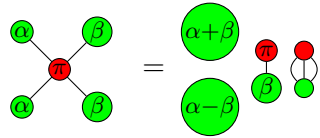
\begin{cor}%
\label{cor:big-scalar-equation}
\[
%\InputIfFileExists{#1.tikz}{}{
\input{./figures/big-scalar-equation.tikz}%} % chktex 27
\]
\end{cor}
\begin{proof}
Notice that $\mu_{\alpha}=2$ and $\mu_{\beta}=3$ in this equation. Hence we need to evaluate it for 12 values of $(\alpha,\beta)$, for instance for $\alpha,\beta\in\{0,\pi,\frac{\pi}{2}\}\times\{0,\pi,\frac{\pi}{2},-\frac{\pi}{2}\}$.
We can actually simplify the proof, by showing that whatever the value of $\beta\in\mathbb{R}$, the equation is derivable for $\alpha\in\{0,\pi,\frac{\pi}{2}\}$. This means the equation is derivable for all
%We can actually do not need to also evaluate $\beta$ because we can prove the equality  for all
$\alpha,\beta\in\{0,\pi,\frac{\pi}{2}\}\times\mathbb{R}$, and a fortiori for all $\alpha,\beta\in\{0,\pi,\frac{\pi}{2}\}\times\{0,\pi,\frac{\pi}{2},-\frac{\pi}{2}\}$ which would be a direct application of the theorem.
Details are in appendix at Section~\ref{prf:big-scalar-equation}.
\end{proof}

\begin{rem}
The number of occurrences of a variable is not to be mistaken for its multiplicity. For instance consider the following equation:
\[\rz{$\alpha$}=\rz{-$\alpha$}\]
This equation is obviously wrong in general, but not for $0$ and $\pi$. If we tried to apply Theorem~\ref{thm:valuations} with the number of occurrences (which seems to be $1$), then we might end up with the wrong conclusion. The multiplicity (here $\mu_{\alpha}=2$) prevents this.
\end{rem}

%
%{\bf TODO}\[\forall \alpha \in \frac{\pi}{4}\mathbb{Z},\quad \zx\vdash~~\tikzfig{example-non-uniform}\] since, in general, $\qquad\zx\nvdash~~\tikzfig{example-non-uniform}$ .

\subsection{Diagram substitution}
%\label{sec:diagram-substitution}

\begin{defi}
A diagram $D:0 \to n$ is symmetric if for any permutation $\tau$ on $\{1,\ldots n\}$, \[Q_\tau(\interp {D}) = \interp{D}\] where $Q_\tau:\mathbb C^{2^r}\to \mathbb C^{2^r}$ is the unique morphism such that:
\[
    \forall \varphi_1,\ldots, \varphi_r\in \mathbb C^{2},\ Q_\tau(\varphi_1 \otimes \ldots \otimes \varphi_r)=\varphi_{\tau(1)} \otimes \ldots \otimes \varphi_{\tau(r)}.
\]
%a permutation diagram is any diagram generated by $\sigma$ and $\mathbb I$ only.
\end{defi}

In particular for any diagram $D_0:0\to 1$, $D_0\otimes \ldots \otimes D_0$ is a symmetric diagram.

\begin{thm}%
\label{thm:diagram-substitution}
For any $D_1(\vec{\alpha}), D_2(\vec{\alpha}):r\to n$ and any symmetric $D(\vec{\alpha}):0\to r$ such that  $D_1(\vec{\alpha})$, $D_2(\vec{\alpha})$, and $D(\vec{\alpha})$ are linear in $\vec{\alpha}$ with constants in $\frac \pi 4\mathbb Z$, if $\forall \alpha_0\in \mathbb R, \forall\vec{\alpha}\in\mathbb{R}^k, \zxct\vdash D_1(\vec{\alpha}) \circ \theta_r(\alpha_0) = D_2(\vec{\alpha}) \circ \theta_r(\alpha_0)$ then $\forall\vec{\alpha}\in\mathbb{R}^k, \zxct\vdash D_1(\vec{\alpha}) \circ D(\vec{\alpha}) = D_2(\vec{\alpha}) \circ D(\vec  \alpha)$ i.e., pictorially:
\begin{align*}
\left(\forall \alpha_0\in \mathbb R, \forall\vec{\alpha}\in\mathbb{R}^k,~~ \zxct\vdash {
%\InputIfFileExists{#1.tikz}{}{
\input{./figures/2-gn-alpha-to-D_1-bis.tikz}%} % chktex 27
}  = {
%\InputIfFileExists{#1.tikz}{}{
\input{./figures/2-gn-alpha-to-D_2-bis.tikz}%} % chktex 27
}\right)
 \implies ~~\left(\forall\vec{\alpha}\in\mathbb{R}^k, \zxct\vdash {
%\InputIfFileExists{#1.tikz}{}{
\input{./figures/2-gn-D_1-X-bis.tikz}%} % chktex 27
} = {
%\InputIfFileExists{#1.tikz}{}{
\input{./figures/2-gn-D_2-X-bis.tikz}%} % chktex 27
}\right)
\end{align*}
\end{thm}

\begin{proof}
If $\forall \alpha_0\in \mathbb R, \forall\vec{\alpha}\in\mathbb{R}^k, \zxct\vdash D_1(\vec{\alpha}) \circ \theta_r(\alpha_0) = D_2(\vec{\alpha}) \circ \theta_r(\alpha_0)$ then $\interp{D_1(\vec{\alpha}) \circ \theta_r(\alpha_0)} = \interp{D_2(\vec{\alpha}) \circ \theta_r(\alpha_0)}$, so according to Lemma~\ref{lem:equivalence-Pk}, $\interp{D_1(\vec{\alpha}) \circ P_r} = \interp{D_2(\vec{\alpha}) \circ P_r}$. It implies that $\zxct\vdash D_1(\vec{\alpha}) \circ P_r  =D_2(\vec{\alpha}) \circ P_r$, so $\zxct\vdash D_1(\vec{\alpha}) \circ P_r \circ D(\vec{\alpha})  =D_2(\vec{\alpha}) \circ P_r\circ D(\vec{\alpha})$. To complete the proof, it is enough to show that $\zxct\vdash P_r\circ D(\vec{\alpha}) = D(\vec{\alpha})$. \\
Let $\mathcal S = \{\interp D ~|~D:0\to n \text{ symmetrical}\}$. First we show that $\mathcal S$ is of dimension at most $r+1$.
 %We have the inclusion $\mathcal S_r:= \mathop{\mathrm{span}} \{ \interp{\theta_r(\alpha)}~|~\alpha \in \mathbb{R}\}\subset \mathcal D$. To prove that $\mathcal S = \mathcal S_{r}$ we show that $\mathcal S$ is of dimension at most $r+1$.
Indeed, notice that if $\varphi \in \mathcal S$, then $\forall i,j\in \{0,\ldots, 2^r-1\}$ s.t. $|i|_1 = |j|_1$, $\varphi_i = \varphi_j$, where $|x|_1$ is the Hamming weight of the binary representation of $x$. As a consequence, for any $\varphi\in \mathcal S$, $\exists a_0, \ldots,a_r\in \mathbb C$ s.t. $\varphi = \sum_{h=0}^n a_h\varphi^{(h)}$ where $\varphi^{(h)}\in \mathbb C^{2^r}$ is defined as $\varphi^{(h)}_i=\begin{cases}1&\text{if $|i|_1 = h$}\\0&\text{otherwise}\end{cases}$. Thus $\mathcal S$ is of dimension at most $r+1$.  Moreover, for any $\alpha\in \mathbb R$, $\interp{\theta_r(\alpha)}\in \mathcal S$, so $\mathcal S \subseteq \mathcal S_r:= \mathop{\mathrm{span}} \{ \interp{\theta_r(\alpha)}~|~\alpha \in \mathbb{R}\}$. Since $\mathcal S_r$ is of dimension $r+1$ (see proof of Lemma~\ref{lem:equivalence-Pk}), $\mathcal S=\mathcal S_r$.
%Let $\mathcal S =  \mathop{\mathrm{span}}\{s_h~|~h\in [0,r]\}$ be the subspace of symmetric states of $\mathbb C^{2^r}$, with  $s_h:=\sum_{i\in [0,2^r-1] ~s.t.~w(i)=h} e^{(i)}$ where $w(i)$ is the Hamming weight of the binary representation of $i$ and $(e^{(i)})_{i\in [0,2^r-1]}$ is the standard basis of $\mathbb C^{2^r}$. \\
%Notice that $\forall \alpha\in \mathbb R$, $\interp {\theta_r(\alpha)} \in \mathcal S$ so $\mathcal S_r= \mathop{\mathrm{span}} \{ \interp{\theta_r(\alpha)}~|~\alpha \in \mathbb{R}\}$ is a subset of $\mathcal S$. Since both $\mathcal S$ and $\mathcal S_r$ are of dimension $r+1$, $\mathcal S_r=\mathcal S$.
As a consequence $\interp D\in \mathcal S_r$, so $\interp {P_r}\circ \interp {D(\vec{\alpha})} = \interp {D(\vec{\alpha})}$, since, according to Lemma~\ref{lem:alphas-on-M} for any $\alpha\in \mathbb R$, $\interp{P_r\circ \theta_r(\alpha)}= \interp {\theta_r(\alpha)}$. Thus, $\zxct\vdash P_r\circ D(\vec{\alpha})$ thanks to Theorem~\ref{thm:lin-diag}.% , we have $\interp {Pr}\circ \interp {D(\vec{\alpha})} = \interp {D(\vec{\alpha})}$, so $\zx\vdash P_r\circ D(\vec{\alpha})$.
\end{proof}

\begin{cor}%
\label{cor:gen-supp}
 \[\forall\alpha,\beta\in\mathbb{R}^2,\quad \zxct\vdash
%\InputIfFileExists{#1.tikz}{}{
\input{./figures/new-supplementarity.tikz}%} % chktex 27
\]
\end{cor}
\begin{proof}
 Indeed, simply by decomposing the colour-swapped version of \supp using \s, we can derive:
 \[\forall\alpha\in\mathbb{R},\quad\zxct\vdash 
%\InputIfFileExists{#1.tikz}{}{
\input{./figures/decomposed-supp.tikz}%} % chktex 27
\]
Now we just need to apply Theorem~\ref{thm:diagram-substitution} with
$
%\InputIfFileExists{#1.tikz}{}{
\input{./figures/diag-substitution-gen-supp-arxiv.tikz}%} % chktex 27
$ which is clearly symmetrical:
\def\fig{decomposed-supp-fin}
\begin{align*}
\input{./figures/\fig/\fig_00.tikz}
\eq{}\input{./figures/\fig/\fig_01.tikz}
\eq{}\input{./figures/\fig/\fig_02.tikz}
\eq{}\input{./figures/\fig/\fig_03.tikz}
\\[-\normalbaselineskip]\tag*{\qedhere}
\end{align*}
%and use \s to merge the adjacent red nodes.
\end{proof}

\section{Completeness for the General ZX-Calculus}%
\label{sec:gen-zx}

The previous result on linear diagrams gives a lot of power to the axiomatisation $\zxct$. We want now to complete this axiomatisation for the unrestricted ZX-Calculus, i.e.~we want to add enough axioms to $\zxct$ so that the resulting axiomatisation makes the general ZX-Calculus complete. This problem has been addressed in~\cite{HNW}, although to answer it, the authors added two generators to the language, and built a set of rules involving around 25 axioms. In the following, we show that we only need to add one axiom to $\zxct$.

\subsection{Incompleteness}

The axiomatisation $\zxct$ is complete for the Clifford+T quantum mechanics  --i.e.\ the $\frac \pi 4$-fragment--, but is not complete in general:

\begin{thm}%
\label{thm:incompleteness}
There exist two ZX-diagrams $D_1$ and $D_2$ such that:
\[\interp{D_1}=\interp{D_2}\qquad\text{and}\qquad \zxct\nvdash D_1=D_2\]
\end{thm}
\begin{proof}
Consider the following equation:
\[
%\InputIfFileExists{#1.tikz}{}{
\input{./figures/incomplete-pi_3.tikz}%} % chktex 27
\]
This equation is sound, it represents \[(1+e^{i\frac{2\pi}{3}})(1+e^{i\frac{4\pi}{3}})=1+e^{i\frac{2\pi}{3}}+e^{i\frac{4\pi}{3}}+e^{i\frac{6\pi}{3}}=1\]
However, consider the interpretation $\interp{.}_9$ that multiplies all the angles by $9$. All the multiples of $\frac{\pi}{4}$ remain unchanged ($\frac{k\pi}{4}\times9=\frac{k\pi}{4}+2k\pi=\frac{k\pi}{4}$). It is then easy to show that all the rules of $\zxct$ hold with this interpretation. However:
\def\fig{incomplete-pi_3-proof}
\[\interp{\input{./figures/\fig/\fig_00.tikz}}_9\eq{}\input{./figures/\fig/\fig_01.tikz}\ \neq\ \input{./figures/\fig/\fig_02.tikz}\eq{}\interp{\input{./figures/\fig/\fig_02.tikz}}_9\]
Indeed the left hand side amounts to $4$ while the right hand side amounts to $1$.
Since all the rules of $\zxct$ hold with this interpretation, if the calculus were complete, then it would prove the above equation and so its interpretation would hold. It does not, so the ZX-Calculus is not complete.
\end{proof}

Notice that thanks to Theorem~\ref{thm:lin-diag}, a completion of the ZX calculus would imply to  add  either non linear axioms, or axioms with constants not multiple of $\pi/4$. Such potential axioms have already been discovered, for instance the cyclotomic supplementarity~\cite{cyclo}:
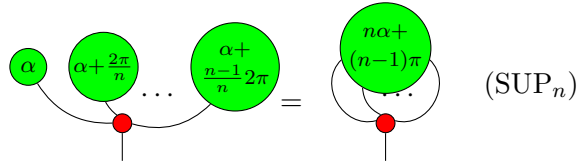
\[
%\InputIfFileExists{#1.tikz}{}{
\input{./figures/cyclo-supp.tikz}%} % chktex 27
\quad\text{($\text{SUP}_n$)}\]
Adding this family of axioms to those of $\zxct$ would nullify the counterexample in the proof of~\ref{thm:incompleteness} (the equality is derivable from $\zxct+\textnormal{(SUP$_3$)}$). However, the ZX-Calculus, with this set of axioms, would still be incomplete. Indeed, the argument given in~\cite{cyclo} still holds here.

In the following, we actually show that adding one axiom to $\zxct$ is sufficient to get the completeness in general. Contrary to the axioms of $\zxct$, this one manipulates angles in a non-linear fashion.

\subsection{From ZX to ZW}
%The ZW$_{\mathbb C}$ calculus introduced very recently by Hadzihasanovic~\cite{Amar} is an extension of the ZW-calculus.  ZW$_{\mathbb C}$ has been proved to be
%The ZW-Calculus is
%universal and complete for complex matrices. From this result follows naturally a completion procedure of the ZX-calculus: b
Both the ZX-Calculus and the \zwc-Calculus are universal for complex matrices, so there exists a pair of translations between the two languages which preserve the semantics ($[.]_X: \zwc\to \zx$ and $[.]_W:\zx\to \zwc$ s.t. $\forall D\in \zx, \interp{[D]_W}=\interp D$ and $\forall D\in \zwc, \interp{[D]_X}=\interp D$). The axiom \add has been chosen so that we can prove that $\left(\zxct+\add\right)\vdash {[[D]_W]_X}= D$ for any generator $D$ of the ZX-calculus and that $\left(\zxct+\add\right)\vdash [D_1]_X=[D_2]_X$ for any axiom $D_1=D_2$ of the ZW-Calculus. The choice of the translations is however essential as the new axiom relies on them. %The translation from ZX to ZW$_{\mathbb C}$ can be canonically defined using the normal forms of the ZW-calculus: for any generator $D$ of the ZX one can define $[D]_W$ as the ZW normal form representation of the matrix $\interp{D}$. This is however convenient to deviate from this canonically defined interpretation for the green spider $\alpha$ by choosing to translate it to a white spider $e^{i\alpha}$, (even if it is not in normal form).

%The $[.]_W$ translation %translation from ZX to ZW$_{\mathbb C}$
%could be canonically defined using the normal form of the \zwc-calculus: for any generator $D$ of the ZX one can define $[D]_W$ as the ZW normal form representation of the matrix $\interp{D}$. It is however convenient to deviate from this canonically defined interpretation for the green and red spiders and for the Hadamard gate. We end up with basically the same translation from ZX to \zwc as in~\cite{NgWang,HNW}:\\
In~\cite{NgWang,HNW}, the authors already use the same proof technique --- i.e.~transporting the completeness from $\zwc$ to ZX ---, and hence provide two interpretations from ZX to $\zwc$ and back. The former, that we denote $[.]_W$ absolutely suits our needs:

\noindent\begin{minipage}{\columnwidth}
\titlerule{$[.]_W$}
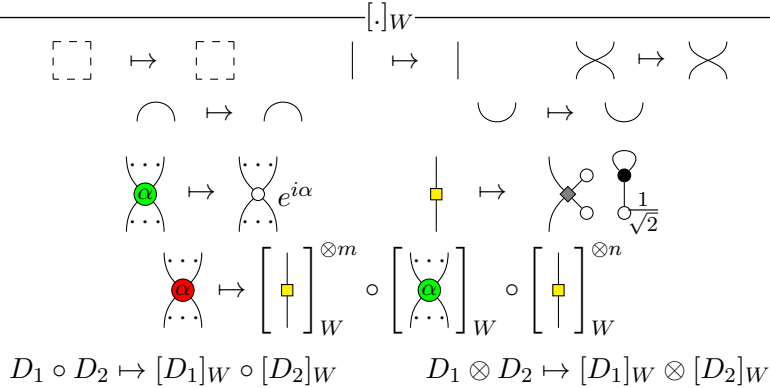
\[
%\InputIfFileExists{#1.tikz}{}{
\input{./figures/empty-diagram.tikz}%} % chktex 27
 \quad\mapsto\quad 
%\InputIfFileExists{#1.tikz}{}{
\input{./figures/empty-diagram.tikz}%} % chktex 27
\qquad\qquad

%\InputIfFileExists{#1.tikz}{}{
\begin{tikzpicture}
	\begin{pgfonlayer}{nodelayer}
		\node [style=none] (0) at (0, 0.2499999) {};
		\node [style=none] (1) at (0, -0.2499999) {};
	\end{pgfonlayer}
	\begin{pgfonlayer}{edgelayer}
		\draw (0.center) to (1.center);
	\end{pgfonlayer}
\end{tikzpicture}%} % chktex 27
 \quad\mapsto\quad 
%\InputIfFileExists{#1.tikz}{}{
\begin{tikzpicture}
	\begin{pgfonlayer}{nodelayer}
		\node [style=none] (0) at (0, 0.2499999) {};
		\node [style=none] (1) at (0, -0.2499999) {};
	\end{pgfonlayer}
	\begin{pgfonlayer}{edgelayer}
		\draw (0.center) to (1.center);
	\end{pgfonlayer}
\end{tikzpicture}%} % chktex 27
\qquad\qquad
%\InputIfFileExists{#1.tikz}{}{
\input{./figures/ZW-to-ZX-braid-no-braid.tikz}%} % chktex 27
\]
\end{minipage}
\[
%\InputIfFileExists{#1.tikz}{}{
\begin{tikzpicture}
	\begin{pgfonlayer}{nodelayer}
		\node [style=none] (0) at (-0.2500001, -0) {};
		\node [style=none] (1) at (0.2500001, -0) {};
	\end{pgfonlayer}
	\begin{pgfonlayer}{edgelayer}
		\draw [bend left=90, looseness=1.75] (0.center) to (1.center);
	\end{pgfonlayer}
\end{tikzpicture}%} % chktex 27
 \quad\raisebox{0.3em}{$\mapsto$}\quad 
%\InputIfFileExists{#1.tikz}{}{
\begin{tikzpicture}
	\begin{pgfonlayer}{nodelayer}
		\node [style=none] (0) at (-0.2500001, -0) {};
		\node [style=none] (1) at (0.2500001, -0) {};
	\end{pgfonlayer}
	\begin{pgfonlayer}{edgelayer}
		\draw [bend left=90, looseness=1.75] (0.center) to (1.center);
	\end{pgfonlayer}
\end{tikzpicture}%} % chktex 27
\hspace{6em}

%\InputIfFileExists{#1.tikz}{}{
\begin{tikzpicture}
	\begin{pgfonlayer}{nodelayer}
		\node [style=none] (0) at (-0.2500001, 0.2500001) {};
		\node [style=none] (1) at (0.2500001, 0.2500001) {};
	\end{pgfonlayer}
	\begin{pgfonlayer}{edgelayer}
		\draw [bend right=90, looseness=1.75] (0.center) to (1.center);
	\end{pgfonlayer}
\end{tikzpicture}%} % chktex 27
 \quad\raisebox{0.3em}{$\mapsto$}\quad 
%\InputIfFileExists{#1.tikz}{}{
\begin{tikzpicture}
	\begin{pgfonlayer}{nodelayer}
		\node [style=none] (0) at (-0.2500001, 0.2500001) {};
		\node [style=none] (1) at (0.2500001, 0.2500001) {};
	\end{pgfonlayer}
	\begin{pgfonlayer}{edgelayer}
		\draw [bend right=90, looseness=1.75] (0.center) to (1.center);
	\end{pgfonlayer}
\end{tikzpicture}%} % chktex 27
\]
\[
%\InputIfFileExists{#1.tikz}{}{
\input{./figures/gn-alpha-ZX-to-ZW.tikz}%} % chktex 27
\qquad\qquad

%\InputIfFileExists{#1.tikz}{}{
\input{./figures/hadamard-ZX-to-ZW.tikz}%} % chktex 27
\]
%$$\begin{array}{c}
%\tikzfig{empty-diagram} ~~\mapsto~~ \tikzfig{empty-diagram}\\\\
%\tikzfig{single-line} ~~\mapsto~~ \tikzfig{single-line}\\\\
%\tikzfig{caps} ~~\raisebox{0.3em}{$\mapsto$}~~ \tikzfig{caps}\\\\
%\tikzfig{cup} ~~\raisebox{0.3em}{$\mapsto$}~~ \tikzfig{cup}
%\end{array}
%\qquad\qquad\begin{array}{c}
%\tikzfig{ZW-to-ZX-braid-no-braid}\\\\
%\tikzfig{gn-alpha-ZX-to-ZW}\\\\
%\tikzfig{hadamard-ZX-to-ZW}
%\end{array}$$
\[
%\InputIfFileExists{#1.tikz}{}{
\input{./figures/rn-alpha-2.tikz}%} % chktex 27
\mapsto\left[~
%\InputIfFileExists{#1.tikz}{}{
\begin{tikzpicture}
	\begin{pgfonlayer}{nodelayer}
		\node [style={H box}] (0) at (0, 0) {};
		\node [style=none] (1) at (0, 0.5) {};
		\node [style=none] (2) at (0, -0.5) {};
	\end{pgfonlayer}
	\begin{pgfonlayer}{edgelayer}
		\draw (2.center) to (1.center);
	\end{pgfonlayer}
\end{tikzpicture}
%} % chktex 27
~\right]_W^{\otimes m}\circ \left[
%\InputIfFileExists{#1.tikz}{}{
\input{./figures/gn-alpha-2.tikz}%} % chktex 27
\right]_W\circ \left[~
%\InputIfFileExists{#1.tikz}{}{
\begin{tikzpicture}
	\begin{pgfonlayer}{nodelayer}
		\node [style={H box}] (0) at (0, 0) {};
		\node [style=none] (1) at (0, 0.5) {};
		\node [style=none] (2) at (0, -0.5) {};
	\end{pgfonlayer}
	\begin{pgfonlayer}{edgelayer}
		\draw (2.center) to (1.center);
	\end{pgfonlayer}
\end{tikzpicture}
%} % chktex 27
~\right]_W^{\otimes n}\]
\noindent\begin{minipage}{\columnwidth}
\[D_1\circ D_2\mapsto [D_1]_W\circ[D_2]_W\qquad\quad D_1\otimes D_2\mapsto [D_1]_W\otimes[D_2]_W\]
\vspace{0.2em}
\end{minipage}
\rule{\columnwidth}{0.5pt}

%\begin{rem}
%There does not seem to be black dots involved in this interpretation (apart from a scalar). However, as Hadamard can be expressed in terms of the other generators in ZX, so does \fittext{\tikzfig{zw-cross}} in ZW, as evidenced by its normal form:
%\[\tikzfig{ZW-crossing-normal-form}\]
%\end{rem}

\subsection{From ZW\texorpdfstring{$\!_{\mathbb{C}}$}{\_C} to ZX}

The $[.]_X$ translation %from ZW$_{\mathbb C}$ to ZX
has already been partially defined in Section~\ref{sec:cliff-t}. To extend it to the generalised white spider present in \zwc, the main subtlety is the encoding of positive real numbers % $\lambda\in \mathbb R$
 in the ZX-diagrams.
In~\cite{NgWang}, the authors decompose, roughly speaking, a positive real number  into its integer part and its non-integer part.
Our translation relies on a different (although not unique) decomposition:
\[\forall z\in\mathbb{C},~~ \exists (n,\theta,\beta)\in\mathbb{N}\times \ropen{0;2\pi}\times\left[0;\frac{\pi}{2}\right],\quad z=2^n\cos(\beta)e^{i\theta}\]

\noindent\begin{minipage}{\columnwidth}
\titlerule{$[.]_X$}
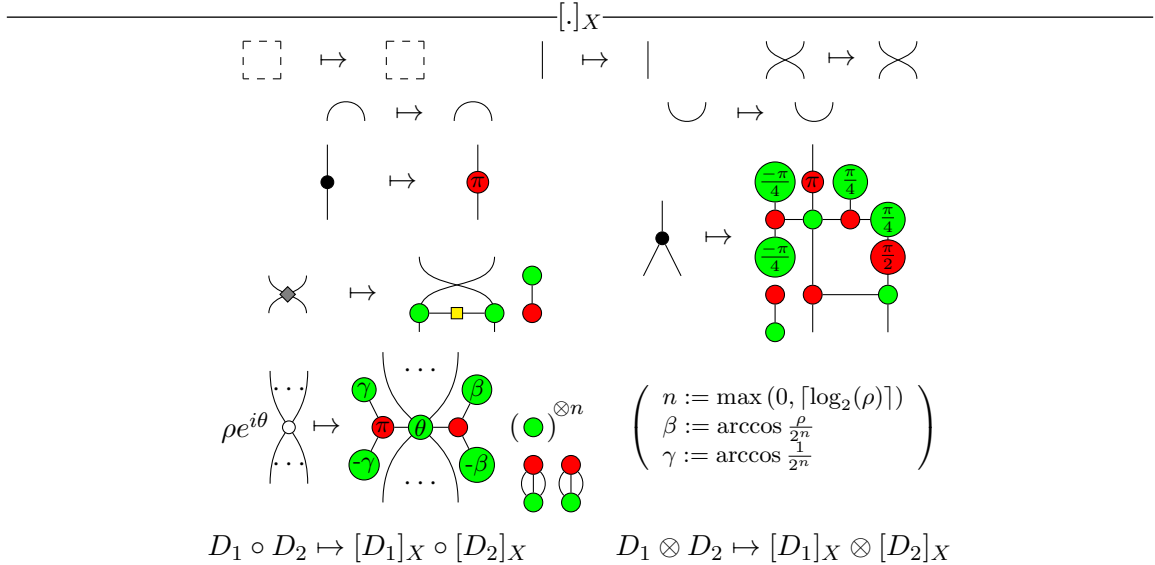
\[
%\InputIfFileExists{#1.tikz}{}{
\input{./figures/empty-diagram.tikz}%} % chktex 27
 \quad\mapsto\quad 
%\InputIfFileExists{#1.tikz}{}{
\input{./figures/empty-diagram.tikz}%} % chktex 27
\qquad\qquad

%\InputIfFileExists{#1.tikz}{}{
\begin{tikzpicture}
	\begin{pgfonlayer}{nodelayer}
		\node [style=none] (0) at (0, 0.2499999) {};
		\node [style=none] (1) at (0, -0.2499999) {};
	\end{pgfonlayer}
	\begin{pgfonlayer}{edgelayer}
		\draw (0.center) to (1.center);
	\end{pgfonlayer}
\end{tikzpicture}%} % chktex 27
 \quad\mapsto\quad 
%\InputIfFileExists{#1.tikz}{}{
\begin{tikzpicture}
	\begin{pgfonlayer}{nodelayer}
		\node [style=none] (0) at (0, 0.2499999) {};
		\node [style=none] (1) at (0, -0.2499999) {};
	\end{pgfonlayer}
	\begin{pgfonlayer}{edgelayer}
		\draw (0.center) to (1.center);
	\end{pgfonlayer}
\end{tikzpicture}%} % chktex 27
\qquad\qquad
%\InputIfFileExists{#1.tikz}{}{
\input{./figures/ZW-to-ZX-braid-no-braid.tikz}%} % chktex 27
\]
\end{minipage}
\[
%\InputIfFileExists{#1.tikz}{}{
\begin{tikzpicture}
	\begin{pgfonlayer}{nodelayer}
		\node [style=none] (0) at (-0.2500001, -0) {};
		\node [style=none] (1) at (0.2500001, -0) {};
	\end{pgfonlayer}
	\begin{pgfonlayer}{edgelayer}
		\draw [bend left=90, looseness=1.75] (0.center) to (1.center);
	\end{pgfonlayer}
\end{tikzpicture}%} % chktex 27
 \quad\raisebox{0.3em}{$\mapsto$}\quad 
%\InputIfFileExists{#1.tikz}{}{
\begin{tikzpicture}
	\begin{pgfonlayer}{nodelayer}
		\node [style=none] (0) at (-0.2500001, -0) {};
		\node [style=none] (1) at (0.2500001, -0) {};
	\end{pgfonlayer}
	\begin{pgfonlayer}{edgelayer}
		\draw [bend left=90, looseness=1.75] (0.center) to (1.center);
	\end{pgfonlayer}
\end{tikzpicture}%} % chktex 27
\hspace{6em}

%\InputIfFileExists{#1.tikz}{}{
\begin{tikzpicture}
	\begin{pgfonlayer}{nodelayer}
		\node [style=none] (0) at (-0.2500001, 0.2500001) {};
		\node [style=none] (1) at (0.2500001, 0.2500001) {};
	\end{pgfonlayer}
	\begin{pgfonlayer}{edgelayer}
		\draw [bend right=90, looseness=1.75] (0.center) to (1.center);
	\end{pgfonlayer}
\end{tikzpicture}%} % chktex 27
 \quad\raisebox{0.3em}{$\mapsto$}\quad 
%\InputIfFileExists{#1.tikz}{}{
\begin{tikzpicture}
	\begin{pgfonlayer}{nodelayer}
		\node [style=none] (0) at (-0.2500001, 0.2500001) {};
		\node [style=none] (1) at (0.2500001, 0.2500001) {};
	\end{pgfonlayer}
	\begin{pgfonlayer}{edgelayer}
		\draw [bend right=90, looseness=1.75] (0.center) to (1.center);
	\end{pgfonlayer}
\end{tikzpicture}%} % chktex 27
\]
\[\begin{array}{c}

%\InputIfFileExists{#1.tikz}{}{
\input{./figures/ZW-to-ZX-dot-1-1.tikz}%} % chktex 27
\\\\

%\InputIfFileExists{#1.tikz}{}{
\input{./figures/ZW-to-ZX-cross-no-braid.tikz}%} % chktex 27

\end{array} \qquad\quad

%\InputIfFileExists{#1.tikz}{}{
\input{./figures/ZW-to-ZX-dot-1-2.tikz}%} % chktex 27
 % chktex 8
\]
\[
%\InputIfFileExists{#1.tikz}{}{
\input{./figures/ZW-white-dot-to-ZX.tikz}%} % chktex 27
\]
\noindent\begin{minipage}{\columnwidth}
\[D_1\circ D_2\mapsto [D_1]_X\circ[D_2]_X\quad\qquad D_1\otimes D_2\mapsto [D_1]_X\otimes[D_2]_X\]
\vspace{0.2em}
\end{minipage}
\rule{\columnwidth}{0.5pt}

\begin{rem}
$n$ is well-defined: Every complex number $x\neq0$ can be expressed as $\rho e^{i\theta}$ where $\rho\in\mathbb{R}^*_+$. If $x=0$, then $n:=0$. However, $\theta$ may take any value, but it makes no difference (see Section~\ref{prf:X-is-homomorphism} in appendix).
\end{rem}

%This interpretation relies on the fact that we can easily create the vector $\begin{pmatrix}1\\ \cos(\alpha)\end{pmatrix}$ with ZX-diagrams since:
%\[\interp{\tikzfig{cos-alpha}}=\frac{1}{\sqrt{2}}\begin{pmatrix}1&0&0&1\\0&1&1&0\end{pmatrix}
%\begin{pmatrix}1\\e^{-i\alpha}\\e^{i\alpha}\\1\end{pmatrix}
%%=\frac{1}{\sqrt{2}}\begin{pmatrix}2\\e^{-i\alpha}+e^{i\alpha}\end{pmatrix}
%=\frac{2}{\sqrt{2}}\begin{pmatrix}1\\\cos(\alpha)\end{pmatrix}\]
%and the scalar $\frac{1}{\sqrt{2}}$ is easy to create.

We may prove the two following propositions:

\begin{prop}%
\label{prop:double-interpretation-equivalence-1}
\[\left(\zxct+\add\right)\vdash D = [[D]_W]_X\]
\end{prop}
Proof in appendix at Section~\ref{prf:double-interp-eq}.

\begin{prop}%
\label{prop:X-is-homomorphism}
\[\zwc\vdash D_1=D_2 \implies \left(\zxct+\add\right)\vdash [D_1]_X=[D_2]_X\]
\end{prop}
Proof in appendix at Section~\ref{prf:X-is-homomorphism}.

The completeness of the calculus is now easy to prove:
\begin{proof}[Proof of Thm.~\ref{thm:gen-zx}]
Let $D_1$ and $D_2$ be two diagrams of the ZX-Calculus such that $\interp{D_1}=\interp{D_2}$. Since $[.]_W$ preserves the the semantics, $\interp{[D_1]_W}=\interp{[D_1]_W}$. By completeness of the \zwc-Calculus, $\zwc\vdash [D_1]_W=[D_2]_W$. By Proposition~\ref{prop:X-is-homomorphism}, $\left(\zxct+\add\right)\vdash [[D_1]_W]_X=[[D_2]_W]_X$. Finally, by Proposition~\ref{prop:double-interpretation-equivalence-1}, $\left(\zxct+\add\right)\vdash D_1=D_2$ which completes the proof.
\end{proof}

\section{Discussion on the New Axioms}%
\label{sec:new-axioms}

\noindent\textbf{
%\InputIfFileExists{#1.tikz}{}{
\input{./figures/commutation-of-controls-general-simplified.tikz}%} % chktex 27
 as a Commutation of Controlled Operations}

A controlled operation is a fairly common concept in quantum circuits. Let $U$ be an $n\to n$ diagram representing a unitary. Then, a larger diagram $n+1\to n+1$, denoted $\Lambda U$, is considered a controlled $U$ if:
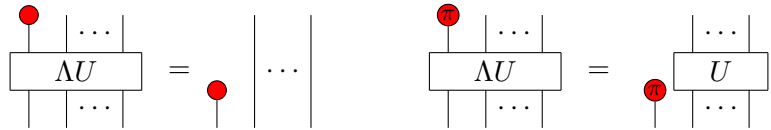
\[
%\InputIfFileExists{#1.tikz}{}{
\input{./figures/controlled-U-ket-0.tikz}%} % chktex 27
\qquad\qquad
%\InputIfFileExists{#1.tikz}{}{
\input{./figures/controlled-U-ket-1.tikz}%} % chktex 27
\]
The leftmost wire in $\Lambda U$ has a particular function: it is called the control wire. When classical data is plugged on this wire, you recover this data at the end of the operation. Plus, if $\ket0$ is plugged, the identity is recovered on the $n$ right wires, whilst if $\ket1$ is plugged, then $U$ is recovered on the right. For instance, the following diagram can be considered as $\Lambda R_Z(2\alpha)$: 
%\InputIfFileExists{#1.tikz}{}{
\input{./figures/controlled-RZ-alpha.tikz}%} % chktex 27
, because
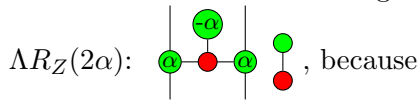
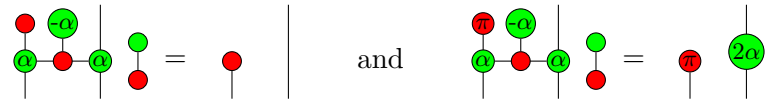
\[
%\InputIfFileExists{#1.tikz}{}{
\input{./figures/controlled-RZ-alpha-ket-0.tikz}%} % chktex 27
\qquad\text{ and }\qquad
%\InputIfFileExists{#1.tikz}{}{
\input{./figures/controlled-RZ-alpha-ket-1.tikz}%} % chktex 27
\]

One can easily reverse the roles of $\ket0$ and $\ket1$ in $\Lambda U$ by adding 
%\InputIfFileExists{#1.tikz}{}{
\begin{tikzpicture}
	\begin{pgfonlayer}{nodelayer}
		\node [style=none] (0) at (0, -0.5) {};
		\node [style=rn] (1) at (0, 0) {$\pi$};
		\node [style=none] (2) at (0, 0.5) {};
	\end{pgfonlayer}
	\begin{pgfonlayer}{edgelayer}
		\draw (2.center) to (0.center);
	\end{pgfonlayer}
\end{tikzpicture}
%} % chktex 27
 before and after the controlled operation on the control wire. The result, denoted $\overline{\Lambda}U$: % chktex 8
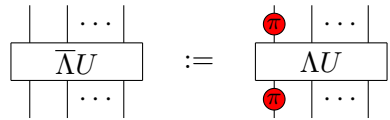
\[
%\InputIfFileExists{#1.tikz}{}{
\input{./figures/anti-controlled-U.tikz}%} % chktex 27
\]
is called an anti-controlled operation. The identity is obtained when $\ket1$ is plugged on the control wire, and $U$ when $\ket0$ is plugged.

Controlled operations haves a nice matrix interpretation: $\interp{\Lambda U} = \begin{pmatrix}I_{2^n} & 0 \\ 0 & \interp{U}\end{pmatrix}$, as well as anti-controlled operations: $\interp{\overline{\Lambda} U} = \begin{pmatrix} \interp{U} & 0 \\ 0 & I_{2^n} \end{pmatrix}$. An interesting property of controlled and anti-controlled operations, that can easily be seen thanks to the interpretation, is that they commute: $\interp{\Lambda U \circ \overline{\Lambda}V} = \interp{\overline{\Lambda}V \circ \Lambda U} = \begin{pmatrix}\interp{V} & 0 \\ 0 & \interp{U}\end{pmatrix}$.

Now, let $U(\beta)=R_Z(2\beta)\otimes \mathbb{I}=~
%\InputIfFileExists{#1.tikz}{}{
\input{./figures/V-for-C.tikz}%} % chktex 27
$ and $V(\alpha,\gamma) =~ 
%\InputIfFileExists{#1.tikz}{}{
\input{./figures/U-for-C.tikz}%} % chktex 27
$. Then we can consider the following diagrams as $\Lambda U(\beta)$ and $\overline{\Lambda}V(\alpha,\gamma)$:
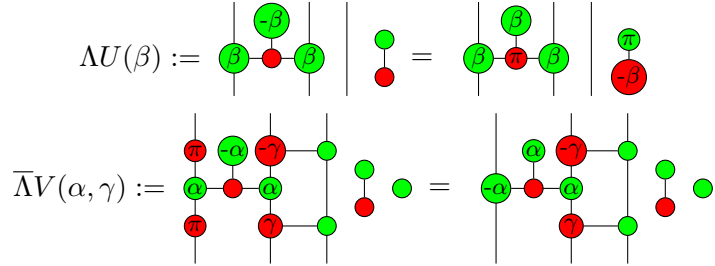
\[\Lambda U(\beta):=~
%\InputIfFileExists{#1.tikz}{}{
\input{./figures/controlled-V-for-C.tikz}%} % chktex 27
\]
\[\overline{\Lambda}V(\alpha,\gamma):=~
%\InputIfFileExists{#1.tikz}{}{
\input{./figures/anti-controlled-U-for-C.tikz}%} % chktex 27
\]
These two diagrams should commute. The rule \ccom expresses this equality in a ZX-style, i.e.~with redundant information being cropped out of the picture. Indeed, if we ignore the scalars --- that are invertible ---, we get on the one hand:
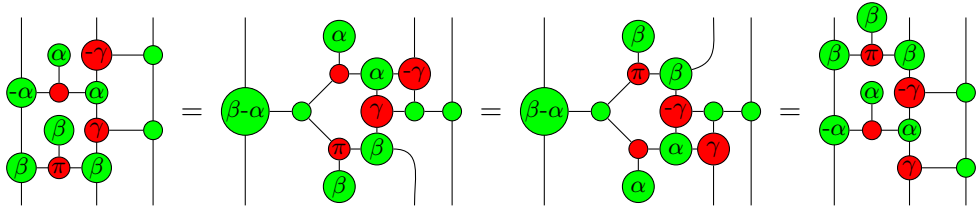
\def\fig{commutation-from-C}
\begin{align*}
\input{./figures/\fig/\fig_00.tikz}
\eq{}\input{./figures/\fig/\fig_01.tikz}
\eq{}\input{./figures/\fig/\fig_02.tikz}
\eq{}\input{./figures/\fig/\fig_03.tikz}
\end{align*}
And on the other hand:
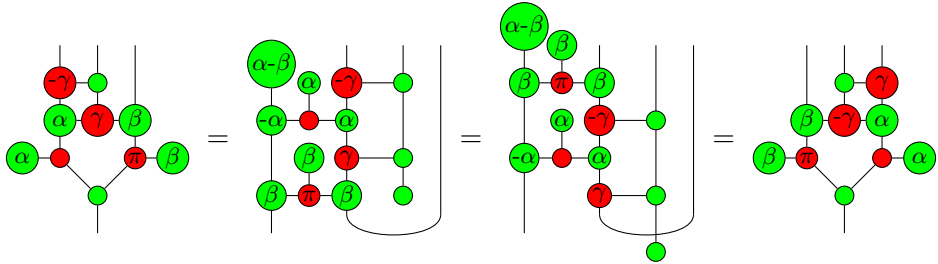
\def\fig{C-from-commutation}
\begin{align*}
\input{./figures/\fig/\fig_00.tikz}
\eq{}\input{./figures/\fig/\fig_01.tikz}
\eq{}\input{./figures/\fig/\fig_02.tikz}
\eq{}\input{./figures/\fig/\fig_03.tikz}
\end{align*}
Hence the two equations are equivalent under the ZX-rules for Clifford.

\noindent\textbf{
%\InputIfFileExists{#1.tikz}{}{
\input{./figures/add-axiom-3.tikz}%} % chktex 27
 as a Sum of Two Controls}

First of all, let us notice that the following pattern, which can be found three times in \add, has interpretation:
\[\interp{
%\InputIfFileExists{#1.tikz}{}{
\input{./figures/control-cos.tikz}%} % chktex 27
}=\sqrt{2}\begin{pmatrix}1\\e^{i\theta}\cos(\alpha)\end{pmatrix}\] % chktex 35
Such a diagram can be seen as a controlled scalar. Up to the global scalar $\sqrt{2}$, if $\bra0$ is plugged at the bottom, we get $1$, but if $\bra1$ is plugged, then we get $e^{i\theta}\cos(\alpha)$. Two of occurrences of this pattern are plugged to the following pattern, with interpretation:
\[\interp{
%\InputIfFileExists{#1.tikz}{}{
\input{./figures/add-part-of-A.tikz}%} % chktex 27
}=\frac{\sqrt{2}\piq{}}{2}\begin{pmatrix}1&0&0&0\\0&\frac{1}{\sqrt{2}}&\frac{1}{\sqrt{2}}&0\end{pmatrix}\]
If two controlled scalars $\begin{pmatrix}1\\x\end{pmatrix}$ and $\begin{pmatrix}1\\y\end{pmatrix}$ are plugged on top, we get:
\[\frac{\sqrt{2}\piq{}}{2}\begin{pmatrix}1&0&0&0\\0&\frac{1}{\sqrt{2}}&\frac{1}{\sqrt{2}}&0\end{pmatrix}\begin{pmatrix}1\\y\\x\\xy\end{pmatrix}=\frac{\sqrt{2}\piq{}}{2}\begin{pmatrix}1\\\frac{1}{\sqrt{2}}(x+y)\end{pmatrix}\]
In our case, $x = e^{i\theta_1}\cos(\alpha)$ and $y = e^{i\theta_2}\cos(\beta)$. For any two such numbers, there exist $\theta_3$ and $\gamma$ such that $e^{i\theta_1}\cos(\alpha)+e^{i\theta_2}\cos(\beta) = 2e^{i\theta_3}\cos(\gamma)$. Hence:
\[\interp{
%\InputIfFileExists{#1.tikz}{}{
\input{./figures/add-axiom-left-up.tikz}%} % chktex 27
}=\sqrt{2}\piq{}\begin{pmatrix}1\\\sqrt{2}e^{i\theta_3}\cos(\gamma)\end{pmatrix}=\piq{}\begin{pmatrix}1&0\\0&\sqrt{2}\end{pmatrix}\sqrt{2}\begin{pmatrix}1\\e^{i\theta_3}\cos(\gamma)\end{pmatrix}\]
The term $\sqrt{2}\begin{pmatrix}1\\e^{i\theta_3}\cos(\gamma)\end{pmatrix}$ can be represented by the first pattern. The term $\piq{}\begin{pmatrix}1&0\\0&\sqrt{2}\end{pmatrix}$, however, can be represented by 
%\InputIfFileExists{#1.tikz}{}{
\input{./figures/normaliser-part-of-A.tikz}%} % chktex 27
.

We end up with:
\[
%\InputIfFileExists{#1.tikz}{}{
\input{./figures/add-axiom-3.tikz}%} % chktex 27
\]
%TODO

\section{Parametrisation of the Triangle}%
\label{sec:param-triangle}

In this section we go further, and make the triangle a generator of the language. However, contrary to~\cite{NgWang,HNW}, where a $\lambda$-box has been introduced to meet the ring structure of \zwc, this property will be achieved by the triangle itself, for it is allowed a parameter:
\[\mathrm{\Delta}(r):1\to1\qquad
%\InputIfFileExists{#1.tikz}{}{
\begin{tikzpicture}
	\begin{pgfonlayer}{nodelayer}
		\node [style=none] (0) at (0, -0.5) {};
		\node [style=uglabel] (1) at (0, 0.25) {$r$};
		\node [style=none] (2) at (0, 0.5) {};
		\node [style=ug] (3) at (0, -0) {};
	\end{pgfonlayer}
	\begin{pgfonlayer}{edgelayer}
		\draw (2.center) to (0.center);
	\end{pgfonlayer}
\end{tikzpicture}%} % chktex 27
\qquad\text{where}\qquad
\forall r\in\mathbb{C},~~\interp{
%\InputIfFileExists{#1.tikz}{}{
\begin{tikzpicture}
	\begin{pgfonlayer}{nodelayer}
		\node [style=none] (0) at (0, -0.5) {};
		\node [style=uglabel] (1) at (0, 0.25) {$r$};
		\node [style=none] (2) at (0, 0.5) {};
		\node [style=ug] (3) at (0, -0) {};
	\end{pgfonlayer}
	\begin{pgfonlayer}{edgelayer}
		\draw (2.center) to (0.center);
	\end{pgfonlayer}
\end{tikzpicture}%} % chktex 27
}=\begin{pmatrix}1&r\\0&1\end{pmatrix}\]
This parametrisation changes the fundamental nature of the ZX-calculus. It leads to a new language in which a phase group and the ring structure coexist. An axiomatisation \zxt is proposed in Figure~\ref{fig:ZX_rules-complete-param-triangles}.
\begin{figure}[!htb]
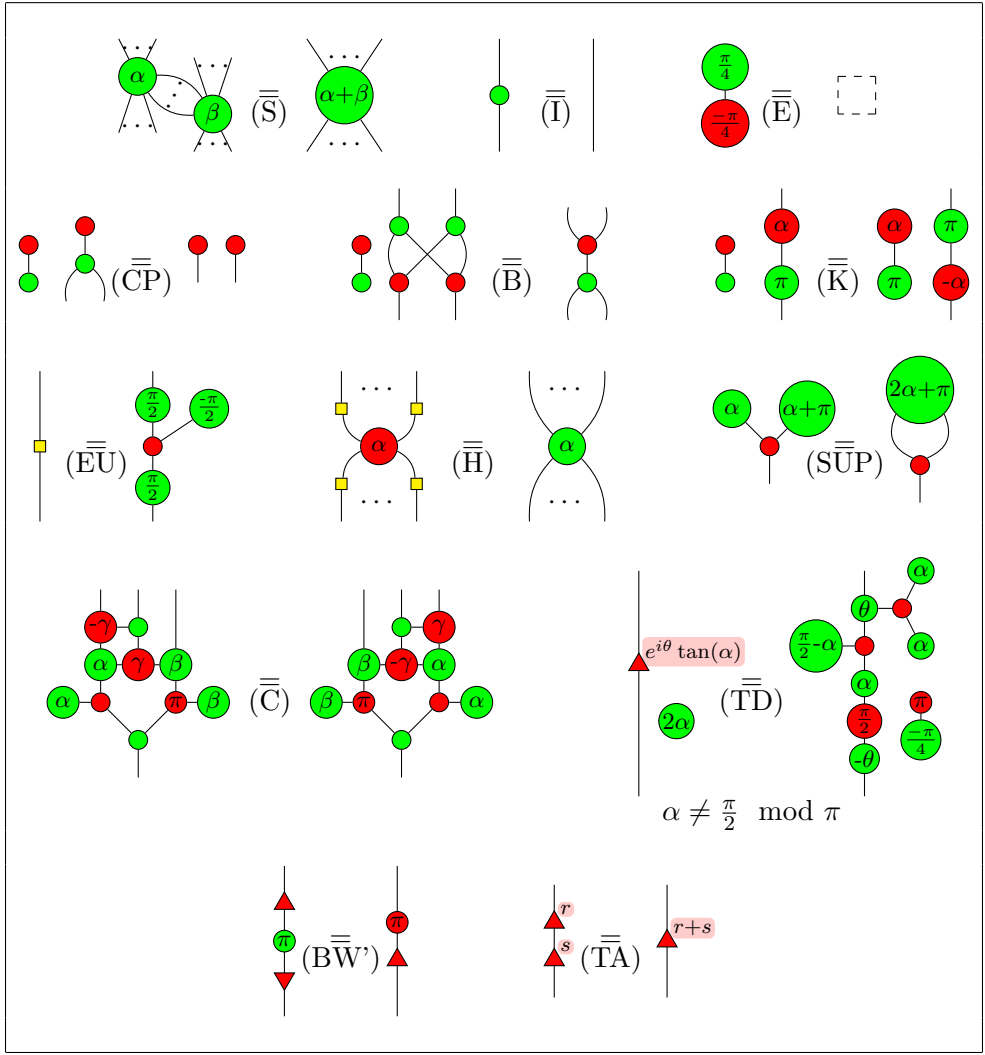

 \centering
 \hypertarget{r:param-triangle-rules}{}
  \begin{tabular}{|c|}
   \hline\\
   
%\InputIfFileExists{#1.tikz}{}{
\input{./figures/spider-1.tikz}%} % chktex 27
$\quad\qquad$
%\InputIfFileExists{#1.tikz}{}{
\input{./figures/s2-simple.tikz}%} % chktex 27
$\quad\qquad$
%\InputIfFileExists{#1.tikz}{}{
\input{./figures/bicolor_pi_4_eq_empty.tikz}%} % chktex 27
\\\\
   
%\InputIfFileExists{#1.tikz}{}{
\input{./figures/b1s.tikz}%} % chktex 27
$\quad\qquad$
%\InputIfFileExists{#1.tikz}{}{
\input{./figures/b2s.tikz}%} % chktex 27
$\quad\qquad$
%\InputIfFileExists{#1.tikz}{}{
\input{./figures/k2s.tikz}%} % chktex 27
\\\\
   
%\InputIfFileExists{#1.tikz}{}{
\input{./figures/euler-decomp-scalar-free.tikz}%} % chktex 27
$\quad\qquad$
%\InputIfFileExists{#1.tikz}{}{
\input{./figures/h2.tikz}%} % chktex 27
$\quad\qquad$
%\InputIfFileExists{#1.tikz}{}{
\input{./figures/former-supp.tikz}%} % chktex 27
\\\\
   
%\InputIfFileExists{#1.tikz}{}{
\input{./figures/commutation-of-controls-general-simplified.tikz}%} % chktex 27
$\qquad\qquad$
%\InputIfFileExists{#1.tikz}{}{
\input{./figures/general-triangle-decomposition-2.tikz}%} % chktex 27
\\\\
   
%\InputIfFileExists{#1.tikz}{}{
\input{./figures/BW-triangle-p.tikz}%} % chktex 27
$\qquad\qquad$
%\InputIfFileExists{#1.tikz}{}{
\input{./figures/triangle-addition.tikz}%} % chktex 27
\\\\
   \hline
  \end{tabular}
 \caption[]{
Set of rules \zxt for the ZX-Calculus with parametrised triangles.
 }%
 \label{fig:ZX_rules-complete-param-triangles}
\end{figure}

We claim that our calculus is simpler than the one presented in~\cite{NgWang}. In this new version of the ZX-Calculus, the interpretation of the parametrised triangle in \zwc $[.]_W$ is:
\[
%\InputIfFileExists{#1.tikz}{}{
\input{./figures/param-triangle-ZX-to-ZW.tikz}%} % chktex 27
\]
while the interpretation in ZX $[.]_X$ of the white dot is simplified:
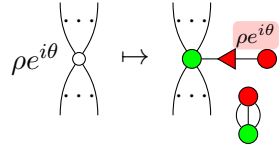
\[
%\InputIfFileExists{#1.tikz}{}{
\input{./figures/ZW-white-dot-to-ZX-param-triangle.tikz}%} % chktex 27
\]
Here again, we define the parametrised green triangle, as well as the triangle with no parameter (which happens to be the same as the one we introduced before):
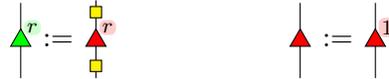
\[
%\InputIfFileExists{#1.tikz}{}{
\input{./figures/param-green-triangle-definition.tikz}%} % chktex 27
\hspace{6em}
%\InputIfFileExists{#1.tikz}{}{
\input{./figures/triangle-1-def.tikz}%} % chktex 27
\]
Indeed, notice that using \td with $\alpha = \frac{\pi}{4}$, we get the initial definition of the triangle, up to a scalar that can be dealt with.

%\begin{figure}[!htb]
% \centering
% \hypertarget{r:param-triangle-rules}{}
%  \begin{tabular}{|c|}
%  \hline\\
%   \tikzfig{general-triangle-decomposition-2}$\qquad\qquad$\tikzfig{not-ug-is-symmetrical}$\qquad\qquad$\tikzfig{triangle-addition}\\\\
%   \hline
%  \end{tabular}
% \caption[]{Set of rules \zxt for the general ZX-Calculus with scalars and parametrised triangles. All of these rules also hold when flipped upside-down, or with the colours red and green swapped.}
% \label{fig:ZX_rules-complete-param-triangles}
%\end{figure}

\begin{prop}%
\label{prop:ZX-param-triangles-contains-ZX-pi_4}
The set of rules for Clifford+T is deducible from the set $\zxt$ in Figure~\ref{fig:ZX_rules-complete-param-triangles}.
\end{prop}
\begin{proof}
All the rules of the Clifford+T fragment except \bw are expressed in the set $\zxt$. Instead, this rule is replace by (BW') in \zxt, but the two are equivalent as shown in Lemma~\ref{lem:not-ug-is-symmetrical}.
\end{proof}

As a corollary, we can prove any equality between two diagrams of the Clifford+T fragment using \zxt. We may now wonder if we can have a theorem for the new language similar to Theorem~\ref{thm:lin-diag}. To do so we need to update our notion of linear diagrams, for now, there can be two kinds of parameters, that we need to distinguish: the ones in green and red nodes, that follow a group structure, and the ones in triangles, that follow a ring structure:

\begin{defi} %Given $C\subseteq [0, 2\pi)$,
A \zxt-diagram is linear in $(\vec{\alpha}, \vec r)$ with constants in $C\subseteq \mathbb R$, if it is generated by $R_Z^{(n,m)}(E_\alpha)$, $R_X^{(n,m)}(E_\alpha)$, $\mathrm{\Delta}(E_r)$, $H$, $e$, $\mathbb I$, $\sigma$, $\epsilon$, $\eta$, and the spacial and sequential compositions, where $n,m\in \mathbb  N$, and:
\begin{itemize}
\item  $E_\alpha$ is of the form $\sum_{i} n_i \alpha_i+c$, with $n_i\in \mathbb Z$ and $c\in C$
\item  $E_r$ is of the form $\sum_{i} n_i r_i+p$, with $n_i\in \mathbb Z$ and $p\in \mathbb{Z}$
\end{itemize}
\end{defi}

It may look strange that $E_r$ has coefficients and constants in $\mathbb{Z}$ since the triangle should be able to hold any complex parameter. This choice actually simplifies the proof of the following theorem. Although, once we have the completeness (Theorem~\ref{thm:completeness-t}), the previous definition can obviously be changed with:
\begin{itemize}
\item  $E_r$ is of the form $\sum_{i} n_i r_i+p$, with $n_i\in \mathbb C$ and $p\in \mathbb{C}$
\end{itemize}

\begin{thm}%
\label{thm:lin-diag-T}
For any diagrams $D_1(\vec\alpha,\vec r)$ and $D_2(\vec\alpha, \vec r)$ linear in $(\vec\alpha, \vec r)$ with constants in $\frac{\pi}{4}\mathbb{Z}$:
\[\left(\forall (\vec\alpha,\vec r)\in \mathbb{R}^{k_\alpha}\times\mathbb{C}^{k_r},~~\interp{D_1(\vec\alpha, \vec r)}=\interp{D_2(\vec\alpha, \vec r)}\right) \equi{} \zxt \vdash D_1(\vec\alpha, \vec r)=D_2(\vec\alpha, \vec r)\]
\end{thm}
To prove this theorem, we will first need a few lemmas (their proofs are in Appendix at page~\pageref{prf:gn-inverse}):

\begin{lem}%
\label{lem:gn-inverse}
The green node 
%\InputIfFileExists{#1.tikz}{}{
\begin{tikzpicture}
	\begin{pgfonlayer}{nodelayer}
		\node [style=gn] (0) at (0, 0) {$\alpha$};
	\end{pgfonlayer}
\end{tikzpicture}
%} % chktex 27
 has an inverse, denoted $
%\InputIfFileExists{#1.tikz}{}{
\begin{tikzpicture}
	\begin{pgfonlayer}{nodelayer}
		\node [style=gn] (0) at (0, 0) {$\alpha$};
	\end{pgfonlayer}
\end{tikzpicture}
%} % chktex 27
^{-1}$, if $\alpha\neq\pi\mod 2\pi$: % chktex 8
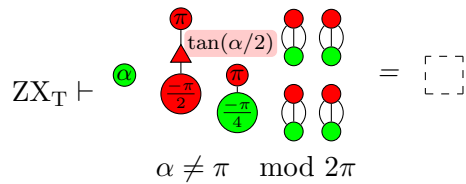
\[\zxt\vdash~~
%\InputIfFileExists{#1.tikz}{}{
\input{./figures/gn-alpha-inverse.tikz}%} % chktex 27
\]
\end{lem}

\begin{lem}%
\label{lem:triangle-times-exp}
Phases can be extracted from the parameter of a triangle:
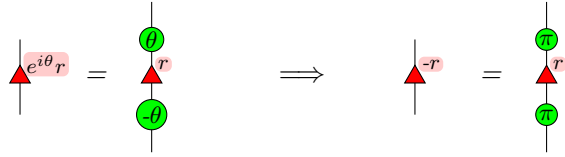
\[
%\InputIfFileExists{#1.tikz}{}{
\input{./figures/param-triangle-times-exp.tikz}%} % chktex 27
\] % chktex 35
\end{lem}

\begin{lem}%
\label{lem:param-triangle-0}
A triangle with parameter $0$ is the identity:
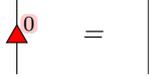
\[
%\InputIfFileExists{#1.tikz}{}{
\input{./figures/param-triangle-0.tikz}%} % chktex 27
\]
\end{lem}

\begin{proof}[Proof of Theorem~\ref{thm:lin-diag-T}]
The idea of the proof is to use \td to recover a diagram of the ZX-Calculus with no triangle, and use Theorem~\ref{thm:lin-diag}.

Let $D_1(\vec\alpha, \vec r)$ and $D_2(\vec\alpha, \vec r)$ be two linear diagrams of \zxt, such that $\interp{D_1(\vec\alpha, \vec r)}=\interp{D_2(\vec\alpha, \vec r)}$.
For any parametrised triangle, in both diagrams, we first develop its expression:
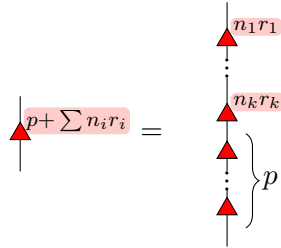
\[
%\InputIfFileExists{#1.tikz}{}{
\input{./figures/lin-param-triangle-decomp.tikz}%} % chktex 27
\]
Then, we decompose each triangle with parameter $n_j r_j$, distinguishing according to the sign of $n_j$, using \ta,~\ref{lem:param-triangle-0} and~\ref{lem:triangle-times-exp}:
\[\begin{array}{c@{\qquad\qquad}c@{\qquad\qquad}c}
n_j = 0: & n_j>0: & n_j < 0:\\

%\InputIfFileExists{#1.tikz}{}{
\input{./figures/param-triangle-0.tikz}%} % chktex 27
 & 
%\InputIfFileExists{#1.tikz}{}{
\input{./figures/lin-param-triangle-decomp-pos.tikz}%} % chktex 27
 & 
%\InputIfFileExists{#1.tikz}{}{
\input{./figures/lin-param-triangle-decomp-neg.tikz}%} % chktex 27

\end{array}\]
%Then, $\forall r_j\in\mathbb{C},~\exists !(\theta_j,\beta_j) \in [0,2\pi[\times [0,\frac{\pi}{2}[,~r_j = e^{i\theta_j}\tan(\beta_j)$. Hence:
We have:
\begin{align*}
%&\left(\forall (\vec\alpha,\vec r) \in \mathbb{R}^\ell\times\mathbb{C}^k, \interp{D_1(\vec\alpha,\vec r)} = \interp{D_2(\vec\alpha,\vec r)} \right)\\
%&\iff \left(
\begin{array}{l}
\forall (\vec\alpha,\vec \theta,\vec \beta) \in \mathbb{R}^\ell\times [0,2\pi[^k\times [0,\frac{\pi}{2}[^k,\\
\quad\interp{D_1\left(\vec\alpha, \left[e^{i\theta_j}\tan(\beta_j)\right]_j\right)} = \interp{D_2\left(\vec\alpha, \left[e^{i\theta_j}\tan(\beta_j)\right]_j\right)}
\end{array}
%\right)
\end{align*}
where $\left[e^{i\theta_j}\tan(\beta_j)\right]_j=e^{i\theta_1}\tan(\beta_1),\ldots,e^{i\theta_k}\tan(\beta_k)$.
In order to use \td, we need to add enough 
%\InputIfFileExists{#1.tikz}{}{
\begin{tikzpicture}
	\begin{pgfonlayer}{nodelayer}
		\node [style=gn] (0) at (0, 0) {$2\beta_j$};
	\end{pgfonlayer}
\end{tikzpicture}
%} % chktex 27
 to the diagrams. We can easily count the number of occurrences of $r_j$ in $D_1$ and $D_2$, and we define $m_j$ to be the maximum of these two quantities. We add $m_j$ occurrences of 
%\InputIfFileExists{#1.tikz}{}{
\begin{tikzpicture}
	\begin{pgfonlayer}{nodelayer}
		\node [style=gn] (0) at (0, 0) {$2\beta_j$};
	\end{pgfonlayer}
\end{tikzpicture}
%} % chktex 27
 to both diagrams: $D_i\otimes\left[ \bigotimes\limits_j\left(
%\InputIfFileExists{#1.tikz}{}{
\begin{tikzpicture}
	\begin{pgfonlayer}{nodelayer}
		\node [style=gn] (0) at (0, 0) {$2\beta_j$};
	\end{pgfonlayer}
\end{tikzpicture}
%} % chktex 27
\right)^{\otimes m_j}\right]$. Notice that the equation remains sound: % chktex 8
\begin{align*}
&\forall (\vec\alpha,\vec \theta,\vec \beta) \in \mathbb{R}^\ell\times \ropen{0,2\pi}^k\times \ropen{0,\frac{\pi}{2}}^k,\\
&\!\interp{D_1\left(\vec\alpha, \left[e^{i\theta_j}\tan(\beta_j)\right]_j\right)\otimes\left[ \bigotimes\limits_j
%\InputIfFileExists{#1.tikz}{}{
\begin{tikzpicture}
	\begin{pgfonlayer}{nodelayer}
		\node [style=gn] (0) at (0, 0) {$2\beta_j$};
	\end{pgfonlayer}
\end{tikzpicture}
%} % chktex 27
^{\otimes m_j}\right]} \hspace{-.3ex}=\hspace{-.3ex} \interp{D_2\left(\vec\alpha, \left[e^{i\theta_j}\tan(\beta_j)\right]_j\right)\otimes\left[ \bigotimes\limits_j
%\InputIfFileExists{#1.tikz}{}{
\begin{tikzpicture}
	\begin{pgfonlayer}{nodelayer}
		\node [style=gn] (0) at (0, 0) {$2\beta_j$};
	\end{pgfonlayer}
\end{tikzpicture}
%} % chktex 27
^{\otimes m_j}\right]}
\end{align*}
Now, there is enough occurrences of 
%\InputIfFileExists{#1.tikz}{}{
\begin{tikzpicture}
	\begin{pgfonlayer}{nodelayer}
		\node [style=gn] (0) at (0, 0) {$2\beta_j$};
	\end{pgfonlayer}
\end{tikzpicture}
%} % chktex 27
 in both diagrams to use \td on all the triangles. We denote $d_i$ the resulting diagrams. These diagrams are triangle free, hence can be assimilated to diagrams of the canonical ZX-Calculus, and are linear in $(\vec\alpha,\vec\theta,\vec\beta)$. Since % chktex 8
\[\forall (\vec\alpha,\vec \theta,\vec \beta) \in \mathbb{R}^\ell\times \ropen{0,2\pi}^k\times \ropen{0,\frac{\pi}{2}}^k,\interp{d_1(\vec\alpha,\vec\theta,\vec\beta)} = \interp{d_2(\vec\alpha,\vec\theta,\vec\beta)}\]
using Theorem~\ref{thm:valuations}, $\forall(\vec\alpha,\vec\theta,\vec\beta)\in\mathbb{R}^{\ell+2k},~\zxct \vdash d_1(\vec\alpha,\vec\theta,\vec\beta)=d_2(\vec\alpha,\vec\theta,\vec\beta)$. However, thanks to Proposition~\ref{prop:ZX-param-triangles-contains-ZX-pi_4}, $\forall(\vec\alpha,\vec\theta,\vec\beta)\in\mathbb{R}^{\ell+2k},~\zxt\vdash d_1(\vec\alpha,\vec\theta,\vec\beta)=d_2(\vec\alpha,\vec\theta,\vec\beta)$. Finally, for all $r\in\mathbb{C}$, there exists $(\theta,\beta) \in \ropen{0,2\pi}\times \ropen{0,\frac{\pi}{2}}$ such that $r=e^{i\theta}\tan(\beta)$. Hence,
thanks to Lemma~\ref{lem:gn-inverse}, %$\zxt\vdash d_1(\vec\alpha,\vec\theta,\vec\beta)=d_2(\vec\alpha,\vec\theta,\vec\beta)\iff\zxt\vdash D_1(\vec\alpha, \vec r)=D_2(\vec\alpha, \vec r)$.
$\forall(\vec\alpha,\vec r)\in\mathbb{R}^\ell\times\mathbb{C}^k,~\zxt\vdash D_1(\vec\alpha, \vec r)=D_2(\vec\alpha, \vec r)$.
\end{proof}

Again, this theorem proves to be very useful for proving the derivability of some equations, for instance:
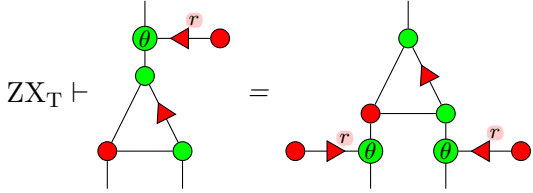
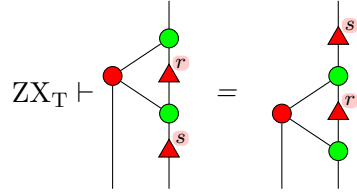
\begin{multicols}{2}
\begin{lem}%
\label{lem:param-tri-distri}
\[\zxt\vdash
%\InputIfFileExists{#1.tikz}{}{
\input{./figures/add-axiom-2-param-triangle.tikz}%} % chktex 27
\]
\end{lem}
\begin{lem}%
\label{lem:C2-gen}
\[\zxt\vdash
%\InputIfFileExists{#1.tikz}{}{
\input{./figures/param-triangle-and-param-not-W-commute.tikz}%} % chktex 27
\]
\end{lem}
\end{multicols}
The last two lemmas are direct applications of Theorem~\ref{thm:lin-diag-T} while the following requires a bit more work.
\begin{lem}%
\label{lem:parallel-param-triangles}
\[\forall r,s\in\mathbb{C},~~\zxt\vdash~~
%\InputIfFileExists{#1.tikz}{}{
\input{./figures/parallel-param-triangles.tikz}%} % chktex 27
\]
\end{lem}
The proof is in Appendix at page~\pageref{prf:parallel-param-triangles}.

After this, it is not hard to show:
\begin{thm}%
\label{thm:completeness-t}
The set of rules \zxt (in~\ref{fig:ZX_rules-complete-param-triangles}) is complete for quantum mechanics.
\end{thm}
The proof is in Appendix at page~\pageref{prf:completeness-t}.

%\bibliography{../quantum-ref}
\newcommand{\etalchar}[1]{$^{#1}$}

\appendix

\section{Lemmas from \texorpdfstring{$\zxct$}{ZX(pi/4)}} % chktex 36

In this appendix (\ref{prf:rules-preserved},~\ref{prf:left-inverse-ZX}) are the proofs of Propositions~\ref{prop:rules-preserved} and~\ref{prop:left-inverse-ZX}. To simplify the following work, we use the new node introduced as a notation in Section~\ref{sec:cliff-t}, and give a few lemmas in Section~\ref{app:lemmas}, and prove them in Section~\ref{app:proof-of-lemmas}. We try to specify at each step which axiom or previously proven lemma has been used. The spider rule \s is so natural in ZX and used so many times that we decided not to list it when we use it. Keep in mind that for any provable equation, its upside down version, its colour-swapped version, and (after Lemma~\ref{lem:-1-interp}) its version with opposed angles are all provable.

\subsection{Lemmas}%
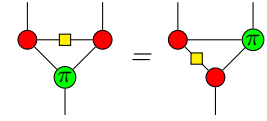
\label{app:lemmas}
\begin{multicols}{3}
\begin{lem}%
\label{lem:2-is-sqrt2-squared}
\[
%\InputIfFileExists{#1.tikz}{}{
\input{./figures/2-is-sqrt2-squared.tikz}%} % chktex 27
\]
\end{lem}

\begin{lem}%
\label{lem:hopf}
\[
%\InputIfFileExists{#1.tikz}{}{
\input{./figures/hopf.tikz}%} % chktex 27
\]
\end{lem}

\begin{lem}%
\label{lem:k1}
\[
%\InputIfFileExists{#1.tikz}{}{
\input{./figures/k1.tikz}%} % chktex 27
\]
\end{lem}

\begin{lem}%
\label{lem:inverse}
\[
%\InputIfFileExists{#1.tikz}{}{
\input{./figures/inverse.tikz}%} % chktex 27
\]
\end{lem}

\begin{lem}%
\label{lem:multiplying-global-phases}
\[
%\InputIfFileExists{#1.tikz}{}{
\input{./figures/multiplying-global-phases.tikz}%} % chktex 27
\]
\end{lem}

\begin{lem}%
\label{lem:bicolor-0-alpha}
\[
%\InputIfFileExists{#1.tikz}{}{
\input{./figures/bicolor-0-alpha.tikz}%} % chktex 27
\]
\end{lem}

\begin{lem}%
\label{lem:hadamard-involution}
\[
%\InputIfFileExists{#1.tikz}{}{
\input{./figures/hadamard-involution.tikz}%} % chktex 27
\]
\end{lem}

\begin{lem}%
\label{lem:h-loop}
\[
%\InputIfFileExists{#1.tikz}{}{
\input{./figures/h-loop.tikz}%} % chktex 27
\]
\end{lem}

\begin{lem}%
\label{lem:red-pi_2-around-green-node}
\[
%\InputIfFileExists{#1.tikz}{}{
\input{./figures/lemma-red-pi_2-around-green-node.tikz}%} % chktex 27
\]
\end{lem}

\begin{lem}%
\label{lem:6-cycle-bialgebra}
\[{
%\InputIfFileExists{#1.tikz}{}{
\input{./figures/6-cycle-bialgebra.tikz}%} % chktex 27
}\]
\end{lem}

\begin{lem}%
\label{lem:control-pi-and-anti-CNOT-commute}
\[
%\InputIfFileExists{#1.tikz}{}{
\input{./figures/control-pi-and-anti-CNOT-commute-simp.tikz}%} % chktex 27
\]
\end{lem}

\end{multicols}

\begin{lem}%
\label{lem:-1-interp}
Let $\interp{.}_{-1}$ be the interpretation that multiplies all the angles by $-1$. Then:
\[ \zxct\vdash D_1=D_2 ~~\iff~~ \zxct\vdash \interp{D_1}_{-1}=\interp{D_2}_{-1} \]
\end{lem}

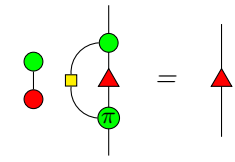
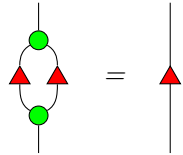
\begin{multicols}{3}

\begin{lem}%
\label{lem:gn-pi_2-0-0-equals-sqrt2-exp-pi_4}
\[
%\InputIfFileExists{#1.tikz}{}{
\begin{tikzpicture}
	\begin{pgfonlayer}{nodelayer}
		\node [style=gn] (0) at (-1, -0) {$\frac{\pi}{2}$};
		\node [style=none] (1) at (0, -0) {=};
		\node [style=gn] (2) at (1, -0.2500001) {$\frac{\pi}{4}$};
		\node [style=rn] (3) at (1, 0.5) {$\pi$};
	\end{pgfonlayer}
	\begin{pgfonlayer}{edgelayer}
		\draw (3) to (2);
	\end{pgfonlayer}
\end{tikzpicture}%} % chktex 27
\] % chktex 35
\end{lem}

\begin{lem}%
\label{lem:green-state-pi_2-is-red-state-minus-pi_2}
\[
%\InputIfFileExists{#1.tikz}{}{
\input{./figures/green-state-pi_2-is-red-state-minus-pi_2.tikz}%} % chktex 27
\]
\end{lem}

\begin{lem}%
\label{lem:euler-decomp-with-scalar}
\[
%\InputIfFileExists{#1.tikz}{}{
\input{./figures/euler-decomp-with-scalar.tikz}%} % chktex 27
\]
\end{lem}

\begin{lem}%
\label{lem:C1-original}
\[
%\InputIfFileExists{#1.tikz}{}{
\input{./figures/control-commutation-2.tikz}%} % chktex 27
\]
\end{lem}

\begin{lem}%
\label{lem:C1-bis}
\[
%\InputIfFileExists{#1.tikz}{}{
\input{./figures/control-commutation-3.tikz}%} % chktex 27
\]
\end{lem}

%\begin{lem}
%\label{lem:new-supp}
%\[\tikzfig{new-supplementarity}\qquad \text{ for }\alpha,\beta\in \frac{\pi}{4}\mathbb{Z}\]
%\end{lem}

\begin{lem}%
\label{lem:supp-to-minus-pi_4}
\[
%\InputIfFileExists{#1.tikz}{}{
\input{./figures/supp-to-minus-pi_4.tikz}%} % chktex 27
\]
\end{lem}

\begin{lem}%
\label{lem:triangle-decomp-2}
\[
%\InputIfFileExists{#1.tikz}{}{
\input{./figures/triangle-decomp-2.tikz}%} % chktex 27
\]
\end{lem}

\begin{lem}%
\label{lem:not-triangle-is-symmetrical}
\[
%\InputIfFileExists{#1.tikz}{}{
\input{./figures/not-ug-is-symmetrical-2.tikz}%} % chktex 27
\]
\end{lem}

\begin{lem}%
\label{lem:red-state-on-triangle}
\[
%\InputIfFileExists{#1.tikz}{}{
\input{./figures/red-state-on-triangle.tikz}%} % chktex 27
\]
\end{lem}

\begin{lem}%
\label{lem:pi-red-state-on-triangle}
\[
%\InputIfFileExists{#1.tikz}{}{
\input{./figures/pi-red-state-on-triangle.tikz}%} % chktex 27
\]
\end{lem}

\begin{lem}%
\label{lem:red-state-on-upside-down-triangle}
\[
%\InputIfFileExists{#1.tikz}{}{
\input{./figures/red-state-on-upside-down-triangle.tikz}%} % chktex 27
\]
\end{lem}

\begin{lem}%
\label{lem:pi-red-state-on-upside-down-triangle}
\[
%\InputIfFileExists{#1.tikz}{}{
\input{./figures/pi-red-state-on-upside-down-triangle.tikz}%} % chktex 27
\]
\end{lem}

\begin{lem}%
\label{lem:pi-green-state-on-upside-down-triangle}
\[
%\InputIfFileExists{#1.tikz}{}{
\input{./figures/pi-green-state-on-upside-down-triangle.tikz}%} % chktex 27
\]
\end{lem}

\begin{lem}%
\label{lem:pi-green-state-on-triangle}
\[
%\InputIfFileExists{#1.tikz}{}{
\input{./figures/pi-green-state-on-triangle.tikz}%} % chktex 27
\]
\end{lem}

\begin{lem}%
\label{lem:triangle-hadamard-parallel}
\[
%\InputIfFileExists{#1.tikz}{}{
\input{./figures/anti-control-pi-control-green.tikz}%} % chktex 27
\]
\end{lem}

\begin{lem}%
\label{lem:looped-triangle}
\[
%\InputIfFileExists{#1.tikz}{}{
\input{./figures/looped-triangle.tikz}%} % chktex 27
\]
\end{lem}

\begin{lem}%
\label{lem:black-dot-swappable-outputs}
\[
%\InputIfFileExists{#1.tikz}{}{
\input{./figures/lemma-black-dot-swappable-outputs.tikz}%} % chktex 27
\]
\end{lem}

\begin{lem}%
\label{lem:not-ug-is-symmetrical}
\[
%\InputIfFileExists{#1.tikz}{}{
\input{./figures/BW-triangle-no-label.tikz}%} % chktex 27
\]
\end{lem}

\begin{lem}%
\label{lem:inverse-of-triangle}
\[
%\InputIfFileExists{#1.tikz}{}{
\input{./figures/inverse-of-triangle.tikz}%} % chktex 27
\]
\end{lem}

\begin{lem}%
\label{lem:symmetric-diagram-with-triangle-hadamard}
\[
%\InputIfFileExists{#1.tikz}{}{
\input{./figures/lemma-symmetric-diagram-with-triangle-hadamard.tikz}%} % chktex 27
\]
\end{lem}

\begin{lem}%
\label{lem:anti-control-hanging-branch-and-control-alpha-commute}
\[\fit{
%\InputIfFileExists{#1.tikz}{}{
\input{./figures/anti-control-hanging-branch-and-control-alpha-commute.tikz}%} % chktex 27
}\]
\end{lem}

\begin{lem}%
\label{lem:n-four}
\[
%\InputIfFileExists{#1.tikz}{}{
\input{./figures/ug-and-not-W-commute-simplified.tikz}%} % chktex 27
\]
\end{lem}

\begin{lem}%
\label{lem:parallel-triangles}
\[
%\InputIfFileExists{#1.tikz}{}{
\input{./figures/parallel-ugs-are-projections.tikz}%} % chktex 27
\]
\end{lem}

\begin{lem}%
\label{lem:triangles-fork-absorbs-anti-CNOT}
\[
%\InputIfFileExists{#1.tikz}{}{
\input{./figures/ug-fork-absorbs-anti-CNOT.tikz}%} % chktex 27
\] and \[
%\InputIfFileExists{#1.tikz}{}{
\input{./figures/ug-fork-absorbs-CNOT.tikz}%} % chktex 27
\]
\end{lem}

\begin{lem}%
\label{lem:n-five}
\[
%\InputIfFileExists{#1.tikz}{}{
\input{./figures/2-diagrams-of-control-triangle-axiom-simplified.tikz}%} % chktex 27
\]
\end{lem}
\vfill\null%
\end{multicols}

\subsection{Proof of Lemmas}%
\label{app:proof-of-lemmas}

\begin{proof}[Proof of Lemma~\ref{lem:2-is-sqrt2-squared}]
\def\fig{2-is-sqrt-2-squared-proof}
\begin{align*}
\input{./figures/\fig/\fig_00.tikz}
\eq{\id}\input{./figures/\fig/\fig_01.tikz}
\eq{\cp}\input{./figures/\fig/\fig_02.tikz}
\eq{}\input{./figures/\fig/\fig_03.tikz}
\eq{\h}\input{./figures/\fig/\fig_04.tikz}
\end{align*}
\end{proof}

\begin{proof}[Proof of Lemmas~\ref{lem:hopf} %,~\ref{lem:hopf},~\ref{lem:k1},~\ref{lem:inverse},~\ref{lem:multiplying-global-phases},~\ref{lem:bicolor-0-alpha},~\ref{lem:hadamard-involution},~\ref{lem:h-loop},~\ref{lem:red-pi_2-around-green-node},~\ref{lem:6-cycle-bialgebra} and~\ref{lem:control-pi-and-anti-CNOT-commute}]$~$\\
 to~\ref{lem:control-pi-and-anti-CNOT-commute}]$~$\\
Lemmas~\ref{lem:multiplying-global-phases} and~\ref{lem:bicolor-0-alpha} are proven in~\cite{simplified-stabilizer,cyclo}. The other lemmas are in the \frag2 and hence are derivable by completeness of this fragment, since the axiomatisation from~\cite{towards-minimal} can easily be recovered from Lemma~\ref{lem:2-is-sqrt2-squared}.
\end{proof}

\begin{proof}[Proof of Lemma~\ref{lem:-1-interp}]
The result is quite obvious for all rules except maybe for \e, \hd and \bw.
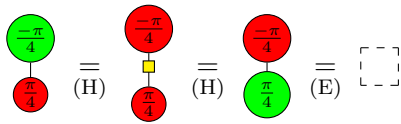
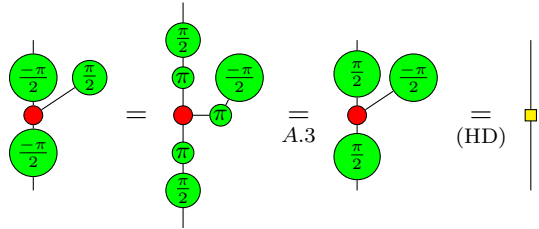
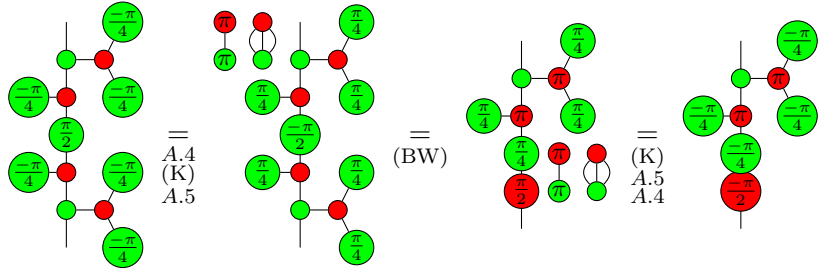
\begin{itemize}
\item \e:
\def\fig{bicolor_pi_4_eq_empty-minus-1}
\begin{align*}
\input{./figures/\fig/\fig_00.tikz}\eq{\h}\input{./figures/\fig/\fig_01.tikz}\eq{\h}
\input{./figures/\fig/\fig_02.tikz}\eq{\e}\input{./figures/\fig/\fig_03.tikz}
\end{align*}
\item \hd:
\def\fig{euler-decomp-scalar-free-minus-1}
\begin{align*}
\input{./figures/\fig/\fig_00.tikz}\eq{}\input{./figures/\fig/\fig_01.tikz}\eq{\ref{lem:k1}}
\input{./figures/\fig/\fig_02.tikz}\eq{\hd}\input{./figures/\fig/\fig_03.tikz}
\end{align*}
\item \bw:
\def\fig{2-diagrams-of-control-triangle-axiom-minus-1-from-simplified}
\begin{align*}
\input{./figures/\fig/\fig_00.tikz}
\eq{\ref{lem:inverse}\\\picom\\\ref{lem:multiplying-global-phases}}\input{./figures/\fig/\fig_01.tikz}
\eq{\bw}\input{./figures/\fig/\fig_02.tikz}
\eq{\picom\\\ref{lem:multiplying-global-phases}\\\ref{lem:inverse}}\input{./figures/\fig/\fig_03.tikz}
\end{align*}
\end{itemize}
Moreover, it is to be noticed that $\interp{
%\InputIfFileExists{#1.tikz}{}{
\begin{tikzpicture}
	\begin{pgfonlayer}{nodelayer}
		\node [style=ug] (0) at (0, -0) {};
		\node [style=none] (1) at (0, 0.5000001) {};
		\node [style=none] (2) at (0, -0.5000001) {};
		\node [style=none] (3) at (0, -0.7499999) {};
		\node [style=none] (4) at (0, 0.7500001) {};
	\end{pgfonlayer}
	\begin{pgfonlayer}{edgelayer}
		\draw (1.center) to (2.center);
	\end{pgfonlayer}
\end{tikzpicture}%} % chktex 27
}_{-1} = 
%\InputIfFileExists{#1.tikz}{}{
\begin{tikzpicture}
	\begin{pgfonlayer}{nodelayer}
		\node [style=ug] (0) at (0, -0) {};
		\node [style=none] (1) at (0, 0.5000001) {};
		\node [style=none] (2) at (0, -0.5000001) {};
		\node [style=none] (3) at (0, -0.7499999) {};
		\node [style=none] (4) at (0, 0.7500001) {};
	\end{pgfonlayer}
	\begin{pgfonlayer}{edgelayer}
		\draw (1.center) to (2.center);
	\end{pgfonlayer}
\end{tikzpicture}%} % chktex 27
$.
\end{proof}

%\begin{proof}[Proof of Lemma~\ref{lem:6-cycle-bialgebra}]
%Using~\ref{lem:hopf} and \bt:
%\def\fig{6-cycle-bialgebra-proof}
%\begin{align*}
%\tikzfigc{00}=\tikzfigc{01}=\tikzfigc{02}=\tikzfigc{03}
%\end{align*}
%\end{proof}

\begin{proof}[Proof of Lemma~\ref{lem:gn-pi_2-0-0-equals-sqrt2-exp-pi_4}]
\def\fig{gn-pi_2-0-0-equals-sqrt2-exp-pi_4-proof} % chktex 8
\begin{align*}
\input{./figures/\fig/\fig_00.tikz}
\eq[~]{\e}\input{./figures/\fig/\fig_01.tikz}
\eq{\h}\input{./figures/\fig/\fig_02.tikz}
\eq{\hd}\input{./figures/\fig/\fig_03.tikz}
\eq{\picom\\\ref{lem:multiplying-global-phases}}\input{./figures/\fig/\fig_04.tikz}
\eq{\supp\\\ref{lem:hopf}}\input{./figures/\fig/\fig_05.tikz}
\eq[~]{\ref{lem:multiplying-global-phases}\\\ref{lem:inverse}}\input{./figures/\fig/\fig_06.tikz}
\end{align*}
\end{proof}

\begin{proof}[Proof of Lemma~\ref{lem:green-state-pi_2-is-red-state-minus-pi_2}]
\def\fig{green-state-pi_2-is-red-state-minus-pi_2-proof}
\begin{align*}
\input{./figures/\fig/\fig_00.tikz}\eq{\h}\input{./figures/\fig/\fig_01.tikz}\eq{\hd}\input{./figures/\fig/\fig_02.tikz}\eq{\cp}\input{./figures/\fig/\fig_03.tikz}\eq{\h\\\ref{lem:gn-pi_2-0-0-equals-sqrt2-exp-pi_4}}\input{./figures/\fig/\fig_04.tikz}
\end{align*}
\end{proof}

\begin{proof}[Proof of Lemma~\ref{lem:euler-decomp-with-scalar}]
\def\fig{euler-decomp-with-scalar-proof}
\begin{align*}
\input{./figures/\fig/\fig_00.tikz}
\eq[]{\hd}\input{./figures/\fig/\fig_01.tikz}
\eq[]{\h\\\ref{lem:inverse}\\\ref{lem:multiplying-global-phases}}\input{./figures/\fig/\fig_02.tikz}
\eq[]{\ref{lem:green-state-pi_2-is-red-state-minus-pi_2}}\input{./figures/\fig/\fig_03.tikz}
\eq[]{\h}\input{./figures/\fig/\fig_04.tikz}
\end{align*}
\end{proof}

\begin{proof}[Proof of Lemma~\ref{lem:C1-original}]
\def\fig{control-commutation-2-proof}
\begin{align*}
\input{./figures/\fig/\fig_00.tikz}
\eq{\ref{lem:euler-decomp-with-scalar}\\\ref{lem:multiplying-global-phases}\\\ref{lem:inverse}}\input{./figures/\fig/\fig_01.tikz}
\eq{\cp}\input{./figures/\fig/\fig_02.tikz}
\eq[~]{\ccom}\input{./figures/\fig/\fig_03.tikz}
\eq[~]{\cp}\input{./figures/\fig/\fig_04.tikz}
\eq{\ref{lem:inverse}\\\ref{lem:multiplying-global-phases}\\\ref{lem:euler-decomp-with-scalar}}\input{./figures/\fig/\fig_05.tikz}
\end{align*}
\end{proof}

\begin{proof}[Proof of Lemma~\ref{lem:C1-bis}]
By completeness of the \frag2:
\[{
%\InputIfFileExists{#1.tikz}{}{
\input{./figures/lemma-for-equivalence-c1.tikz}%} % chktex 27
}\]
Then:
\def\fig{control-commutation-3-proof}
\begin{align*}
\input{./figures/\fig/\fig_00.tikz}
\eq{}\input{./figures/\fig/\fig_01.tikz}
\eq{\ref{lem:C1-original}}\input{./figures/\fig/\fig_02.tikz}
\eq{\ref{lem:hadamard-involution}}\input{./figures/\fig/\fig_03.tikz}
\eq{\h\\\ref{lem:2-is-sqrt2-squared}\\\ref{lem:hopf}}\input{./figures/\fig/\fig_04.tikz}
\end{align*}
\end{proof}

%\begin{proof}[Proof of Lemma~\ref{lem:new-supp}]
%First, if $\alpha$ or $\beta$ is $0\mod \pi$, the result is obvious, using \bo and~\ref{lem:k1}. Then, if say, $\beta=\pi/2$, using \so, \picom, \hd, \supp,~\ref{lem:hopf}:
%\def\fig{new-supp-proof-pi_2}
%\begin{align*}
%\tikzfigc{00}=\tikzfigc{01}=\tikzfigc{02}=\tikzfigc{03}\\=\tikzfigc{04}=\tikzfigc{05}=\tikzfigc{06}
%\end{align*}
%If $\beta=-\pi/2$, the proof is similar. Now, the remaining cases are when $\alpha$ and $\beta$ are both $\pi/4 \mod \pi/2$. If $\beta = \alpha+\pi$, the result is obvious thanks to \supp. If $\beta = \alpha +\pi/2$, using \so,~\ref{lem:C1-bis}, \bt, \bo, \picom and~\ref{lem:hopf}:
%\def\fig{new-supp-proof-alpha-plus-pi_2}
%\begin{align*}
%\tikzfigc{00}=\tikzfigc{01}=\tikzfigc{02}\\=\tikzfigc{03}=\tikzfigc{04}=\tikzfigc{05}\\=\tikzfigc{06}=\tikzfigc{07}=\tikzfigc{08}\\=\tikzfigc{09}=\tikzfigc{10}
%\end{align*}
%The proof is similar when $\beta = \alpha -\pi/2$. Finally, when $\alpha=\beta$, one can get to one of the previous cases by using the rule \picom:
%\[\fit{\tikzfig{new-supp-proof-beta-eq-alpha}}\]
%\end{proof}

\begin{proof}[Proof of Lemma~\ref{lem:supp-to-minus-pi_4}]
\def\fig{supp-to-minus-pi_4-proof}
\begin{align*}
\input{./figures/\fig/\fig_00.tikz}
\eq{\ref{lem:inverse}\\\picom\\\ref{lem:multiplying-global-phases}}\input{./figures/\fig/\fig_01.tikz}
\eq{\ref{lem:C1-bis}\\\ref{lem:inverse}}\input{./figures/\fig/\fig_02.tikz}
\eq{\ref{lem:green-state-pi_2-is-red-state-minus-pi_2}\\\ref{lem:multiplying-global-phases}}\input{./figures/\fig/\fig_03.tikz}
\eq{\ref{lem:euler-decomp-with-scalar}\\\id\\\ref{lem:multiplying-global-phases}}\input{./figures/\fig/\fig_04.tikz}\\
\eq{\id\\\ref{lem:inverse}\\\ref{lem:hopf}}\input{./figures/\fig/\fig_05.tikz}
\eq{\picom\\\ref{lem:multiplying-global-phases}}\input{./figures/\fig/\fig_06.tikz}
\eq{\ref{lem:2-is-sqrt2-squared}\\\ref{lem:inverse}}\input{./figures/\fig/\fig_07.tikz}
\end{align*}
\end{proof}

%\begin{proof}[Proof of Lemma~\ref{lem:red-pi_2-around-green-node}]
%Using \so, \bt, \hd, \h and~\ref{lem:h-loop}:
%\def\fig{lemma-red-pi_2-around-green-node-proof}
%\begin{align*}
%\tikzfigc{00}=\tikzfigc{01}=\tikzfigc{02}=\tikzfigc{03}=\tikzfigc{04}\\=\tikzfigc{05}=\tikzfigc{06}
%\end{align*}
%\end{proof}

\begin{proof}[Proof of Lemma~\ref{lem:triangle-decomp-2}]
\def\fig{triangle-decomp-2-proof}
\begin{align*}
\input{./figures/\fig/\fig_00.tikz}
\eq{\ref{def:triangle}\\\picom\\\ref{lem:multiplying-global-phases}}\input{./figures/\fig/\fig_01.tikz}
\eq{\h\\\ref{lem:hopf}}\input{./figures/\fig/\fig_02.tikz}
\eq{\ref{lem:hadamard-involution}\\\ref{lem:C1-original}}\input{./figures/\fig/\fig_03.tikz}
\eq{\h\\\ref{lem:euler-decomp-with-scalar}\\\ref{lem:multiplying-global-phases}}\input{./figures/\fig/\fig_04.tikz}\\
\eq{\ref{lem:red-pi_2-around-green-node}\\\ref{lem:hopf}}\input{./figures/\fig/\fig_05.tikz}
\eq{\ref{lem:euler-decomp-with-scalar}\\\ref{lem:multiplying-global-phases}}\input{./figures/\fig/\fig_06.tikz}
\eq{\h}\input{./figures/\fig/\fig_07.tikz}
\end{align*}
\end{proof}

\begin{proof}[Proof of Lemma~\ref{lem:not-triangle-is-symmetrical}]
\def\fig{not-triangle-is-symmetrical-proof}
\begin{align*}
\input{./figures/\fig/\fig_00.tikz}
\eq{\ref{lem:triangle-decomp-2}}\input{./figures/\fig/\fig_01.tikz}
\eq{\ref{lem:6-cycle-bialgebra}}\input{./figures/\fig/\fig_02.tikz}
\eq{}\input{./figures/\fig/\fig_03.tikz}
\eq{\ref{lem:triangle-decomp-2}}\input{./figures/\fig/\fig_04.tikz}
\end{align*}
\end{proof}

\begin{proof}[Proof of Lemma~\ref{lem:red-state-on-triangle}]
\def\fig{red-state-on-triangle-proof}
\begin{align*}
\input{./figures/\fig/\fig_00.tikz}
\eq{\ref{def:triangle}}\input{./figures/\fig/\fig_01.tikz}
\eq{\ref{lem:inverse}\\\cp}\input{./figures/\fig/\fig_02.tikz}
\eq{\ref{lem:gn-pi_2-0-0-equals-sqrt2-exp-pi_4}}\input{./figures/\fig/\fig_03.tikz}
\eq{\ref{lem:green-state-pi_2-is-red-state-minus-pi_2}\\\ref{lem:multiplying-global-phases}}\input{./figures/\fig/\fig_04.tikz}
\eq{}\input{./figures/\fig/\fig_05.tikz}
\end{align*}
\end{proof}

\begin{proof}[Proof of Lemma~\ref{lem:pi-red-state-on-triangle}]
\def\fig{pi-red-state-on-triangle-proof}
\begin{align*}
\input{./figures/\fig/\fig_00.tikz}
\eq{\ref{def:triangle}}\input{./figures/\fig/\fig_01.tikz}
\eq{\ref{lem:k1}\\\cp}\input{./figures/\fig/\fig_02.tikz}
\eq{\picom\\\ref{lem:multiplying-global-phases}}\input{./figures/\fig/\fig_03.tikz}
\eq{\cp\\\ref{lem:bicolor-0-alpha}}\input{./figures/\fig/\fig_04.tikz}
\eq{\ref{lem:2-is-sqrt2-squared}\\\ref{lem:inverse}}\input{./figures/\fig/\fig_05.tikz}
\end{align*}
\end{proof}

\begin{proof}[Proof of Lemmas~\ref{lem:red-state-on-upside-down-triangle} and~\ref{lem:pi-red-state-on-upside-down-triangle}]
The result comes naturally from~\ref{lem:not-triangle-is-symmetrical},~\ref{lem:pi-red-state-on-triangle} and~\ref{lem:red-state-on-triangle}.
\end{proof}

\begin{proof}[Proof of Lemma~\ref{lem:pi-green-state-on-upside-down-triangle}]
\def\fig{pi-green-state-on-upside-down-triangle-proof}
\begin{align*}
\input{./figures/\fig/\fig_00.tikz}
\eq[~]{\ref{def:triangle}}\input{./figures/\fig/\fig_01.tikz}
\eq[~]{\s\\\ref{lem:k1}\\\cp}\input{./figures/\fig/\fig_02.tikz}
\eq[~]{\supp\\\ref{lem:inverse}\\\ref{lem:hopf}}\input{./figures/\fig/\fig_03.tikz}
\eq[~]{\cp}\input{./figures/\fig/\fig_04.tikz}
\eq[~]{\ref{lem:gn-pi_2-0-0-equals-sqrt2-exp-pi_4}}\input{./figures/\fig/\fig_05.tikz}
\eq[~]{\ref{lem:multiplying-global-phases}\\\ref{lem:inverse}}\input{./figures/\fig/\fig_06.tikz}
\end{align*}
\end{proof}

\begin{proof}[Proof of Lemma~\ref{lem:pi-green-state-on-triangle}]
\def\fig{pi-green-state-on-triangle-proof}
\begin{align*}
\input{./figures/\fig/\fig_00.tikz}
\eq{\ref{lem:not-triangle-is-symmetrical}}\input{./figures/\fig/\fig_01.tikz}
\eq{\ref{lem:k1}\\\cp}\input{./figures/\fig/\fig_02.tikz}
\eq{\ref{lem:inverse}\\\ref{lem:pi-green-state-on-upside-down-triangle}}\input{./figures/\fig/\fig_03.tikz}
\eq{\ref{lem:k1}\\\cp\\\ref{lem:multiplying-global-phases}}\input{./figures/\fig/\fig_04.tikz}
\end{align*}
\end{proof}

\begin{proof}[Proof of Lemma~\ref{lem:triangle-hadamard-parallel}]
\def\fig{anti-control-pi-control-green-aux}
\begin{align*}
\input{./figures/\fig/\fig_00.tikz}
\eq{\h}\input{./figures/\fig/\fig_01.tikz}
\eq{\ref{lem:euler-decomp-with-scalar}\\\id\\\ref{lem:inverse}}\input{./figures/\fig/\fig_02.tikz}
\eq{\bialg}\input{./figures/\fig/\fig_03.tikz}
\eq{\picom\\\ref{lem:multiplying-global-phases}}\input{./figures/\fig/\fig_04.tikz}
\eq{\ref{lem:euler-decomp-with-scalar}\\\id}\input{./figures/\fig/\fig_05.tikz}
\end{align*}
Then:
% Je suppose que "the previous result est lem:pi-green-state-on-triangle
\def\fig{anti-control-pi-control-green-fin}
\begin{align*}
\input{./figures/\fig/\fig_00.tikz}
\eq{\ref{def:triangle}}\input{./figures/\fig/\fig_01.tikz}
\eq{}\input{./figures/\fig/\fig_02.tikz}
\eq{\ref{def:triangle}}\input{./figures/\fig/\fig_03.tikz}
\end{align*}
\end{proof}

\begin{proof}[Proof of Lemma~\ref{lem:looped-triangle}]
First:
\def\fig{looped-triangle-proof}
\begin{align*}
\input{./figures/\fig/\fig_00.tikz}
\quad\eq[]{\ref{def:triangle}}\input{./figures/\fig/\fig_01.tikz}
\eq{\bialg}\input{./figures/\fig/\fig_02.tikz}
\eq{}\input{./figures/\fig/\fig_03.tikz}
\eq{\ref{def:triangle}}\input{./figures/\fig/\fig_04.tikz}
\end{align*}
Then:
\def\fig{looped-triangle-proof-fin}
\begin{align*}
\input{./figures/\fig/\fig_00.tikz}
\eq{\ref{lem:not-triangle-is-symmetrical}}\input{./figures/\fig/\fig_01.tikz}
\eq{\ref{lem:k1}}\input{./figures/\fig/\fig_02.tikz}
\eq{}\input{./figures/\fig/\fig_03.tikz}
\eq{\ref{lem:not-triangle-is-symmetrical}\\\id}\input{./figures/\fig/\fig_04.tikz}
\end{align*}
\end{proof}

\begin{proof}[Proof of Lemma~\ref{lem:black-dot-swappable-outputs}]
\def\fig{lemma-black-dot-swappable-outputs-proof}
\begin{align*}
\input{./figures/\fig/\fig_00.tikz}
\eq[\quad]{\bialg}\input{./figures/\fig/\fig_01.tikz}
\eq[\quad]{\s\\\ref{lem:looped-triangle}}\input{./figures/\fig/\fig_02.tikz}
\end{align*}
\end{proof}

\begin{proof}[Proof of Lemma~\ref{lem:not-ug-is-symmetrical}]
%\def\fig{not-ug-is-symmetrical-proof-from-simplified}
%\begin{align*}
%\tikzfigc{00}
%\eq{\ref{def:triangle}}\tikzfigc{01}
%\eq{\ref{lem:inverse}\\\picom}\tikzfigc{02}
%\eq{\id\\\picom\\\ref{lem:multiplying-global-phases}}\tikzfigc{03}
%\eq{\bw}\tikzfigc{04}\\
%\eq{\ref{lem:k1}\\\picom\\\ref{lem:multiplying-global-phases}\\\ref{lem:inverse}}\tikzfigc{05}
%\eq{\ref{def:triangle}}\tikzfigc{06}
%\end{align*}
\def\fig{bw-equivalence}
\begin{align*}
\input{./figures/\fig/\fig_00.tikz}\eq{}\input{./figures/\fig/\fig_01.tikz}
\equi{}\hspace{-1em}\input{./figures/\fig/\fig_02.tikz}\hspace{-1.5em}\eq{}\input{./figures/\fig/\fig_03.tikz}
\equi{\picom\\\ref{lem:inverse}\\\id\\\ref{lem:k1}}\input{./figures/\fig/\fig_04.tikz}\eq{}\input{./figures/\fig/\fig_05.tikz}\\
\equi{\picom\\\ref{lem:multiplying-global-phases}\\\ref{lem:inverse}}\input{./figures/\fig/\fig_06.tikz}\eq{}\input{./figures/\fig/\fig_07.tikz}
\equi{\ref{lem:multiplying-global-phases}\\\ref{lem:inverse}}\input{./figures/\fig/\fig_08.tikz}\eq{}\input{./figures/\fig/\fig_09.tikz}
\end{align*}
\end{proof}

\begin{proof}[Proof of Lemma~\ref{lem:inverse-of-triangle}]
\def\fig{inverse-of-triangle-proof}
\begin{align*}
\input{./figures/\fig/\fig_00.tikz}
\eq{\id\\\ref{lem:not-triangle-is-symmetrical}}\input{./figures/\fig/\fig_01.tikz}
\eq{\ref{def:triangle}\\\ref{lem:k1}}\input{./figures/\fig/\fig_02.tikz}
\eq{\ref{lem:supp-to-minus-pi_4}\\\ref{lem:k1}\\\ref{lem:inverse}\\\picom\\\ref{lem:multiplying-global-phases}}\input{./figures/\fig/\fig_03.tikz}
\eq{\supp\\\ref{lem:inverse}\\\ref{lem:hopf}}\input{./figures/\fig/\fig_04.tikz}
\eq{\ref{lem:gn-pi_2-0-0-equals-sqrt2-exp-pi_4}\\\picom}\input{./figures/\fig/\fig_05.tikz}
\eq{\id\\\ref{lem:multiplying-global-phases}\\\ref{lem:inverse}}\input{./figures/\fig/\fig_06.tikz}
\end{align*}
\end{proof}

\begin{proof}[Proof of Lemma~\ref{lem:symmetric-diagram-with-triangle-hadamard}]
First:
\def\fig{symmetric-diagram-with-triangle-hadamard-proof-1}
\begin{align*}
\input{./figures/\fig/\fig_00.tikz}
\eq{\ref{lem:euler-decomp-with-scalar}}\input{./figures/\fig/\fig_01.tikz}
\eq{\ref{lem:multiplying-global-phases}\\\ref{lem:gn-pi_2-0-0-equals-sqrt2-exp-pi_4}\\\id\\\supp}\input{./figures/\fig/\fig_02.tikz}
\eq{\id\\\picom\\\ref{lem:multiplying-global-phases}\\\ref{lem:k1}}\input{./figures/\fig/\fig_03.tikz}
\eq{\ref{lem:multiplying-global-phases}\\\ref{lem:supp-to-minus-pi_4}}\input{./figures/\fig/\fig_04.tikz}\\
\eq{\ref{def:triangle}\\\picom\\\ref{lem:multiplying-global-phases}}\input{./figures/\fig/\fig_05.tikz}
\eq{\ref{def:triangle}}\input{./figures/\fig/\fig_06.tikz}
\eq{\ref{lem:not-triangle-is-symmetrical}}\input{./figures/\fig/\fig_07.tikz}
\end{align*}
Then:
\def\fig{symmetric-diagram-with-triangle-hadamard-proof-2}
\begin{align*}
\input{./figures/\fig/\fig_00.tikz}
\eq{\ref{lem:inverse-of-triangle}}\input{./figures/\fig/\fig_01.tikz}
\eq{}\input{./figures/\fig/\fig_02.tikz}
\eq{\h}\input{./figures/\fig/\fig_03.tikz}
\eq{\ref{lem:not-triangle-is-symmetrical}}\input{./figures/\fig/\fig_04.tikz}
\end{align*}
\end{proof}

\begin{proof}[Proof of Lemma~\ref{lem:anti-control-hanging-branch-and-control-alpha-commute}]
First:
\def\fig{anti-control-hanging-branch-and-control-alpha-commute-proof-aux}
\begin{align*}
\input{./figures/\fig/\fig_00.tikz}
\eq{\h}\input{./figures/\fig/\fig_01.tikz}
\eq{\ref{lem:symmetric-diagram-with-triangle-hadamard}\\\ref{lem:not-triangle-is-symmetrical}}\input{./figures/\fig/\fig_02.tikz}
\eq{\ref{def:triangle}\\\bialg\\\ref{lem:k1}}\input{./figures/\fig/\fig_03.tikz}\\
\eq{\h}\input{./figures/\fig/\fig_04.tikz}
\eq{\ref{lem:C1-original}}\input{./figures/\fig/\fig_05.tikz}
\eq{\h\\\ref{lem:hopf}}\input{./figures/\fig/\fig_06.tikz}\\
\eq{\ref{lem:green-state-pi_2-is-red-state-minus-pi_2}}\input{./figures/\fig/\fig_07.tikz}
\eq{\ref{lem:supp-to-minus-pi_4}\\\id}\input{./figures/\fig/\fig_08.tikz}
\eq{\ref{lem:inverse}\\\ref{lem:hopf}\\\id}\input{./figures/\fig/\fig_09.tikz}
\eq{\h}\input{./figures/\fig/\fig_10.tikz}
\end{align*}
then:
\def\fig{anti-control-hanging-branch-and-control-alpha-commute-proof-fin}
\begin{align*}
\input{./figures/\fig/\fig_00.tikz}
\eq{\ref{lem:inverse}}\!\!\!\!\!\input{./figures/\fig/\fig_01.tikz}
\eq{\ref{lem:inverse}\\\ref{lem:hopf}\\\id\\\ref{def:triangle}}\input{./figures/\fig/\fig_02.tikz}
\eq{\bialg\\\ref{lem:inverse}}\input{./figures/\fig/\fig_03.tikz}\\
\eq{}\input{./figures/\fig/\fig_04.tikz}
\eq{\ccom}\input{./figures/\fig/\fig_05.tikz}
\eq{}\input{./figures/\fig/\fig_06.tikz}
\end{align*}
where the penultimate diagram is the mirror of the previous one with $\gamma:=-\gamma$, so the last step consists in rolling back to the mirror of the initial diagram with $\gamma:=-\gamma$.
\end{proof}

\begin{proof}[Proof of Lemma~\ref{lem:n-four}]
\def\fig{ug-and-not-W-commute-simplified-proof-from-new-c1}
\begin{align*}
\input{./figures/\fig/\fig_00.tikz}
\eq{\id\\\ref{lem:not-triangle-is-symmetrical}}\input{./figures/\fig/\fig_01.tikz}
\eq{\ref{def:triangle}\\\ref{lem:k1}}\input{./figures/\fig/\fig_02.tikz}
\eq{\ref{lem:supp-to-minus-pi_4}\\\picom\\\ref{lem:multiplying-global-phases}}\input{./figures/\fig/\fig_03.tikz}
\eq{\bialg\\\ref{lem:euler-decomp-with-scalar}}\input{./figures/\fig/\fig_04.tikz}\\
\eq{\h}\input{./figures/\fig/\fig_05.tikz}
\eq{\h\\\hd}\input{./figures/\fig/\fig_06.tikz}
\eq{\picom\\\ref{lem:anti-control-hanging-branch-and-control-alpha-commute}}\input{./figures/\fig/\fig_07.tikz}
\eq{\hd}\input{./figures/\fig/\fig_08.tikz}
\eq{\ref{lem:euler-decomp-with-scalar}}\input{./figures/\fig/\fig_09.tikz}\\
\eq{\ref{lem:control-pi-and-anti-CNOT-commute}}\input{./figures/\fig/\fig_10.tikz}
\eq{\ref{lem:euler-decomp-with-scalar}}\input{./figures/\fig/\fig_11.tikz}
\eq{\h\\\picom\\\ref{lem:multiplying-global-phases}}\input{./figures/\fig/\fig_12.tikz}
\eq{\ref{lem:euler-decomp-with-scalar}\\\ref{lem:multiplying-global-phases}}\input{./figures/\fig/\fig_13.tikz}
\eq{\ref{lem:supp-to-minus-pi_4}\\\ref{def:triangle}\\\ref{lem:not-triangle-is-symmetrical}}\input{./figures/\fig/\fig_14.tikz}
\end{align*}
\end{proof}

%\begin{proof}[Proof of Lemma~\ref{lem:big-scalar-equation}]
%Using \so, \bo,~\ref{lem:new-supp} and \picom:
%\def\fig{new-supp-proves-big-scalar-equation}
%\begin{align*}
%\tikzfigc{00}=\tikzfigc{01}=\tikzfigc{02}\\=\tikzfigc{03}=\tikzfigc{04}
%\end{align*}
%\end{proof}

\begin{proof}[Proof of Lemma~\ref{lem:parallel-triangles}]
First:
\def\fig{parallel-triangles-are-projections-proof}
\begin{align*}
\input{./figures/\fig/\fig_00.tikz}
\eq{\ref{lem:red-pi_2-around-green-node}}\input{./figures/\fig/\fig_01.tikz}
\eq{\h}\input{./figures/\fig/\fig_02.tikz}
\eq{\picom\\\ref{lem:C1-original}\\\ref{lem:hadamard-involution}}\input{./figures/\fig/\fig_03.tikz}
\eq{\hd\\\ref{lem:green-state-pi_2-is-red-state-minus-pi_2}\\\ref{lem:inverse}}\input{./figures/\fig/\fig_04.tikz}\\
\eq{\ref{lem:inverse}\\\bialg}\input{./figures/\fig/\fig_05.tikz}
\eq{\ref{lem:inverse}\\\cp\\\picom\\\ref{lem:multiplying-global-phases}}\input{./figures/\fig/\fig_06.tikz}
\end{align*}
Then:
\def\fig{parallel-triangles-are-projections-proof-fin}
\begin{align*}
\input{./figures/\fig/\fig_00.tikz}
\eq{\ref{def:triangle}}\input{./figures/\fig/\fig_01.tikz}
\eq{}\input{./figures/\fig/\fig_02.tikz}
\eq{\ref{def:triangle}\\\ref{lem:supp-to-minus-pi_4}}\input{./figures/\fig/\fig_03.tikz}
\eq{\ref{lem:multiplying-global-phases}\\\ref{lem:inverse}}\input{./figures/\fig/\fig_04.tikz}
\end{align*}
\end{proof}

\begin{proof}[Proof of Lemma~\ref{lem:triangles-fork-absorbs-anti-CNOT}]
First:
\def\fig{ug-fork-absorbs-CNOT-proof}
\begin{align*}
\input{./figures/\fig/\fig_00.tikz}
\eq{\ref{lem:inverse}\\\ref{lem:multiplying-global-phases}\\\ref{lem:euler-decomp-with-scalar}}\input{./figures/\fig/\fig_01.tikz}
\eq{\bialg}\input{./figures/\fig/\fig_02.tikz}
\eq{\h}\input{./figures/\fig/\fig_03.tikz}\\
\eq{\ref{lem:C1-original}}\input{./figures/\fig/\fig_04.tikz}
\eq{\h\\\ref{lem:hopf}\\\ref{lem:euler-decomp-with-scalar}}\input{./figures/\fig/\fig_05.tikz}
\eq{\ref{lem:euler-decomp-with-scalar}\\\ref{lem:multiplying-global-phases}\\\ref{lem:inverse}}\input{./figures/\fig/\fig_06.tikz}
\end{align*}
Hence:
\def\fig{ug-fork-absorbs-CNOT-proof-2}
\begin{align*}
\input{./figures/\fig/\fig_00.tikz}
\eq{\ref{def:triangle}}\input{./figures/\fig/\fig_01.tikz}
\eq{}\input{./figures/\fig/\fig_02.tikz}
\eq{\ref{def:triangle}}\input{./figures/\fig/\fig_03.tikz}
\end{align*}
Then:
\def\fig{ug-fork-absorbs-anti-CNOT-proof}
\begin{align*}
\input{./figures/\fig/\fig_00.tikz}
\eq{\ref{lem:not-triangle-is-symmetrical}\\\ref{lem:k1}}\input{./figures/\fig/\fig_01.tikz}
\eq{\ref{lem:not-triangle-is-symmetrical}\\\ref{lem:k1}}\input{./figures/\fig/\fig_02.tikz}
\eq{}\input{./figures/\fig/\fig_03.tikz}
\eq{\ref{lem:not-triangle-is-symmetrical}\\\ref{lem:k1}}\input{./figures/\fig/\fig_04.tikz}
\end{align*}
\end{proof}

\begin{proof}[Proof of Lemma~\ref{lem:n-five}]
First:
\def\fig{2-diagrams-of-control-triangle-from-simplified-BW-1}
\begin{align*}
\input{./figures/\fig/\fig_00.tikz}
\eq{\id\\\ref{lem:not-ug-is-symmetrical}}\input{./figures/\fig/\fig_01.tikz}
\eq{\h}\input{./figures/\fig/\fig_02.tikz}
\eq{\ref{lem:symmetric-diagram-with-triangle-hadamard}}\input{./figures/\fig/\fig_03.tikz}
\eq{\ref{lem:k1}\\\ref{lem:not-triangle-is-symmetrical}}\input{./figures/\fig/\fig_04.tikz}
\eq{\id\\\cp\\\ref{lem:inverse}}\input{./figures/\fig/\fig_05.tikz}
\eq{\ref{lem:black-dot-swappable-outputs}\\\ref{lem:n-four}}\input{./figures/\fig/\fig_06.tikz}
\end{align*}
Moreover, from~\ref{lem:symmetric-diagram-with-triangle-hadamard}, we can easily derive:
\def\fig{2-diagrams-of-control-triangle-from-simplified-BW-2}
\begin{align*}
\input{./figures/\fig/\fig_00.tikz}
\eq{\id\\\ref{lem:symmetric-diagram-with-triangle-hadamard}}\input{./figures/\fig/\fig_01.tikz}
\eq{\h\\\ref{lem:k1}\\\ref{lem:not-triangle-is-symmetrical}}\input{./figures/\fig/\fig_02.tikz}
\end{align*}
and:
\def\fig{2-diagrams-of-control-triangle-from-simplified-BW-3}
\begin{align*}
\input{./figures/\fig/\fig_00.tikz}
\eq{\id\\\ref{lem:not-triangle-is-symmetrical}}\input{./figures/\fig/\fig_01.tikz}
\eq{\ref{lem:symmetric-diagram-with-triangle-hadamard}}\input{./figures/\fig/\fig_02.tikz}
\eq{\h}\input{./figures/\fig/\fig_03.tikz}
\end{align*}
Finally:
\def\fig{2-diagrams-of-control-triangle-from-simplified-BW-4}
\begin{align*}
\input{./figures/\fig/\fig_00.tikz}
\eq{\h}\input{./figures/\fig/\fig_01.tikz}
\eq{}\input{./figures/\fig/\fig_02.tikz}
\eq{\ref{lem:k1}\\\ref{lem:n-four}\\\ref{lem:black-dot-swappable-outputs}}\input{./figures/\fig/\fig_03.tikz}
\eq{}\input{./figures/\fig/\fig_04.tikz}\\
\eq{\bialg\\\ref{lem:inverse}}\input{./figures/\fig/\fig_05.tikz}
\eq{}\input{./figures/\fig/\fig_06.tikz}
\end{align*}
\end{proof}
%
%\begin{proof}[Proof of Lemma~\ref{lem:2-triangle-cycle}]
%\needs[2-triangle-cycle]{lem:2-triangle-cycle}{lem:not-ug-is-symmetrical,lem:supp-to-minus-pi_4,lem:euler-decomp-with-scalar,lem:k1,lem:h-loop,lem:green-state-pi_2-is-red-state-minus-pi_2}
%Using \so,~\ref{lem:not-ug-is-symmetrical},~\ref{lem:supp-to-minus-pi_4},~\ref{lem:euler-decomp-with-scalar}, \bt, \hd, \h, \ccom, \picom,~\ref{lem:k1},~\ref{lem:h-loop} and~\ref{lem:green-state-pi_2-is-red-state-minus-pi_2}:
%\[\tikzfig{2-triangle-cycle-proof-1}\]
%\[\tikzfig{2-triangle-cycle-proof-2}\]
%\[\tikzfig{2-triangle-cycle-proof-3}\]
%\end{proof}
%
%\begin{proof}[Proof of Lemma~\ref{lem:2-cycle-triangle-with-one-not}]
%\needs[2-cycle-triangle-with-one-not]{lem:2-cycle-triangle-with-one-not}{lem:looped-triangle,lem:black-dot-swappable-outputs,lem:2-triangle-cycle,lem:hopf}
%Using~\ref{lem:looped-triangle},~\ref{lem:black-dot-swappable-outputs},~\ref{lem:2-triangle-cycle} and~\ref{lem:hopf}:
%\[\tikzfig{2-cycle-triangle-with-one-not-proof}\]
%\end{proof}

\section{Proof of Propositions~\ref{prop:rules-preserved} and~\ref{prop:left-inverse-ZX}}

We first derive an easy but useful lemma for the following:
\begin{lem}%
\label{lem:arity-1-black-dot}
As shown in Section~\ref{sec:cliff-t}:
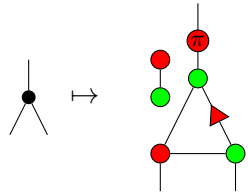
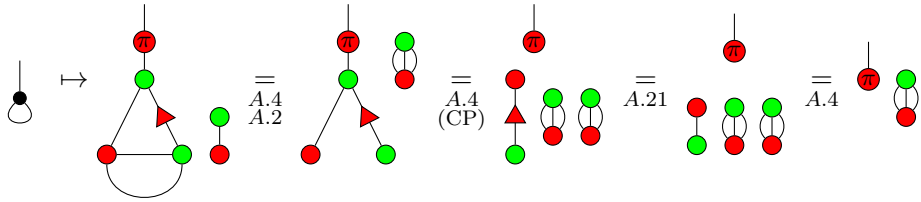
\[
%\InputIfFileExists{#1.tikz}{}{
\input{./figures/ZW-to-ZX-dot-1-2-simplified.tikz}%} % chktex 27
\]
Then:
\def\fig{W-1-0} % chktex 8
\begin{align*}
\input{./figures/\fig/\fig_00.tikz}\mapsto\input{./figures/\fig/\fig_01.tikz}
\eq{\ref{lem:inverse}\\\ref{lem:hopf}}\input{./figures/\fig/\fig_02.tikz}
\eq{\ref{lem:inverse}\\\cp}\input{./figures/\fig/\fig_03.tikz}
\eq{\ref{lem:red-state-on-triangle}}\input{./figures/\fig/\fig_04.tikz}
\eq{\ref{lem:inverse}}\input{./figures/\fig/\fig_05.tikz}
\end{align*}
\end{lem}

\subsection{Proof of Proposition~\ref{prop:rules-preserved}}%
\label{prf:rules-preserved}
We prove here that all the rules of the \zwh-Calculus are preserved by $[.]_X$.
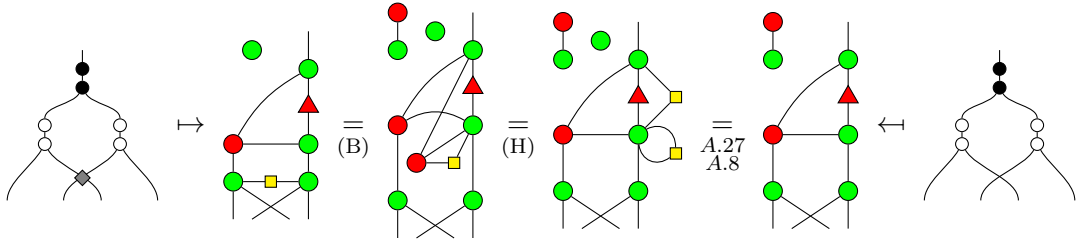
\begin{itemize}
\item X\@:
\def\fig{rule-X-proof-no-braid}
\begin{align*}
\input{./figures/\fig/\fig_00.tikz}~~\mapsto~~\input{./figures/\fig/\fig_01.tikz}
\eq{\bialg}\input{./figures/\fig/\fig_02.tikz}
\eq{\h}\input{./figures/\fig/\fig_03.tikz}
\eq{\ref{lem:triangle-hadamard-parallel}\\\ref{lem:h-loop}}\input{./figures/\fig/\fig_04.tikz}~~\mapsfrom~~\input{./figures/\fig/\fig_05.tikz}
\end{align*}
\item $0a$, $0c$, $0d$ and $0d'$ come directly from the paradigm \emph{Only Topology Matters}.
\item $0b$:
\def\fig{rule-0b-proof-no-braid}
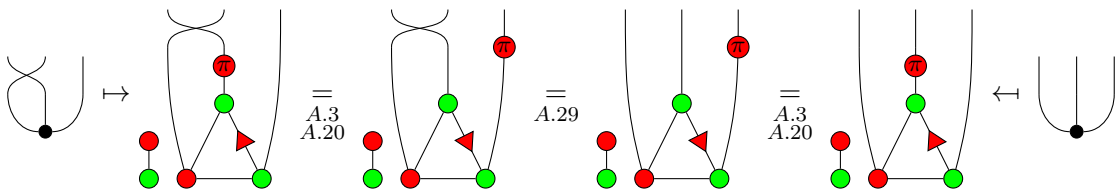
\begin{align*}
\input{./figures/\fig/\fig_00.tikz}~~\mapsto\input{./figures/\fig/\fig_01.tikz}
\eq{\ref{lem:k1}\\\ref{lem:not-triangle-is-symmetrical}}\input{./figures/\fig/\fig_02.tikz}
\eq{\ref{lem:black-dot-swappable-outputs}}\input{./figures/\fig/\fig_03.tikz}
\eq{\ref{lem:k1}\\\ref{lem:not-triangle-is-symmetrical}}\input{./figures/\fig/\fig_04.tikz}~\mapsfrom~~\input{./figures/\fig/\fig_05.tikz}
\end{align*}
\item $0b'$: Using the result for rule $0b$,
\def\fig{rule-0bp-proof-no-braid}
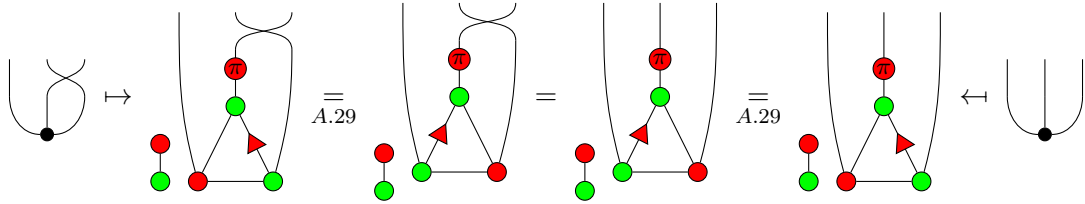
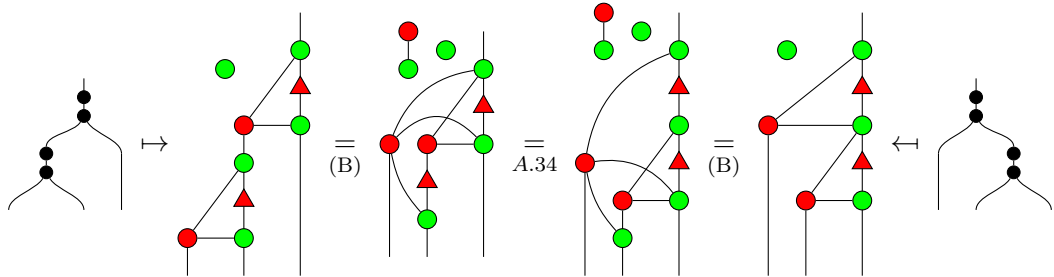
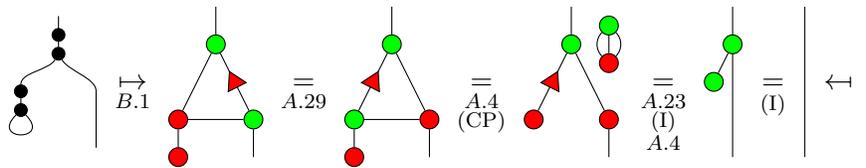
\begin{align*}
\input{./figures/\fig/\fig_00.tikz}~~\mapsto~~\input{./figures/\fig/\fig_01.tikz}
\eq{\ref{lem:black-dot-swappable-outputs}}\input{./figures/\fig/\fig_02.tikz}
\eq{}\input{./figures/\fig/\fig_03.tikz}
\eq{\ref{lem:black-dot-swappable-outputs}}\input{./figures/\fig/\fig_04.tikz}~~\mapsfrom~~\input{./figures/\fig/\fig_05.tikz}
\end{align*}
\item $1a$:
\def\fig{rule-1a-proof}
\begin{align*}
\input{./figures/\fig/\fig_00.tikz}~~\mapsto\input{./figures/\fig/\fig_01.tikz}
\eq{\bialg}\input{./figures/\fig/\fig_02.tikz}
\eq{\ref{lem:n-four}}\input{./figures/\fig/\fig_03.tikz}
\eq{\bialg}\input{./figures/\fig/\fig_04.tikz}~~\mapsfrom~~\input{./figures/\fig/\fig_05.tikz}
\end{align*}
\item $1b$:
\def\fig{rule-1b-proof}
\begin{align*}
\input{./figures/\fig/\fig_00.tikz}~~\underset{\ref{lem:arity-1-black-dot}}{\mapsto}~~\input{./figures/\fig/\fig_01.tikz}
\eq{\ref{lem:black-dot-swappable-outputs}}\input{./figures/\fig/\fig_02.tikz}
\eq{\ref{lem:inverse}\\\cp}\input{./figures/\fig/\fig_03.tikz}
\eq{\ref{lem:red-state-on-upside-down-triangle}\\\id\\\ref{lem:inverse}}\input{./figures/\fig/\fig_04.tikz}
\eq{\id}\input{./figures/\fig/\fig_05.tikz}~~\mapsfrom~~\input{./figures/\fig/\fig_05.tikz}
\end{align*}
\item $1c$, $1d$, $2a$ and $2b$ come directly from the spider rules \s and \id.
\item $3a$ is the expression of the colour-swapped version of Lemma~\ref{lem:k1}.
\item $3b$:
\def\fig{rule-3b-proof}
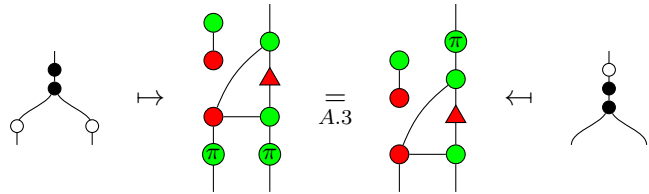
\begin{align*}
\input{./figures/\fig/\fig_00.tikz}\quad\mapsto\quad\input{./figures/\fig/\fig_01.tikz}
\eq[\quad]{\ref{lem:k1}}\input{./figures/\fig/\fig_02.tikz}\quad\mapsfrom\quad\input{./figures/\fig/\fig_03.tikz}
\end{align*}
\item $4$ comes from the spider rule \s.
\item $5a$: We will need a few steps to prove this equality.
%\step \label{step:one}
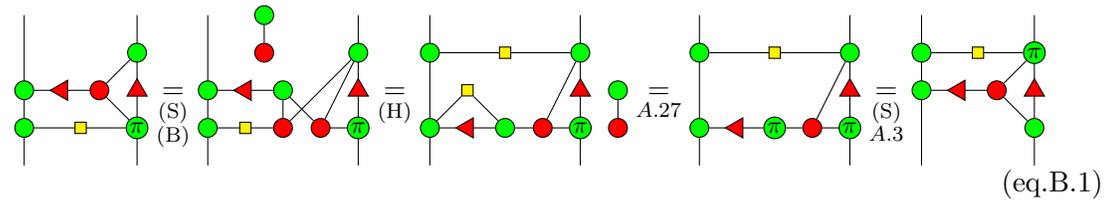
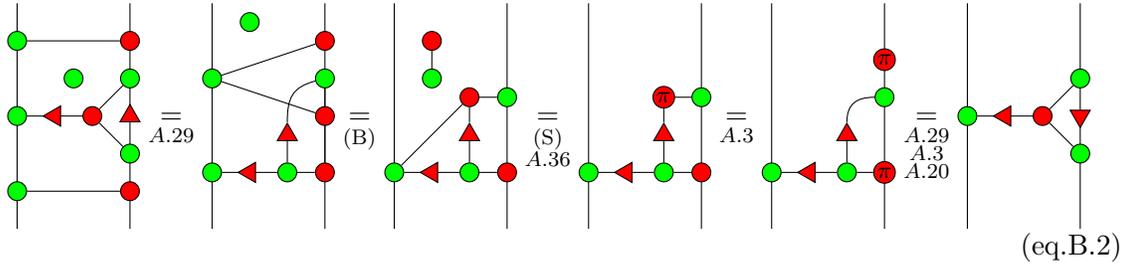
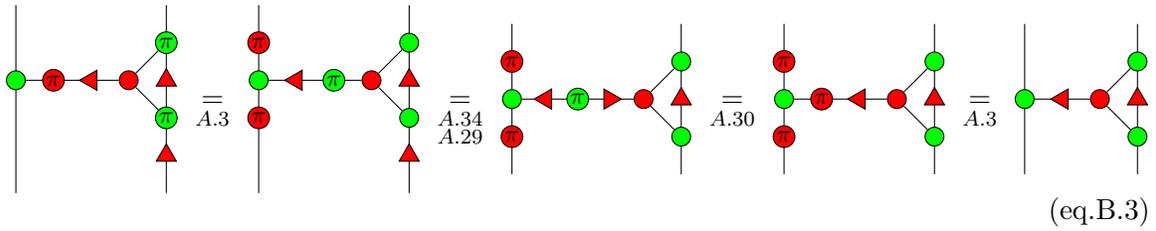
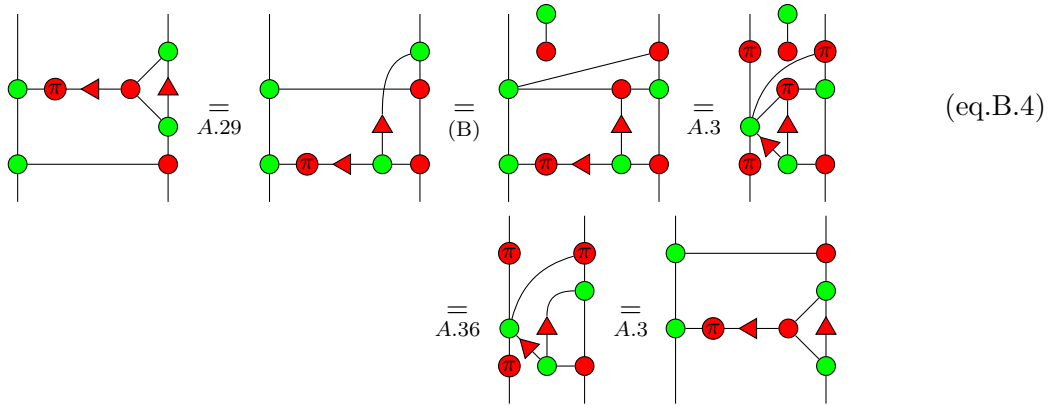
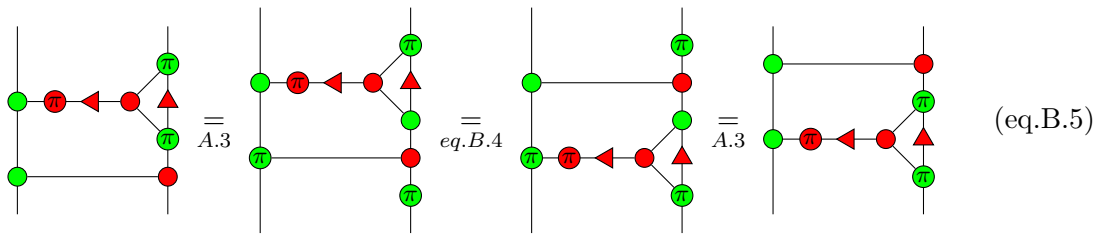
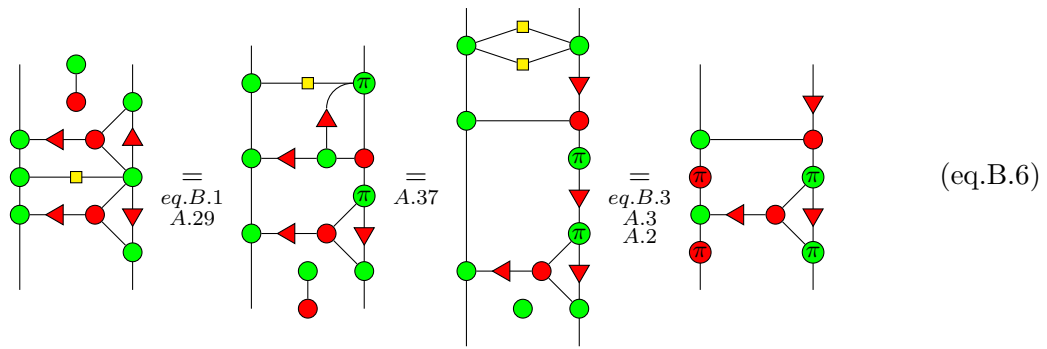
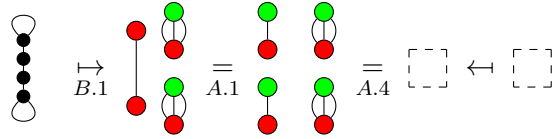
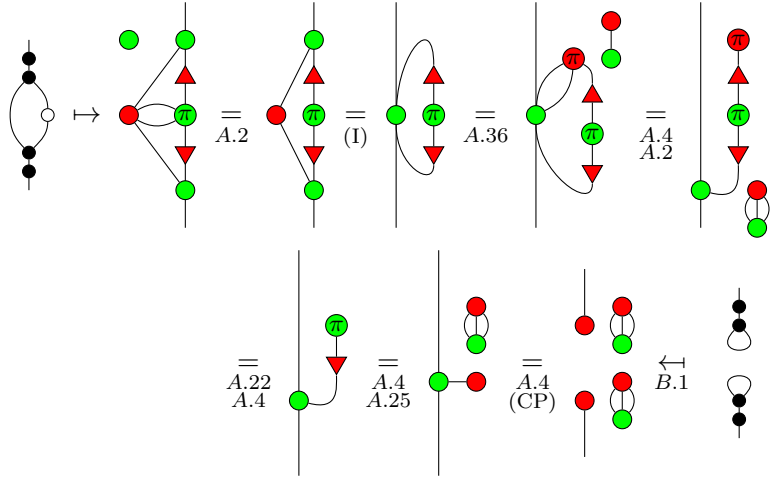
\def\fig{anti-control-pi-and-control-triangle-commute}
\begin{align}\label{step:one}
\input{./figures/\fig/\fig_00.tikz}
\eq[]{\s\\\bialg}\input{./figures/\fig/\fig_01.tikz}
\eq[]{\h}\input{./figures/\fig/\fig_02.tikz}
\eq[]{\ref{lem:triangle-hadamard-parallel}}\input{./figures/\fig/\fig_03.tikz}
\eq[]{\s\\\ref{lem:k1}}\input{./figures/\fig/\fig_04.tikz}
\end{align}
%\step \label{step:two}
\def\fig{CNOT-control-triangle-CNOT-is-control-transpose-triangle}
\begin{align}\label{step:two}
\input{./figures/\fig/\fig_00.tikz}
\eq[]{\ref{lem:black-dot-swappable-outputs}}\input{./figures/\fig/\fig_01.tikz}
\eq[]{\bialg}\input{./figures/\fig/\fig_02.tikz}
\eq[]{\s\\\ref{lem:triangles-fork-absorbs-anti-CNOT}}\input{./figures/\fig/\fig_03.tikz}
\eq[]{\ref{lem:k1}}\input{./figures/\fig/\fig_04.tikz}
\eq[]{\ref{lem:black-dot-swappable-outputs}\\\ref{lem:k1}\\\ref{lem:not-triangle-is-symmetrical}}\input{./figures/\fig/\fig_05.tikz}
\end{align}
%\step \label{step:three}
\def\fig{anti-control-inverse-triangle-is-control-triangle-times-inverse-triangle-ter}
\begin{align}\label{step:three}
\input{./figures/\fig/\fig_00.tikz}
\eq{\ref{lem:k1}}\input{./figures/\fig/\fig_01.tikz}
\eq{\ref{lem:n-four}\\\ref{lem:black-dot-swappable-outputs}}\input{./figures/\fig/\fig_02.tikz}
\eq{\ref{lem:not-ug-is-symmetrical}}\input{./figures/\fig/\fig_03.tikz}
\eq{\ref{lem:k1}}\input{./figures/\fig/\fig_04.tikz}
\end{align}
%\def\fig{anti-control-inverse-triangle-is-control-triangle-times-inverse-triangle}
%\begin{align}\label{step:three}
%\tikzfigc{00}
%\eq{\ref{lem:not-ug-is-symmetrical}}\tikzfigc{01}
%\eq{\ref{lem:n-four}}\tikzfigc{02}
%\eq{\ref{lem:k1}}\tikzfigc{03}
%\end{align}
%%\step \label{step:three-bis}
%\def\fig{anti-control-inverse-triangle-is-control-triangle-times-inverse-triangle-bis}
%\begin{align}\label{step:three-bis}
%\tikzfigc{00}
%\eq{\ref{lem:k1}}\tikzfigc{01}
%\eq{\ref{step:three}}\tikzfigc{02}
%\eq{\ref{lem:k1}}\tikzfigc{03}
%\end{align}
%\step \label{step:four}
\def\fig{anti-control-triangle-and-CNOT-commute}
\begin{align}\label{step:four}
\input{./figures/\fig/\fig_00.tikz}
\eq{\ref{lem:black-dot-swappable-outputs}}\input{./figures/\fig/\fig_01.tikz}
\eq{\bialg}\input{./figures/\fig/\fig_02.tikz}
\eq{\ref{lem:k1}}\input{./figures/\fig/\fig_03.tikz}\\\notag
\eq{\ref{lem:triangles-fork-absorbs-anti-CNOT}}\input{./figures/\fig/\fig_04.tikz}
\eq{\ref{lem:k1}}\input{./figures/\fig/\fig_05.tikz}
\end{align}
%\step \label{step:five}
\def\fig{anti-control-inverse-triangle-and-CNOT-commute}
\begin{align}\label{step:five}
\input{./figures/\fig/\fig_00.tikz}
\eq{\ref{lem:k1}}\input{./figures/\fig/\fig_01.tikz}
\eq{\ref{step:four}}\input{./figures/\fig/\fig_02.tikz}
\eq{\ref{lem:k1}}\input{./figures/\fig/\fig_03.tikz}
\end{align}
%\step \label{step:six}
\def\fig{control-not-ug-is-symmetrical-proof}
\begin{align}\label{step:six}
\input{./figures/\fig/\fig_00.tikz}
\eq{\ref{step:one}\\\ref{lem:black-dot-swappable-outputs}}\input{./figures/\fig/\fig_01.tikz}
\eq{\ref{lem:n-five}}\input{./figures/\fig/\fig_02.tikz}
\eq{\ref{step:three}\\\ref{lem:k1}\\\ref{lem:hopf}}\input{./figures/\fig/\fig_03.tikz}\\\notag
\eq{\ref{lem:k1}\\\ref{step:five}}\input{./figures/\fig/\fig_04.tikz}
\eq{\ref{lem:k1}\\\ref{step:three}}\input{./figures/\fig/\fig_05.tikz}
\eq{\ref{lem:hopf}\\\ref{step:two}}\input{./figures/\fig/\fig_06.tikz}
\end{align}
%\step \label{step:seven}
\def\fig{2-diagrams-of-control-triangle-proof}
\begin{align}\label{step:seven}
\input{./figures/\fig/\fig_00.tikz}
\eq{\ref{lem:k1}\\\ref{lem:not-triangle-is-symmetrical}}\input{./figures/\fig/\fig_01.tikz}
\eq{\bialg}\input{./figures/\fig/\fig_02.tikz}
\eq{}\input{./figures/\fig/\fig_03.tikz}
\eq{\bialg}\input{./figures/\fig/\fig_04.tikz}\\\notag
\eq{\ref{lem:looped-triangle}\\\ref{lem:black-dot-swappable-outputs}}\input{./figures/\fig/\fig_05.tikz}
\eq{\ref{lem:2-is-sqrt2-squared}\\\ref{lem:triangles-fork-absorbs-anti-CNOT}}\input{./figures/\fig/\fig_06.tikz}
\eq{\ref{lem:2-is-sqrt2-squared}\\\bialg}\input{./figures/\fig/\fig_07.tikz}
\eq{\ref{lem:looped-triangle}\\\ref{lem:black-dot-swappable-outputs}}\input{./figures/\fig/\fig_08.tikz}
\end{align}
Finally,
\def\fig{rule-5a-proof}
\begin{align*}
\input{./figures/\fig/\fig_00.tikz}~~\mapsto~~\input{./figures/\fig/\fig_01.tikz}
\eq{}\input{./figures/\fig/\fig_02.tikz}
\eq{\ref{lem:hopf}\\\ref{lem:2-is-sqrt2-squared}\\\ref{lem:control-pi-and-anti-CNOT-commute}}\input{./figures/\fig/\fig_03.tikz}
\eq{\ref{step:seven} \\\ref{lem:hopf}}\input{./figures/\fig/\fig_04.tikz}\\
\eq{\ref{step:six}}\input{./figures/\fig/\fig_05.tikz}
\eq{\bialg}\input{./figures/\fig/\fig_06.tikz}
\eq{\ref{step:seven}\\\ref{lem:hopf}}\input{./figures/\fig/\fig_07.tikz}~~\mapsfrom\!\!\input{./figures/\fig/\fig_08.tikz}=\input{./figures/\fig/\fig_09.tikz}
\end{align*}
\item $5b$:
\def\fig{rule-5b-proof}
\begin{align*}
\input{./figures/\fig/\fig_00.tikz}~~\underset{\ref{lem:arity-1-black-dot}}{\mapsto}~~\input{./figures/\fig/\fig_01.tikz}
\eq{\ref{lem:inverse}\\\cp}\input{./figures/\fig/\fig_02.tikz}
\eq{\ref{lem:red-state-on-triangle}}\input{./figures/\fig/\fig_03.tikz}
\eq{\ref{lem:inverse}\\\cp}\input{./figures/\fig/\fig_04.tikz}~~\underset{\ref{lem:arity-1-black-dot}}{\mapsfrom}~~\input{./figures/\fig/\fig_05.tikz}
\end{align*}
\item $5c$:
\def\fig{rule-5c-proof}
\begin{align*}
\input{./figures/\fig/\fig_00.tikz}~~\underset{\ref{lem:arity-1-black-dot}}{\mapsto}~~\input{./figures/\fig/\fig_01.tikz}
\eq{\ref{lem:2-is-sqrt2-squared}}\input{./figures/\fig/\fig_02.tikz}
\eq{\ref{lem:inverse}}\input{./figures/\fig/\fig_03.tikz}~~\mapsfrom~~\input{./figures/\fig/\fig_03.tikz}
\end{align*}
\item $5d$:
\def\fig{rule-5d-proof}
\begin{align*}
\input{./figures/\fig/\fig_00.tikz}~~\mapsto~~\input{./figures/\fig/\fig_01.tikz}
\eq{\ref{lem:hopf}}\input{./figures/\fig/\fig_02.tikz}
\eq{\id}\input{./figures/\fig/\fig_03.tikz}
\eq{\ref{lem:triangles-fork-absorbs-anti-CNOT}}\input{./figures/\fig/\fig_04.tikz}
\eq{\ref{lem:inverse}\\\ref{lem:hopf}}\input{./figures/\fig/\fig_05.tikz}\\
\eq{\ref{lem:pi-red-state-on-triangle}\\\ref{lem:inverse}}\input{./figures/\fig/\fig_06.tikz}
\eq{\ref{lem:inverse}\\\ref{lem:pi-green-state-on-upside-down-triangle}}\input{./figures/\fig/\fig_07.tikz}
\eq{\ref{lem:inverse}\\\cp}\input{./figures/\fig/\fig_08.tikz}~~\underset{\ref{lem:arity-1-black-dot}}{\mapsfrom}~~\input{./figures/\fig/\fig_09.tikz}
\end{align*}
\item $6a$: Thanks to the rule X we can get rid of \resizebox{!}{1em}{
%\InputIfFileExists{#1.tikz}{}{
\input{./figures/control-pi.tikz}%} % chktex 27
} induced by the crossing. Then,
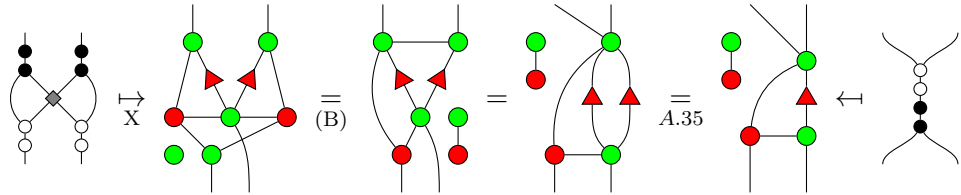
\def\fig{rule-6a-proof-no-braid}
\begin{align*}
\input{./figures/\fig/\fig_00.tikz}~~\underset{\text{X}}{\mapsto}~~\input{./figures/\fig/\fig_01.tikz}
\eq{\bialg}\input{./figures/\fig/\fig_02.tikz}
\eq{}\input{./figures/\fig/\fig_03.tikz}
\eq{\ref{lem:parallel-triangles}}\input{./figures/\fig/\fig_04.tikz}~~\mapsfrom~~\input{./figures/\fig/\fig_05.tikz}
\end{align*}
\item $6b$ is exactly the copy rule \cp.
\item $6c$:
\def\fig{rule-6c-proof}
\begin{align*}
\input{./figures/\fig/\fig_00.tikz}~~\mapsto~~\input{./figures/\fig/\fig_01.tikz}
\eq{\ref{lem:inverse}\\\ref{lem:hopf}}\input{./figures/\fig/\fig_02.tikz}
\eq{\ref{lem:inverse}\\\cp\\\id}\input{./figures/\fig/\fig_03.tikz}
\eq{\ref{lem:red-state-on-triangle}}\input{./figures/\fig/\fig_04.tikz}~~\underset{\ref{lem:arity-1-black-dot}}{\mapsfrom}\input{./figures/\fig/\fig_05.tikz}
\end{align*}
\item $7a$:
\def\fig{rule-7a-proof-no-braid}
\begin{align*}
\input{./figures/\fig/\fig_00.tikz}~~\mapsto~~\input{./figures/\fig/\fig_01.tikz}
\eq[]{\s\\\h}\input{./figures/\fig/\fig_02.tikz}
\eq[]{\bialg}\input{./figures/\fig/\fig_03.tikz}
\eq[]{\h}\input{./figures/\fig/\fig_04.tikz}~~\mapsfrom~~\input{./figures/\fig/\fig_05.tikz}
\end{align*}
\item $7b$: using~\ref{lem:k1}, \h and \s:
\def\fig{rule-7b-proof-no-braid}
\begin{align*}
\input{./figures/\fig/\fig_00.tikz}~~\mapsto~~\input{./figures/\fig/\fig_01.tikz}
\eq{\ref{lem:k1}\\\h}\input{./figures/\fig/\fig_02.tikz}
~~\mapsfrom~~\input{./figures/\fig/\fig_04.tikz}
\end{align*}
\item R$_2$:
\def\fig{rule-r2-proof}
\begin{align*}
\input{./figures/\fig/\fig_00.tikz}~~\mapsto~~\input{./figures/\fig/\fig_01.tikz}
\eq{}\input{./figures/\fig/\fig_02.tikz}
\eq{\ref{lem:hopf}}\input{./figures/\fig/\fig_03.tikz}~~\mapsfrom~~\input{./figures/\fig/\fig_03.tikz}
\end{align*}
\item R$_3$:
\def\fig{rule-r3-proof}
\begin{align*}
\input{./figures/\fig/\fig_00.tikz}~~\mapsto~~\input{./figures/\fig/\fig_01.tikz}
\eq{}\input{./figures/\fig/\fig_02.tikz}
\eq{}\input{./figures/\fig/\fig_03.tikz}~~\mapsfrom~~\input{./figures/\fig/\fig_04.tikz}
\end{align*}
\item $iv$: using \id, \s,~\ref{lem:2-is-sqrt2-squared} and~\ref{lem:inverse},
\def\fig{rule-half-proof}
\begin{align*}
\input{./figures/\fig/\fig_00.tikz}~~\mapsto~~\input{./figures/\fig/\fig_01.tikz}
\eq{\id}\input{./figures/\fig/\fig_02.tikz}
\eq{}\input{./figures/\fig/\fig_03.tikz}
\eq{\ref{lem:2-is-sqrt2-squared}\\\ref{lem:inverse}}\input{./figures/\fig/\fig_04.tikz}
~~\mapsfrom~~\input{./figures/\fig/\fig_05.tikz}
\end{align*}

\subsection{Proof of Proposition~\ref{prop:left-inverse-ZX}}%
\label{prf:left-inverse-ZX}

%Let us write $\dblinterp{.}=\interpwx{\interpxw{.}}$.
We can show inductively that:
\[\zxct\vdash \left[[D]_W\right]_X\circ\left(
%\InputIfFileExists{#1.tikz}{}{
\input{./figures/top-composition.tikz}%} % chktex 27
\right) = D\otimes \left(
%\InputIfFileExists{#1.tikz}{}{
\input{./figures/theta.tikz}%} % chktex 27
\right)\]
which is the expression of Proposition~\ref{prop:left-inverse}.
%\item \tikzfig{ug-node}: Using \so, \st,~\ref{lem:looped-triangle},
%\[ \dblinterp{\tikzfig{ug-node}} = \interpwx{~\tikzfig{u-diagram-interpretation}~} = \tikzfig{u-diag-double-interpretation} \]
%from which we easily derive $\zxt\vdash \dblinterp{\tikzfig{ug-node}}\hspace{0em}\circ\left(\tikzfig{top-composition-1}\right) ~=~~ \tikzfig{ug-node}\hspace{0.1em}\raisebox{-0.7em}{\tikzfig{theta}}$\\
\item The result is obvious for the generators 
%\InputIfFileExists{#1.tikz}{}{
\input{./figures/empty-diagram.tikz}%} % chktex 27
, 
%\InputIfFileExists{#1.tikz}{}{
\begin{tikzpicture}
	\begin{pgfonlayer}{nodelayer}
		\node [style=none] (0) at (0, 0.2499999) {};
		\node [style=none] (1) at (0, -0.2499999) {};
	\end{pgfonlayer}
	\begin{pgfonlayer}{edgelayer}
		\draw (0.center) to (1.center);
	\end{pgfonlayer}
\end{tikzpicture}%} % chktex 27
, 
%\InputIfFileExists{#1.tikz}{}{
\input{./figures/crossing.tikz}%} % chktex 27
, 
%\InputIfFileExists{#1.tikz}{}{
\begin{tikzpicture}
	\begin{pgfonlayer}{nodelayer}
		\node [style=none] (0) at (-0.2500001, -0) {};
		\node [style=none] (1) at (0.2500001, -0) {};
	\end{pgfonlayer}
	\begin{pgfonlayer}{edgelayer}
		\draw [bend left=90, looseness=1.75] (0.center) to (1.center);
	\end{pgfonlayer}
\end{tikzpicture}%} % chktex 27
, and 
%\InputIfFileExists{#1.tikz}{}{
\begin{tikzpicture}
	\begin{pgfonlayer}{nodelayer}
		\node [style=none] (0) at (-0.2500001, 0.2500001) {};
		\node [style=none] (1) at (0.2500001, 0.2500001) {};
	\end{pgfonlayer}
	\begin{pgfonlayer}{edgelayer}
		\draw [bend right=90, looseness=1.75] (0.center) to (1.center);
	\end{pgfonlayer}
\end{tikzpicture}%} % chktex 27
.
\item 
%\InputIfFileExists{#1.tikz}{}{
\begin{tikzpicture}
	\begin{pgfonlayer}{nodelayer}
		\node [style={H box}] (0) at (0, 0) {};
		\node [style=none] (1) at (0, 0.5) {};
		\node [style=none] (2) at (0, -0.5) {};
	\end{pgfonlayer}
	\begin{pgfonlayer}{edgelayer}
		\draw (2.center) to (1.center);
	\end{pgfonlayer}
\end{tikzpicture}
%} % chktex 27
:
\def\fig{hadamard-double-interpretation-cliff-t}
\begin{align*}
\left[~
%\InputIfFileExists{#1.tikz}{}{
\input{./figures/hadamard-interpretation-2-no-braid.tikz}%} % chktex 27
~\right]_X\hspace{-1em} \eq{} \input{./figures/\fig/\fig_00.tikz}
\eq{\ref{lem:inverse}}\input{./figures/\fig/\fig_01.tikz}
\eq{\h\\\ref{lem:k1}}\input{./figures/\fig/\fig_02.tikz}
\end{align*}
and, using \s, \hd,~\ref{lem:inverse}, \h,~\ref{lem:hopf} and~\ref{lem:green-state-pi_2-is-red-state-minus-pi_2}:
\def\fig{sqrt2-double-interp-times-theta}
\begin{align*}
\input{./figures/\fig/\fig_00.tikz}
\eq{\picom\\\id\\\hd}\input{./figures/\fig/\fig_01.tikz}
\eq{\h}\input{./figures/\fig/\fig_02.tikz}
\eq{\ref{lem:inverse}\\\ref{lem:hopf}}\input{./figures/\fig/\fig_03.tikz}
\eq{\ref{lem:green-state-pi_2-is-red-state-minus-pi_2}}\input{./figures/\fig/\fig_04.tikz}
\end{align*}
Hence $ \zxct\vdash \left[\left[
%\InputIfFileExists{#1.tikz}{}{
\begin{tikzpicture}
	\begin{pgfonlayer}{nodelayer}
		\node [style={H box}] (0) at (0, 0) {};
		\node [style=none] (1) at (0, 0.5) {};
		\node [style=none] (2) at (0, -0.5) {};
	\end{pgfonlayer}
	\begin{pgfonlayer}{edgelayer}
		\draw (2.center) to (1.center);
	\end{pgfonlayer}
\end{tikzpicture}
%} % chktex 27
\right]_W\right]_X\hspace{0em}\circ\left(
%\InputIfFileExists{#1.tikz}{}{
\input{./figures/top-composition-1.tikz}%} % chktex 27
\right) ~=~~ 
%\InputIfFileExists{#1.tikz}{}{
\begin{tikzpicture}
	\begin{pgfonlayer}{nodelayer}
		\node [style={H box}] (0) at (0, 0) {};
		\node [style=none] (1) at (0, 0.5) {};
		\node [style=none] (2) at (0, -0.5) {};
	\end{pgfonlayer}
	\begin{pgfonlayer}{edgelayer}
		\draw (2.center) to (1.center);
	\end{pgfonlayer}
\end{tikzpicture}
%} % chktex 27
\hspace{0.1em}\raisebox{-0.7em}{
%\InputIfFileExists{#1.tikz}{}{
\input{./figures/theta.tikz}%} % chktex 27
}$
\item $\left[~
%\InputIfFileExists{#1.tikz}{}{
\input{./figures/gn-0-1-m-interpretation.tikz}%} % chktex 27
~~\right]_X\circ\left(
%\InputIfFileExists{#1.tikz}{}{
\input{./figures/top-composition-1.tikz}%} % chktex 27
\right) = 
%\InputIfFileExists{#1.tikz}{}{
\input{./figures/gn-0-1-m-theta.tikz}%} % chktex 27
$
\vspace{0.5em}
\item $\left[~
%\InputIfFileExists{#1.tikz}{}{
\input{./figures/gn-0-n-1-interpretation.tikz}%} % chktex 27
~~\right]_X\circ\left(
%\InputIfFileExists{#1.tikz}{}{
\input{./figures/top-composition.tikz}%} % chktex 27
\right) = 
%\InputIfFileExists{#1.tikz}{}{
\input{./figures/gn-0-n-1-theta.tikz}%} % chktex 27
$
\vspace{0.5em}
\item 
%\InputIfFileExists{#1.tikz}{}{
\begin{tikzpicture}
	\begin{pgfonlayer}{nodelayer}
		\node [style=gn] (0) at (0, -0) {$\frac{\pi}{4}$};
		\node [style=none] (1) at (0, 0.5000001) {};
		\node [style=none] (2) at (0, -0.5000001) {};
	\end{pgfonlayer}
	\begin{pgfonlayer}{edgelayer}
		\draw (1.center) to (2.center);
	\end{pgfonlayer}
\end{tikzpicture}%} % chktex 27
: % chktex 8
\def\fig{gn-pi_4-double-interpretation}
\begin{align*}
\left[
%\InputIfFileExists{#1.tikz}{}{
\input{./figures/gn-pi_4-interpretation-2-no-braid.tikz}%} % chktex 27
\right]_X \!\!\!\!\eq{} \input{./figures/\fig/\fig_00.tikz}
\eq{\ref{lem:inverse}\\\h\\\ref{lem:k1}\\\cp}\input{./figures/\fig/\fig_01.tikz}
\end{align*}
But:
\def\fig{gn-state-pi_2-on-CRG-is-control-pi_2}
\begin{align*}
\input{./figures/\fig/\fig_00.tikz}
\eq{\ref{def:triangle}}\input{./figures/\fig/\fig_01.tikz}
\eq{\ref{lem:supp-to-minus-pi_4}\\\ref{lem:euler-decomp-with-scalar}}\input{./figures/\fig/\fig_02.tikz}
\eq{\h}\input{./figures/\fig/\fig_03.tikz}\\
\eq{\ref{lem:C1-original}}\input{./figures/\fig/\fig_04.tikz}
\eq{\h\\\cp\\\id\\\picom}\input{./figures/\fig/\fig_05.tikz}
\end{align*}
So that:
\def\fig{gn-pi_4-double-interp-times-theta}
\begin{align*}
\input{./figures/\fig/\fig_00.tikz}
\eq{\hd}\input{./figures/\fig/\fig_01.tikz}
\eq{\h\\\ref{lem:hopf}\\\ref{lem:not-triangle-is-symmetrical}}\input{./figures/\fig/\fig_02.tikz}
\eq{}\input{./figures/\fig/\fig_03.tikz}\\
\eq{\bialg}\input{./figures/\fig/\fig_04.tikz}
\eq{}\input{./figures/\fig/\fig_05.tikz}
\eq{\cp}\input{./figures/\fig/\fig_06.tikz}
\end{align*}
which means $ \zxct\vdash \left[\left[
%\InputIfFileExists{#1.tikz}{}{
\begin{tikzpicture}
	\begin{pgfonlayer}{nodelayer}
		\node [style=gn] (0) at (0, -0) {$\frac{\pi}{4}$};
		\node [style=none] (1) at (0, 0.5000001) {};
		\node [style=none] (2) at (0, -0.5000001) {};
	\end{pgfonlayer}
	\begin{pgfonlayer}{edgelayer}
		\draw (1.center) to (2.center);
	\end{pgfonlayer}
\end{tikzpicture}%} % chktex 27
\right]_W\right]_X\circ\left(
%\InputIfFileExists{#1.tikz}{}{
\input{./figures/top-composition-1.tikz}%} % chktex 27
\right) ~=~~ 
%\InputIfFileExists{#1.tikz}{}{
\begin{tikzpicture}
	\begin{pgfonlayer}{nodelayer}
		\node [style=gn] (0) at (0, -0) {$\frac{\pi}{4}$};
		\node [style=none] (1) at (0, 0.5000001) {};
		\node [style=none] (2) at (0, -0.5000001) {};
	\end{pgfonlayer}
	\begin{pgfonlayer}{edgelayer}
		\draw (1.center) to (2.center);
	\end{pgfonlayer}
\end{tikzpicture}%} % chktex 27
\hspace{0.1em}\raisebox{-0.7em}{
%\InputIfFileExists{#1.tikz}{}{
\input{./figures/theta.tikz}%} % chktex 27
}$
\item $D_1\circ D_2$:

\noindent
It is to be noticed that $[D_1\circ D_2]_X=[D_1]_X\circ [D_2]_X$ and $[D_1\otimes D_2]_X=[D_1]_X\otimes [D_2]_X$.

\noindent
Let us write $\omega = 
%\InputIfFileExists{#1.tikz}{}{
\input{./figures/theta.tikz}%} % chktex 27
$. Then:
\begin{align*}
\zxct\vdash \left[[D_1\circ D_2]_W\right]_X\circ(\mathbb{I}\otimes \omega)
&= \left[[D_1]_W\right]_X\circ \left[[D_2]_W\right]_X\circ(\mathbb{I}\otimes \omega)\\
&=  \left[[D_1]_W\right]_X\circ (D_2\otimes \omega) \\&= \left[[D_1]_W\right]_X\circ(\mathbb{I}\otimes \omega)\circ D_2 = (D_1\otimes\omega)\circ D_2 \\
&= (D_1\circ D_2)\otimes \omega
\end{align*}
\item $D_1\otimes D_2$:
\def\fig{interp-tensor-product-times-theta-graphical}
\begin{align*}
\zxct\vdash&~~\left[[D_1\otimes D_2]_W\right]_X\circ(\mathbb{I}\otimes \omega)
\eq{}\input{./figures/\fig/\fig_00.tikz}\\
&\eq{}\input{./figures/\fig/\fig_01.tikz}
\eq{}\input{./figures/\fig/\fig_02.tikz}\\
&\eq{}\input{./figures/\fig/\fig_03.tikz}
\eq{}\input{./figures/\fig/\fig_04.tikz}
\eq{}D_1\otimes D_2\otimes\omega
\end{align*}
By compositions, for any diagram $D$, $\zxct\vdash \left[[D]_W\right]_X\circ(\mathbb{I}\otimes \omega) = D\otimes\omega$. Then, using Lemmas~\ref{lem:bicolor-0-alpha} and~\ref{lem:inverse}:
\begin{align*}
\forall D\in \zxct,\quad
\zxct\vdash \left(\hspace{-0.2em}\scalebox{0.8}{
%\InputIfFileExists{#1.tikz}{}{
\input{./figures/bottom-composition.tikz}%} % chktex 27
}\hspace{-0.3em}\right)\circ\left[[D]_W\right]_X\circ\left(\hspace{-0.2em}\scalebox{0.8}{
%\InputIfFileExists{#1.tikz}{}{
\input{./figures/top-composition.tikz}%} % chktex 27
}\hspace{-0.3em}\right)=D
\end{align*}
\end{itemize}

%========================================

\section{Linear Diagrams}

\subsection{Proof of Lemma~\ref{lem:alphas-on-X}}%
\label{prf:alphas-on-X}
%\begin{proof}[Proof of Lemma~\ref{lem:alphas-on-X}]
\def\fig{second-matrix-gn-alpha-proof-arxiv}
\begin{align*}
\input{./figures/\fig/\fig_00.tikz}
\eq{\h\\\ref{lem:euler-decomp-with-scalar}}\input{./figures/\fig/\fig_01.tikz}
\eq{\bialg}\input{./figures/\fig/\fig_02.tikz}\\
\eq{\h}\input{./figures/\fig/\fig_03.tikz}
\eq{\ref{lem:C1-original}}\input{./figures/\fig/\fig_04.tikz}
\eq{\h}\input{./figures/\fig/\fig_05.tikz}\\
\eq{\h\\\ref{lem:hopf}}\input{./figures/\fig/\fig_06.tikz}
\eq{\ref{lem:green-state-pi_2-is-red-state-minus-pi_2}\\\picom\\\ref{lem:multiplying-global-phases}}\input{./figures/\fig/\fig_07.tikz}
\eq{\ref{lem:supp-to-minus-pi_4}\\\ref{lem:multiplying-global-phases}}\input{./figures/\fig/\fig_08.tikz}
\eq{\id\\\ref{lem:hopf}}\input{./figures/\fig/\fig_09.tikz}
\end{align*}
\qed%\end{proof}

\subsection{Details of the Proof for Corollary~\ref{cor:distribution}}%
%\begin{proof}[Corollary~\ref{cor:distribution}]
%\phantomsection
\label{prf:distribution}
We first plug the basis $\left(\rx{}, \rx{$\pi$}\right)$ in the input:
\begin{itemize}
\item \rx{}
\begin{itemize}
\item Left hand side:
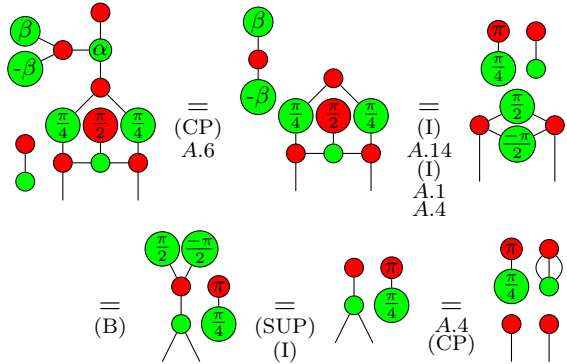
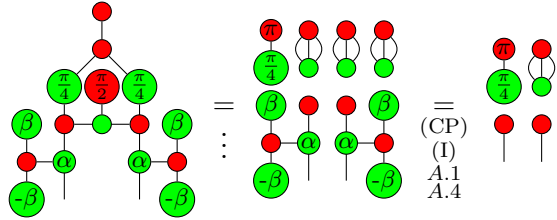
\def\fig{add-axiom-2-l-0-arxiv}
\begin{align*}
\input{./figures/\fig/\fig_00.tikz}
\eq{\cp\\\ref{lem:bicolor-0-alpha}}\input{./figures/\fig/\fig_01.tikz}
\eq{\id\\\ref{lem:green-state-pi_2-is-red-state-minus-pi_2}\\\id\\\ref{lem:2-is-sqrt2-squared}\\\ref{lem:inverse}}\input{./figures/\fig/\fig_02.tikz}\\
\eq{\bialg}\input{./figures/\fig/\fig_03.tikz}
\eq{\supp\\\id}\input{./figures/\fig/\fig_04.tikz}
\eq{\ref{lem:inverse}\\\cp}\input{./figures/\fig/\fig_05.tikz}
\end{align*}
\item Right hand side:
\def\fig{add-axiom-2-r-0}
\begin{align*}
\input{./figures/\fig/\fig_00.tikz}
\eq{\vdots}\input{./figures/\fig/\fig_01.tikz}
\eq{\cp\\\id\\\ref{lem:2-is-sqrt2-squared}\\\ref{lem:inverse}}\input{./figures/\fig/\fig_02.tikz}
\end{align*}
where the first step was already developed for the left hand side.
\end{itemize}
The resulting two diagrams are equal when \rx{} is plugged.
\item \rx{$\pi$}
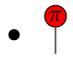
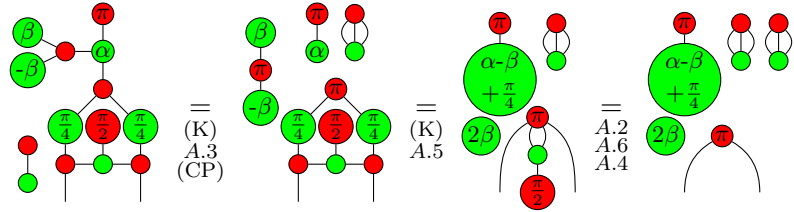
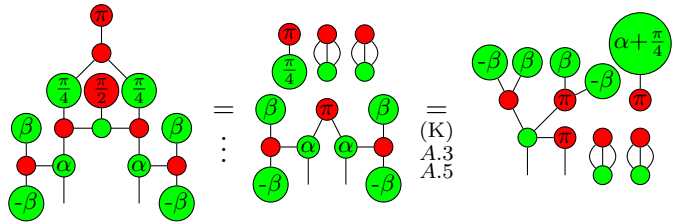
\begin{itemize}
\item Left hand side:
\def\fig{add-axiom-2-l-pi}
\begin{align*}
\input{./figures/\fig/\fig_00.tikz}
\eq{\picom\\\ref{lem:k1}\\\cp}\input{./figures/\fig/\fig_01.tikz}
\eq{\picom\\\ref{lem:multiplying-global-phases}}\input{./figures/\fig/\fig_02.tikz}
\eq{\ref{lem:hopf}\\\ref{lem:bicolor-0-alpha}\\\ref{lem:inverse}}\input{./figures/\fig/\fig_03.tikz}
\end{align*}
\item Right hand side:
\def\fig{add-axiom-2-r-pi}
\begin{align*}
\input{./figures/\fig/\fig_00.tikz}
\eq{\vdots}\input{./figures/\fig/\fig_01.tikz}
\eq{\picom\\\ref{lem:k1}\\\ref{lem:multiplying-global-phases}}\input{./figures/\fig/\fig_02.tikz}
\end{align*}
where again the first step was already developed for the left hand side.

Now we could have concluded directly with the help of Corollaries~\ref{cor:gen-supp} and~\ref{cor:big-scalar-equation}. For the sake of the example, though, we are going to plug our basis on, say, the left hanging branch:
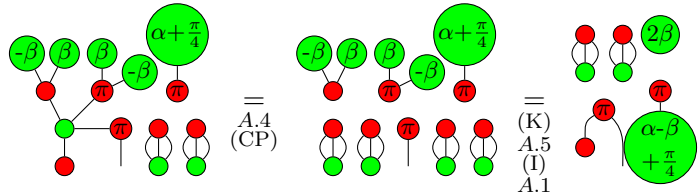
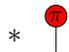
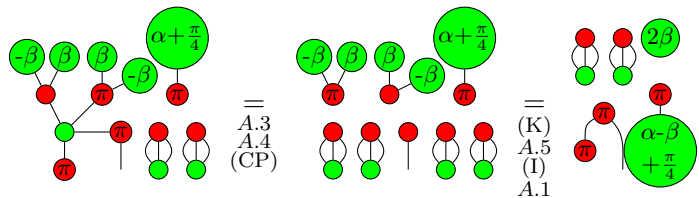
\begin{itemize}
\item \rx{}
\def\fig{add-axiom-2-r-pi-0}
\begin{align*}
\input{./figures/\fig/\fig_00.tikz}
\eq{\ref{lem:inverse}\\\cp}\input{./figures/\fig/\fig_01.tikz}
\eq{\picom\\\ref{lem:multiplying-global-phases}\\\id\\\ref{lem:2-is-sqrt2-squared}}\input{./figures/\fig/\fig_03.tikz}
\end{align*}
\item \rx{$\pi$}
\def\fig{add-axiom-2-r-pi-pi}
\begin{align*}
\input{./figures/\fig/\fig_00.tikz}
\eq{\ref{lem:k1}\\\ref{lem:inverse}\\\cp}\input{./figures/\fig/\fig_01.tikz}
\eq{\picom\\\ref{lem:multiplying-global-phases}\\\id\\\ref{lem:2-is-sqrt2-squared}}\input{./figures/\fig/\fig_03.tikz}
\end{align*}
\end{itemize}
\end{itemize}
\end{itemize}
Hence, the two initial diagrams result in the same diagram when the basis is applied. Thanks to Theorem~\ref{thm:basis}, the ZX-Calculus proves the equality between the two initial diagrams.
\qed%\end{proof}

\subsection{Details of the Proof for Corollary~\ref{cor:big-scalar-equation}}%
%\begin{proof}[Corollary~\ref{cor:big-scalar-equation}]
%\phantomsection
\label{prf:big-scalar-equation}\hfill
\begin{itemize}
\item $\alpha=0$:
	\begin{itemize}
	\item Left hand side:
		\def\fig{big-scalar-equation-l-0}
		\begin{align*}
		\input{./figures/\fig/\fig_00.tikz}
		\eq{\ref{lem:inverse}\\\cp}\input{./figures/\fig/\fig_01.tikz}
		\eq{\ref{lem:2-is-sqrt2-squared}\\\ref{lem:inverse}}\input{./figures/\fig/\fig_02.tikz}
		\end{align*}
	\item Right hand side:
		\def\fig{big-scalar-equation-r-0}
		\begin{align*}
		\input{./figures/\fig/\fig_00.tikz}
		\eq{\picom\\\ref{lem:multiplying-global-phases}}\input{./figures/\fig/\fig_01.tikz}
		\end{align*}
	\end{itemize}
\item $\alpha=\pi$:
	\begin{itemize}
	\item Left hand side:
		\def\fig{big-scalar-equation-l-pi-arxiv}
		\begin{align*}
		\input{./figures/\fig/\fig_00.tikz}
		\eq{\picom\\\ref{lem:k1}\\\ref{lem:inverse}\\\cp}\input{./figures/\fig/\fig_01.tikz}
		\eq{\ref{lem:2-is-sqrt2-squared}\\\ref{lem:inverse}}\input{./figures/\fig/\fig_02.tikz}
		\end{align*}
	\item Right hand side:
		\def\fig{big-scalar-equation-r-pi}
		\begin{align*}
		\input{./figures/\fig/\fig_00.tikz}
		\eq{\picom\\\ref{lem:multiplying-global-phases}}\input{./figures/\fig/\fig_01.tikz}
		\end{align*}
	\end{itemize}
\item $\alpha=\frac{\pi}{2}$:
	\begin{itemize}
	\item Left hand side:
		\def\fig{big-scalar-equation-l-pi_2}
		\begin{align*}
		\input{./figures/\fig/\fig_00.tikz}
		\eq{\ref{lem:inverse}\\\picom}\input{./figures/\fig/\fig_01.tikz}
		\eq{\supp}\input{./figures/\fig/\fig_02.tikz}
		\eq{\id}\input{./figures/\fig/\fig_03.tikz}
		\end{align*}
	\item Right hand side:
		\def\fig{big-scalar-equation-r-pi_2-arxiv}
		\begin{align*}
		\input{./figures/\fig/\fig_00.tikz}
		\eq{\picom\\\ref{lem:multiplying-global-phases}\\\cp}\input{./figures/\fig/\fig_01.tikz}
		\eq{\supp}\input{./figures/\fig/\fig_02.tikz}
		\eq{\ref{lem:inverse}\\\cp}\input{./figures/\fig/\fig_03.tikz}
		\end{align*}
	\end{itemize}
\end{itemize}
The results are the same for three different values of $\alpha$. This is enough to get the equation in Corollary~\ref{cor:big-scalar-equation}, according to Theorem~\ref{thm:valuations}.
\qed%\end{proof}

%\subsection{Rules of the ZW-Calculus}
%\label{sec:rules-zw}
%~
%\begin{center}
%\def\scale{0.9}
%\scalebox{\scale}{\tikzfig{ZW-rule-0-no-braid}}\\~\\~\\
%\scalebox{\scale}{\tikzfig{ZW-rule-1}} \\~\\~\\
%\scalebox{\scale}{\tikzfig{ZW-rule-2-no-braid}}\\~\\~\\
%\scalebox{\scale}{\tikzfig{ZW-rule-3}} \\~\\~\\
%\scalebox{\scale}{\tikzfig{ZW-rule-4}} \\~\\~\\
%\scalebox{\scale}{\tikzfig{ZW-rule-5-no-braid}} \\~\\~\\
%\scalebox{\scale}{\tikzfig{ZW-rule-6-a-b}} \\~\\~\\
%\scalebox{\scale}{\tikzfig{ZW-rule-6-c}} \\~\\~\\
%\scalebox{\scale}{\tikzfig{ZW-rule-7-no-braid}}\\~\\~\\
%\scalebox{\scale}{\tikzfig{ZW-rule-X-no-braid}}\\~\\~\\
%\scalebox{\scale}{\tikzfig{reidmeister-3}}
%\end{center}

\section{General ZX-Calculus}

\subsection{Proof of Proposition~\ref{prop:double-interpretation-equivalence-1}}%
\label{prf:double-interp-eq}
%\begin{proof}[Proposition~\ref{prop:double-interpretation-equivalence-1}]
The result is obvious for cups, caps, single wires, empty diagrams and swaps. Moreover, if we have the result for green dots and the Hadamard gate, then we also have it for red dots by construction.\\
For green dots, since $n=\max\left(0,\left\lceil \log_2(1)\right\rceil\right)=0$, $\beta = \gamma = 0$:
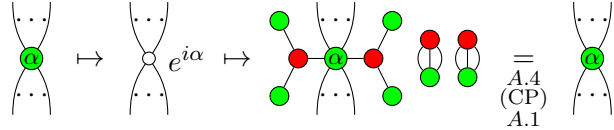
\def\fig{gn-alpha-double-interp}
\begin{align*}
\input{./figures/\fig/\fig_00.tikz}
~~\mapsto~~\input{./figures/\fig/\fig_01.tikz}
~~\mapsto~~\input{./figures/\fig/\fig_02.tikz}
\eq{\ref{lem:inverse}\\\cp\\\ref{lem:2-is-sqrt2-squared}}\input{./figures/\fig/\fig_03.tikz}
\end{align*}
%\[\tikzfig{gn-alpha-double-interp}\]
For Hadamard, first notice:
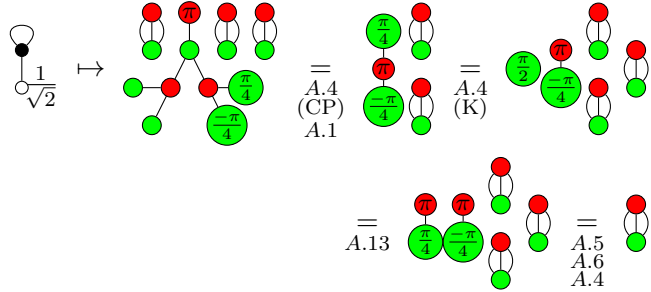
\def\fig{1-over-sqrt2-ZW-to-ZX}
\begin{align*}
\input{./figures/\fig/\fig_00.tikz}
~~\mapsto~~\input{./figures/\fig/\fig_01.tikz}
\eq{\ref{lem:inverse}\\\cp\\\ref{lem:2-is-sqrt2-squared}}\input{./figures/\fig/\fig_02.tikz}
\eq{\ref{lem:inverse}\\\picom}\input{./figures/\fig/\fig_03.tikz}\\
\eq{\ref{lem:gn-pi_2-0-0-equals-sqrt2-exp-pi_4}}\input{./figures/\fig/\fig_04.tikz}
\eq{\ref{lem:multiplying-global-phases}\\\ref{lem:bicolor-0-alpha}\\\ref{lem:inverse}}\input{./figures/\fig/\fig_05.tikz}
\end{align*}
%\[\tikzfig{1-over-sqrt2-ZW-to-ZX}\]
since $n=0$, $\beta=\arccos{1/\sqrt{2}}=\pi/4$, $\gamma=\arccos{1}=0$. Finally:
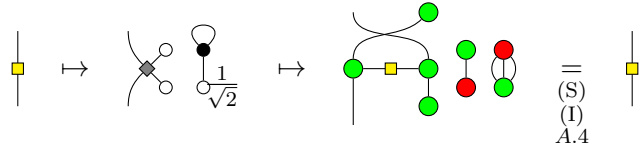
\def\fig{hadamard-double-interpretation}
\begin{align*}
\input{./figures/\fig/\fig_00.tikz}\quad\mapsto\quad\input{./figures/\fig/\fig_01.tikz}\quad\mapsto\quad
\input{./figures/\fig/\fig_02.tikz}\eq[\quad]{\s\\\id\\\ref{lem:inverse}}\input{./figures/\fig/\fig_03.tikz}
\end{align*}
\qed%\end{proof}

\subsection{Lemmas for The General ZX-Calculus}

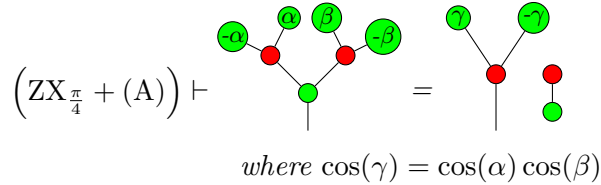
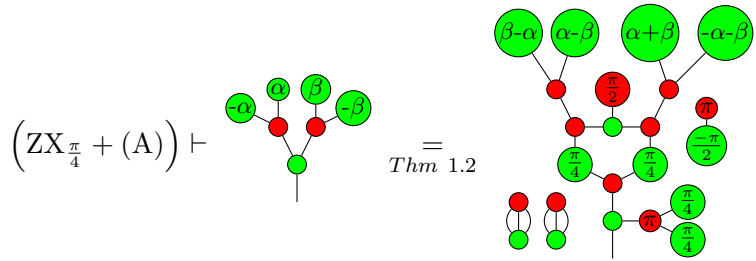
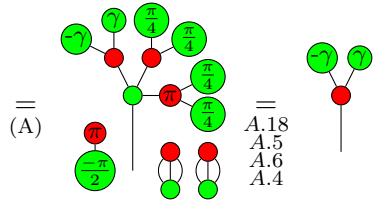
\begin{lem}%
\label{lem:prod-cos}
\[\left(\zxct+\add\right)\vdash
%\InputIfFileExists{#1.tikz}{}{
\input{./figures/add-axiom.tikz}%} % chktex 27
\]
\end{lem}
\begin{proof}
\def\fig{add-axiom-3-to-A1}
\begin{align*}
\left(\zxct+\add\right)\vdash~~\input{./figures/\fig/\fig_00.tikz}
\eq{Thm~\ref{thm:lin-diag}}\input{./figures/\fig/\fig_01.tikz}\\
\eq{\add}\input{./figures/\fig/\fig_02.tikz}
\eq{\ref{lem:supp-to-minus-pi_4}\\\ref{lem:multiplying-global-phases}\\\ref{lem:bicolor-0-alpha}\\\ref{lem:inverse}}\input{./figures/\fig/\fig_03.tikz}
\end{align*}
%\[\left(\zxct+\add\right)\vdash\tikzfig{add-axiom-3-to-A1}\]
where $\cos(\gamma) = \frac{1}{2}(\cos(\alpha-\beta)+ \cos(\alpha+\beta)) = \cos(\alpha)\cos(\beta)$.
\end{proof}

\begin{lem}%
\label{lem:add-bis}
We can deduce an equality similar to the rule \add:
\[\left(\zxct+\add\right)\vdash
%\InputIfFileExists{#1.tikz}{}{
\input{./figures/add-axiom-3-simplified-2.tikz}%} % chktex 27
\]
\end{lem}
\begin{proof}
\def\fig{add-axiom-3-simp-2-proof}
\begin{align*}
\left(\zxct+\add\right)\vdash~~\input{./figures/\fig/\fig_00.tikz}
\eq{\ref{lem:supp-to-minus-pi_4}\\\picom\\\ref{lem:multiplying-global-phases}}\input{./figures/\fig/\fig_01.tikz}
\eq{\ref{lem:prod-cos}}\input{./figures/\fig/\fig_02.tikz}\\
\eq{\ref{lem:prod-cos}}\input{./figures/\fig/\fig_03.tikz}
\eq{\picom\\\ref{lem:multiplying-global-phases}}\input{./figures/\fig/\fig_04.tikz}
\eq{\add}\input{./figures/\fig/\fig_05.tikz}
\end{align*}
%\[\left(\zxct+\add\right)\vdash\tikzfig{add-axiom-3-simp-2-proof}\]
where $\cos(\gamma')=\frac{\cos(\gamma)}{\cos(\frac{\pi}{4})}=\sqrt{2}\cos(\gamma)$ and $\cos(\gamma'')=\sqrt{2}\cos(\gamma')=2\cos(\gamma)$. We end up with the right part of the rule \add, and applying the rule with $\cos(\gamma'')$ gives the wanted condition on the angles.
\end{proof}

\begin{lem}%
\label{lem:white-dot-to-zx-generalised-form}~\\
Let $\rho\in\mathbb{R}+$. Then, for any $n_1,n_2\geq\max\left(0,\left\lceil \log_2(\rho)\right\rceil\right)$:
\[\left(\zxct+\add\right)\vdash~~\scalebox{0.91}{
%\InputIfFileExists{#1.tikz}{}{
\input{./figures/ZW-white-dot-to-ZX-generalised.tikz}%} % chktex 27
}\]
\end{lem}

\begin{proof}First we prove:
\def\fig{lemma-pi_3-branches-arxiv}
\begin{align*}
\left(\zxct+\add\right)\vdash~~\input{./figures/\fig/\fig_00.tikz}
\eq{\ref{cor:gen-supp}\\\ref{lem:inverse}\\\ref{lem:hopf}}\input{./figures/\fig/\fig_01.tikz}
\eq{\ref{lem:prod-cos}}\input{./figures/\fig/\fig_02.tikz}\\
\eq{\ref{lem:inverse}\\\ref{lem:k1}\\\ref{lem:multiplying-global-phases}}\input{./figures/\fig/\fig_03.tikz}
\eq{\ref{lem:supp-to-minus-pi_4}}\input{./figures/\fig/\fig_04.tikz}
\eq{\ref{lem:2-is-sqrt2-squared}\\\ref{lem:multiplying-global-phases}\\\ref{lem:inverse}}\input{./figures/\fig/\fig_05.tikz}
\end{align*}
%\[\left(\zxct+\add\right)\vdash\tikzfig{lemma-pi_3-branches}\]
We now show the result for $n\geq\max\left(0,\left\lceil \log_2(\rho)\right\rceil\right)$ and $n+1$, which then generalises to lemma~\ref{lem:white-dot-to-zx-generalised-form} by induction:
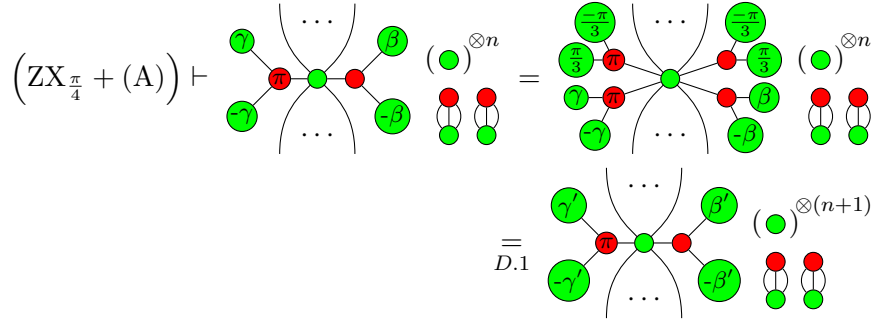
\def\fig{ZW-white-dot-to-ZX-generalised-proof}
\begin{align*}
\left(\zxct+\add\right)\vdash~~\input{./figures/\fig/\fig_00.tikz}
\eq{}\input{./figures/\fig/\fig_01.tikz}\\
\eq{\ref{lem:prod-cos}}\input{./figures/\fig/\fig_02.tikz}
\end{align*}
%\[\left(\zxct+\add\right)\vdash\tikzfig{ZW-white-dot-to-ZX-generalised-proof}\]
with:
\begin{align*}
\beta &= \arccos{\frac{\rho}{2^n}}\qquad
\gamma = \arccos{\frac{1}{2^n}}\\
\beta' &= \arccos{\frac{\rho}{2^n}\cos(\pi/3)}=\arccos{\frac{\rho}{2^{n+1}}}\\
\gamma' &= \arccos{\frac{1}{2^{n+1}}} \qedhere
\end{align*}
%$\beta = \arccos{\frac{\rho}{2^n}}$, $\gamma = \arccos{\frac{1}{2^n}}$, and $\beta' = \arccos{\frac{\rho}{2^n}\cos(\pi/3)}=\arccos{\frac{\rho}{2^{n+1}}}$, $\gamma' = \arccos{\frac{1}{2^{n+1}}}$
\end{proof}

\begin{cor}%
\label{cor:arccos-2^-n}
For any $n\in\mathbb{N}$, with $\gamma=\arccos{\frac{1}{2^n}}$:
\def\fig{lemma-arccos-1_2-to-the-n-branches}
\begin{align*}
\left(\zxct+\add\right)\vdash~~\input{./figures/\fig/\fig_00.tikz}
\eq{}\input{./figures/\fig/\fig_01.tikz}
\eq{}\input{./figures/\fig/\fig_02.tikz}
\end{align*}
\end{cor}

\begin{lem}%
\label{lem:gn-inverse-c}
The green node \begin{tikzpicture}
	\begin{pgfonlayer}{nodelayer}
		\node [style=gn] (0) at (0,0) {$\alpha$};
	\end{pgfonlayer}
\end{tikzpicture}
has an inverse if $\alpha\neq\pi\mod 2\pi$:
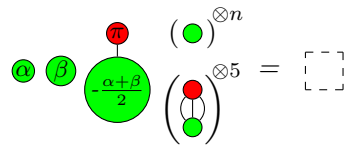
\[
%\InputIfFileExists{#1.tikz}{}{
\input{./figures/gn-alpha-inverse-no-triangle.tikz}%} % chktex 27
\]
for $n\geq \log_2\left(\frac{1}{|\cos(\alpha/2)|}\right)$ and $\beta = 2\arccos{\frac{1}{2^n\cos(\alpha/2)}}$.
\end{lem}

\begin{proof}
Notice that $\beta$ is well defined if $\alpha\neq\pi\mod 2\pi$. With these values of $n$ and $\beta$, $\cos(\alpha/2)\cos(\beta/2)=\cos(\gamma)$ with $\gamma = \arccos{\frac{1}{2^n}}$. Then:
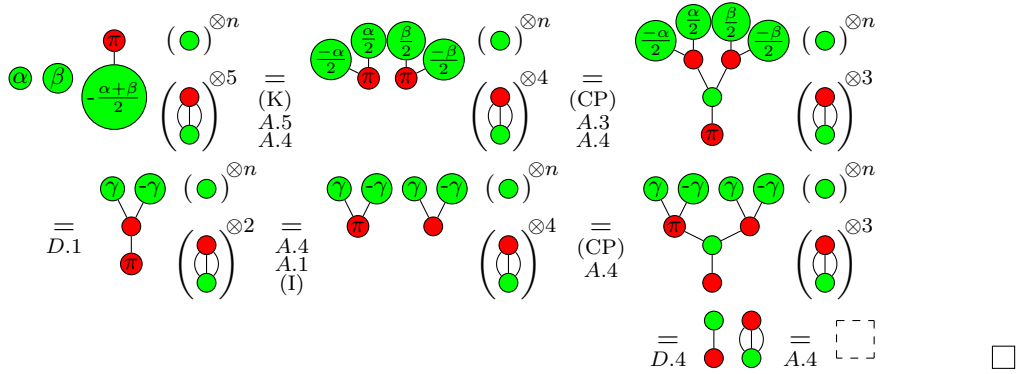
\def\fig{gn-alpha-inverse-no-triangle-proof}
\begin{align*}
\input{./figures/\fig/\fig_00.tikz}
\eq{\picom\\\ref{lem:multiplying-global-phases}\\\ref{lem:inverse}}\input{./figures/\fig/\fig_01.tikz}
\eq{\cp\\\ref{lem:k1}\\\ref{lem:inverse}}\input{./figures/\fig/\fig_02.tikz}\\
\eq{\ref{lem:prod-cos}}\input{./figures/\fig/\fig_03.tikz}
\eq{\ref{lem:inverse}\\\ref{lem:2-is-sqrt2-squared}\\\id}\input{./figures/\fig/\fig_04.tikz}
\eq{\cp\\\ref{lem:inverse}}\input{./figures/\fig/\fig_05.tikz}\\
\eq{\ref{cor:arccos-2^-n}}\input{./figures/\fig/\fig_06.tikz}
\eq{\ref{lem:inverse}}\input{./figures/\fig/\fig_07.tikz}
\\[-\normalbaselineskip]\tag*{\qedhere}
\end{align*}
\end{proof}

\subsection{Proof of Proposition~\ref{prop:X-is-homomorphism}}%
\label{prf:X-is-homomorphism}
%\begin{proof}[Proposition~\ref{prop:X-is-homomorphism}]
Since we have built the set of rules $\left(\zxct+\add\right)$ upon $\zxct$ which is complete for Clifford+T, we basically just need to prove the result for the ZW-rules in which a parameter (different from $\pm 1$) appears: $1c$, $3b$, $4a$, $4b$ and $6c$. Notice that the rule $0c$ is obvious.
\begin{itemize}
\item $1c$:
\def\fig{rule-1c-proof}
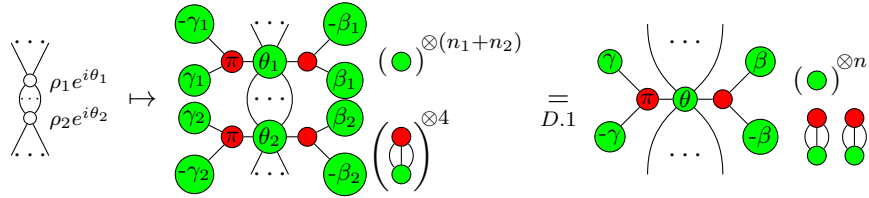
\begin{align*}
\input{./figures/\fig/\fig_00.tikz}
~~\mapsto~~\input{./figures/\fig/\fig_01.tikz}
\eq{\ref{lem:prod-cos}}\input{./figures/\fig/\fig_02.tikz}
\end{align*}
%\[\tikzfig{rule-1c-proof}\]
where:
\begin{align*}
n_k=\max\left(0,\left\lceil \log_2(\rho_k)\right\rceil\right)\qquad&\qquad
n=n_1+n_2\\
\beta_k = \arccos{\frac{\rho_k}{2^{n_k}}}\qquad&\qquad
\beta = \arccos{\frac{\rho}{2^n}}\\
\gamma_k = \arccos{\frac{1}{2^{n_k}}}\qquad&\qquad
\gamma = \arccos{\frac{1}{2^n}}
\end{align*}
%\begin{align*}
%n_k&=\max\left(0,\left\lceil \log_2(\rho_k)\right\rceil\right)\\
%\beta_k &= \arccos{\frac{\rho_k}{2^{n_k}}}\\
%\gamma_k &= \arccos{\frac{1}{2^{n_k}}}\\
%n&=n_1+n_2\\
%\beta &= \arccos{\frac{\rho}{2^n}}\\
%\gamma &= \arccos{\frac{1}{2^n}}
%\end{align*}
Notice that $\left\lceil \log_2(\rho_1\rho_2)\right\rceil=\left\lceil \log_2(\rho_1)+\log_2(\rho_2)\right\rceil \leq \left\lceil \log_2(\rho_1)\right\rceil+\left\lceil \log_2(\rho_2)\right\rceil$, so the result might not be precisely the one given by $\left[
%\InputIfFileExists{#1.tikz}{}{
\input{./figures/white-dot-rho1-rho2.tikz}%} % chktex 27
\right]_X$, but it can be patched thanks to lemma~\ref{lem:white-dot-to-zx-generalised-form}.
\item $3b$: corollary~\ref{cor:distribution}.
\item $4a$: suppose $\rho_1\geq\rho_2$, then using lemma~\ref{lem:white-dot-to-zx-generalised-form} to have the same $n$ on both sides:
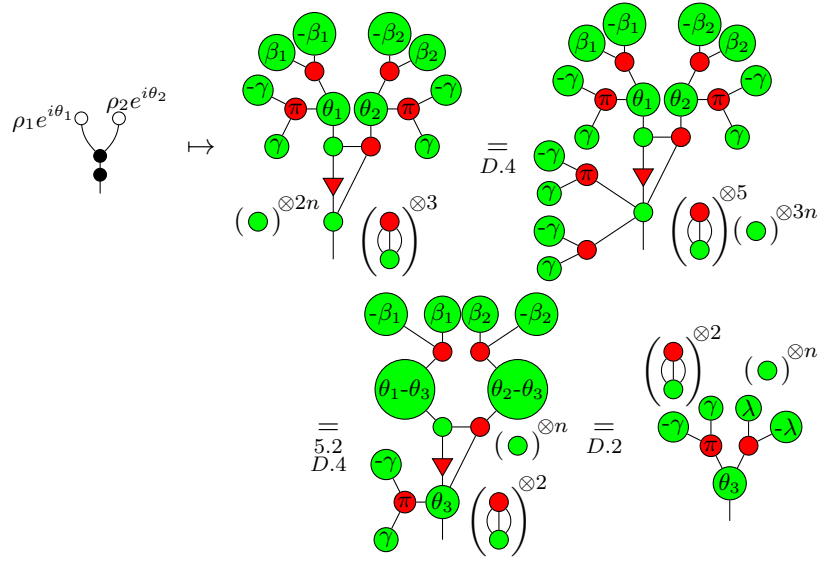
\def\fig{rule-4a-proof}
\begin{align*}
\input{./figures/\fig/\fig_00.tikz}
~~\mapsto~~\input{./figures/\fig/\fig_01.tikz}
\eq{\ref{cor:arccos-2^-n}}\input{./figures/\fig/\fig_02.tikz}\\
\eq{\ref{cor:distribution}\\\ref{cor:arccos-2^-n}}\input{./figures/\fig/\fig_03.tikz}
\eq{\ref{lem:add-bis}}\input{./figures/\fig/\fig_04.tikz}
\end{align*}
%\[\fit{\tikzfig{rule-4a-proof}}\]
with
\begin{align*}
\beta_k &= \arccos{\frac{\rho_k}{2^n}}\qquad
\gamma = \arccos{\frac{1}{2^n}}\\
\theta_3 &=\arg(\rho_1e^{i\theta_1}+\rho_2e^{i\theta_2})\\
\lambda &= \arccos{e^{i\theta_1-\theta_3}\cos{\beta_1}+e^{i\theta_2-\theta_3}\cos{\beta_2}}\\
 &= \arccos{\frac{\rho_1e^{i\theta_1}+\rho_2e^{i\theta_2}}{e^{i\theta_3}2^n}}
\end{align*}
which is what $\left[
%\InputIfFileExists{#1.tikz}{}{
\begin{tikzpicture}
	\begin{pgfonlayer}{nodelayer}
		\node [anchor=west, style=none] (0) at (0.25, 0.5) {\footnotesize $\rho_1e^{i\theta_1}+\rho_2e^{i\theta_2}$};
		\node [style=white dot] (1) at (0, 0.5) {};
		\node [style=none] (2) at (0, -0.5) {};
	\end{pgfonlayer}
	\begin{pgfonlayer}{edgelayer}
		\draw (1) to (2.center);
	\end{pgfonlayer}
\end{tikzpicture}%} % chktex 27
\right]_X$ gives.
\item $4b$:
\def\fig{rule-4b-proof}
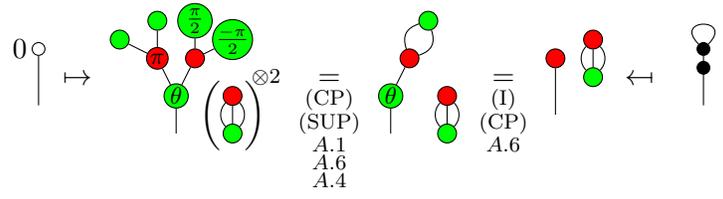
\begin{align*}
\input{./figures/\fig/\fig_00.tikz}
~~\mapsto~~\input{./figures/\fig/\fig_01.tikz}
\eq{\cp\\\supp\\\ref{lem:2-is-sqrt2-squared}\\\ref{lem:bicolor-0-alpha}\\\ref{lem:inverse}}\input{./figures/\fig/\fig_02.tikz}
\eq{\id\\\cp\\\ref{lem:bicolor-0-alpha}}\input{./figures/\fig/\fig_03.tikz}
~~\mapsfrom~~\input{./figures/\fig/\fig_04.tikz}
\end{align*}
\item $6c$: First, using corollary~\ref{cor:arccos-2^-n}, for $\gamma = \arccos{\frac{1}{2^n}}$,
\def\fig{lemma-for-rule-6c}
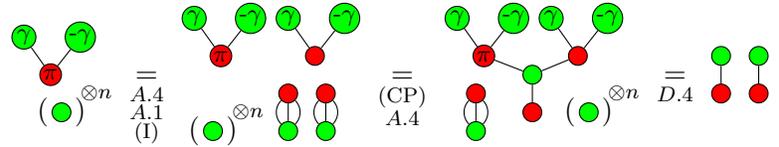
\begin{align*}
\input{./figures/\fig/\fig_00.tikz}
\eq{\ref{lem:inverse}\\\ref{lem:2-is-sqrt2-squared}\\\id}\input{./figures/\fig/\fig_01.tikz}
\eq{\cp\\\ref{lem:inverse}}\input{./figures/\fig/\fig_02.tikz}
\eq{\ref{cor:arccos-2^-n}}\input{./figures/\fig/\fig_03.tikz}
\end{align*}
then, with:
\begin{align*}
n &=\max\left(0,\left\lceil \log_2(\rho)\right\rceil\right)\\
\beta &= \arccos{\frac{\rho}{2^{n}}}\\
\gamma &= \arccos{\frac{1}{2^{n}}}
\end{align*}
\def\fig{rule-6c-proof-2}
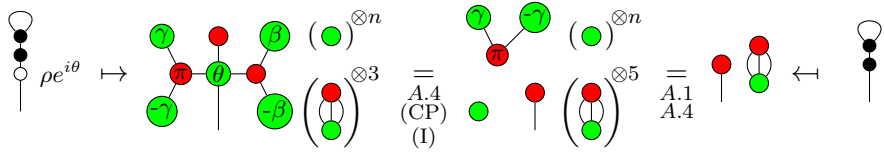
\begin{align*}
\input{./figures/\fig/\fig_00.tikz}
~~\mapsto~~\input{./figures/\fig/\fig_01.tikz}
\eq{\ref{lem:inverse}\\\cp\\\id}\input{./figures/\fig/\fig_02.tikz}
\eq{\ref{lem:2-is-sqrt2-squared}\\\ref{lem:inverse}}\input{./figures/\fig/\fig_03.tikz}
~~\mapsfrom~~\input{./figures/\fig/\fig_04.tikz}
\end{align*}
This is enough to show that rule $6c$ stands, because the case $r=1$ has already been treated to show the completeness of $\zx$ for Clifford+T.
\qed%\end{proof}
\end{itemize}

\section{Parametrised Triangles}

\subsection{Lemmas}

\begin{proof}[Proof of Lemma~\ref{lem:gn-inverse}]
\phantomsection\label{prf:gn-inverse}
For $2\alpha$ where $\alpha\neq\frac{\pi}{2}\mod\pi$:
\def\fig{gn-alpha-inverse-proof}
\begin{align*}
\zxt\vdash~~
\input{./figures/\fig/\fig_00.tikz}
\hspace{-1em}\eq{\td}\input{./figures/\fig/\fig_01.tikz}
\eq{\ref{lem:k1}\\\cp\\\ref{lem:bicolor-0-alpha}\\\ref{lem:multiplying-global-phases}}\input{./figures/\fig/\fig_02.tikz}
\eq{\ref{lem:multiplying-global-phases}\\\picom}\input{./figures/\fig/\fig_03.tikz}
\eq{\ref{lem:multiplying-global-phases}\\\ref{lem:bicolor-0-alpha}\\\ref{lem:2-is-sqrt2-squared}\\\ref{lem:inverse}}\input{./figures/\fig/\fig_04.tikz}
\end{align*}
\end{proof}

\begin{proof}[Proof of Lemma~\ref{lem:triangle-times-exp}]
\phantomsection\label{prf:triangle-times-exp}
Let $r = e^{i\gamma}\tan(\alpha)$. Then:
\def\fig{triangle-times-exp-proof}
\begin{align*}
\input{./figures/\fig/\fig_00.tikz}
\eq{\ref{lem:gn-inverse}}\input{./figures/\fig/\fig_01.tikz}
\eq{\td}\input{./figures/\fig/\fig_02.tikz}
\eq{}\input{./figures/\fig/\fig_03.tikz}\\
\eq{\td}\input{./figures/\fig/\fig_04.tikz}
\eq{\ref{lem:gn-inverse}}\input{./figures/\fig/\fig_05.tikz}
\end{align*}
\end{proof}

\begin{proof}[Proof of Lemma~\ref{lem:param-triangle-0}]
\phantomsection\label{prf:param-triangle-0}
\def\fig{triangle-0}
\begin{align*}
\input{./figures/\fig/\fig_00.tikz}
\eq{\ref{lem:inverse}\\\ref{lem:2-is-sqrt2-squared}}\input{./figures/\fig/\fig_01.tikz}
\eq{\td}\input{./figures/\fig/\fig_02.tikz}
\eq{\cp\\\id\\\ref{lem:inverse}\\\ref{lem:green-state-pi_2-is-red-state-minus-pi_2}}\input{./figures/\fig/\fig_03.tikz}
\eq{\id}\input{./figures/\fig/\fig_04.tikz}
\end{align*}
\end{proof}

\begin{proof}[Proof of Lemma~\ref{lem:parallel-param-triangles}]
\phantomsection\label{prf:parallel-param-triangles}
Thanks to Theorem~\ref{thm:lin-diag-T}, we get:
%First of all, let us use Theorem~\ref{thm:lin-diag} once more:
\begin{equation}
\label{eq:distrib-T}
\forall r,a,b\in\mathbb{C},~~~\zxt\vdash ~
%\InputIfFileExists{#1.tikz}{}{
\input{./figures/distribution-lemma.tikz}%} % chktex 27

\end{equation}
%This equality is semantically correct: it basically says $r(a+b)=ra+rb$; and there exist $\rho,\rho_a,\rho_b\in[0;\pi/2[$ and $\theta,\theta_a,\theta_b\in\mathbb{R}$ such that $r=\tan(\rho)e^{i\theta},a=\tan(\rho_a)e^{i\theta_a},b=\tan(\rho_b)e^{i\theta_b}$, so that:
%\def\fig{distribution-lemma-proof}
%\begin{align*}
%\tikzfigc{00}
%\eq{\td}\tikzfigc{01}
%\eq{\picom}\tikzfigc{02}\\
%\eq{Thm\~\ref{thm:lin-diag}}\tikzfigc{03}
%\eq{\td}\tikzfigc{04}
%\end{align*}
%Scalars being equal and invertible on both sides, we end up with (D). Then:
Also:
\begin{equation}
\label{eq:parallel-param-triangles-integer}
\forall (r,n,\theta)\in\mathbb{C}{\times}\mathbb{N}{\times}\mathbb{R},~~~ \zxt\vdash 
%\InputIfFileExists{#1.tikz}{}{
\input{./figures/parallel-param-triangles-integer.tikz}%} % chktex 27

\end{equation}
Indeed:
\def\fig{parallel-param-triangles-integer-proof}
\begin{align*}
\input{./figures/\fig/\fig_00.tikz}
\eq{\ref{lem:triangle-times-exp}\\\ta}\input{./figures/\fig/\fig_01.tikz}
\eq{Thm~\ref{thm:lin-diag-T}}\input{./figures/\fig/\fig_02.tikz}
\eq{\ta}\input{./figures/\fig/\fig_03.tikz}
\eq{\td}\input{./figures/\fig/\fig_04.tikz}
\end{align*}
Finally, let $s\in\mathbb{C}$. Then there exist $n\in\mathbb{N}, \theta,\alpha\in\mathbb{R}$ such that $s=n(e^{i\theta}+e^{-i\theta})e^{i\alpha}$. We then deduce:
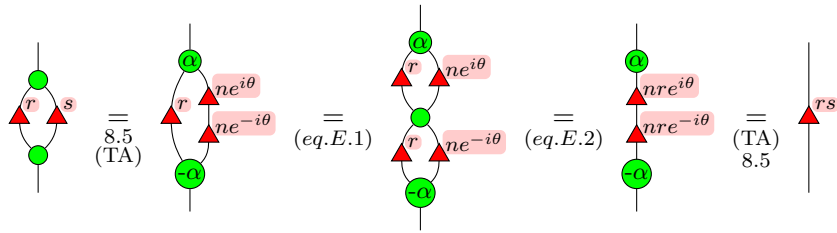
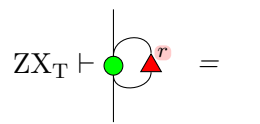
\def\fig{parallel-param-triangles-real-proof}
\begin{align*}
\input{./figures/\fig/\fig_00.tikz}
\eq{\ref{lem:triangle-times-exp}\\\ta}\input{./figures/\fig/\fig_01.tikz}
\eq{(\ref{eq:distrib-T})}\input{./figures/\fig/\fig_02.tikz}
\eq{(\ref{eq:parallel-param-triangles-integer})}\input{./figures/\fig/\fig_03.tikz}
\eq{\ta\\\ref{lem:triangle-times-exp}}\input{./figures/\fig/\fig_04.tikz}
\end{align*}
\end{proof}

\begin{lem}%
\label{lem:triangle-looping}
\[\zxt\vdash
%\InputIfFileExists{#1.tikz}{}{
\input{./figures/triangle-looping-on-green-state.tikz}%} % chktex 27
\]
\end{lem}
\begin{proof} Let $\alpha\neq \pi/2\mod \pi$ and $\theta$ be such that $r=e^{i\theta}\tan(\alpha)$,
\def\fig{rule-1c-proof-t-prelim}
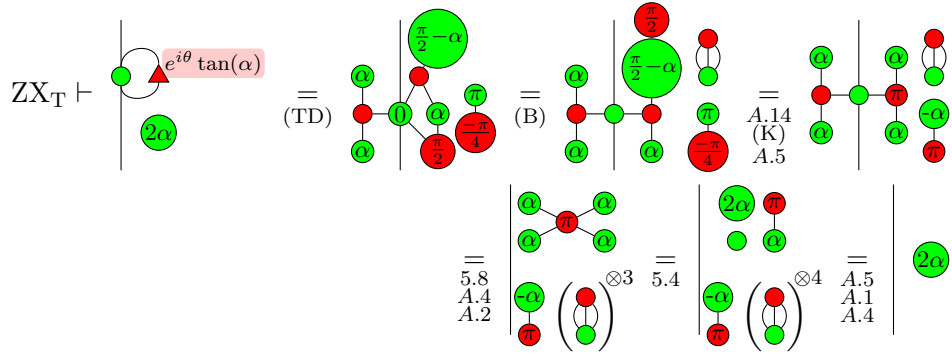
\begin{align*}
\zxt\vdash~~\input{./figures/\fig/\fig_00.tikz}
\hspace{-1em}\eq{\td}\input{./figures/\fig/\fig_01.tikz}
\eq{\bialg}\input{./figures/\fig/\fig_02.tikz}
\eq{\ref{lem:green-state-pi_2-is-red-state-minus-pi_2}\\\picom\\\ref{lem:multiplying-global-phases}}\input{./figures/\fig/\fig_03.tikz}\\
\eq{\ref{cor:gen-supp}\\\ref{lem:inverse}\\\ref{lem:hopf}}\input{./figures/\fig/\fig_04.tikz}
\eq{\ref{cor:big-scalar-equation}}\input{./figures/\fig/\fig_05.tikz}
\eq{\ref{lem:multiplying-global-phases}\\\ref{lem:2-is-sqrt2-squared}\\\ref{lem:inverse}}\input{./figures/\fig/\fig_06.tikz}
\end{align*}
and thanks to lemma~\ref{lem:gn-inverse}, \begin{tikzpicture}
	\begin{pgfonlayer}{nodelayer}
		\node [style=gn] (0) at (0, -0) {$2\alpha$};
	\end{pgfonlayer}
\end{tikzpicture}
can be removed on both sides.
\end{proof}

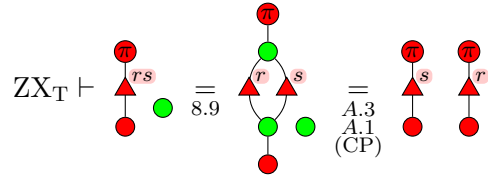
\begin{lem}%
\label{lem:param-triangle-scalar-product}
%\[\tikzfig{param-triangle-scalar-product}\]
%\end{lem}
%\begin{proof}
\def\fig{param-triangle-scalar-product-proof}
\[\zxt\vdash~~\input{./figures/\fig/\fig_00.tikz}
\eq{\ref{lem:parallel-param-triangles}}\input{./figures/\fig/\fig_01.tikz}
\eq{\ref{lem:k1}\\\ref{lem:2-is-sqrt2-squared}\\\cp}\input{./figures/\fig/\fig_02.tikz}\]
%\end{proof}
\end{lem}

\subsection{Proof of Prop.~\ref{thm:completeness-t}}%
\label{prf:completeness-t}
%\begin{proof}[Proposition~\ref{thm:completeness-t}]
First we need to prove $\zxt\vdash D=\left[[D]_W\right]_X$. Nearly all the work was done in proposition~\ref{prop:double-interpretation-equivalence-1} (thanks to proposition~\ref{prop:ZX-param-triangles-contains-ZX-pi_4}), we just need to check the result for:
\begin{itemize}
\item parametrised triangles:
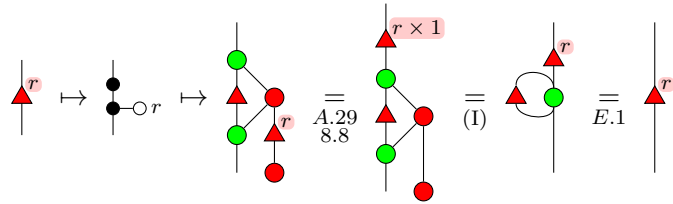
\def\fig{param-triangle-double-interpretation}
\begin{align*}
\input{./figures/\fig/\fig_00.tikz}~~\mapsto~~\input{./figures/\fig/\fig_01.tikz}~~\mapsto~~
\input{./figures/\fig/\fig_02.tikz}
\eq{\ref{lem:black-dot-swappable-outputs}\\\ref{lem:C2-gen}}\input{./figures/\fig/\fig_03.tikz}
\eq{\id}\input{./figures/\fig/\fig_04.tikz}
\eq{\ref{lem:triangle-looping}}\input{./figures/\fig/\fig_05.tikz}
\end{align*}
\item Hadamard: Here the problem relies on the scalars. First:
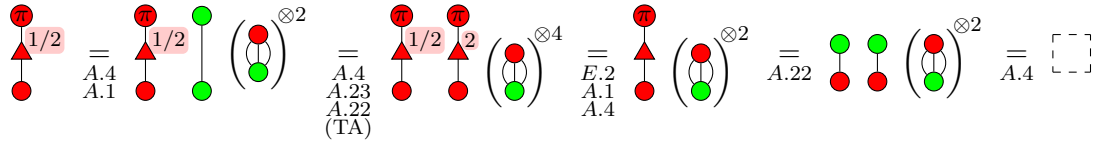
\def\fig{param-triangle-scalar-1_2}
\begin{align*}
\input{./figures/\fig/\fig_00.tikz}
\hspace{-1em}\eq{\ref{lem:inverse}\\\ref{lem:2-is-sqrt2-squared}}\input{./figures/\fig/\fig_01.tikz}
\hspace{-1em}\eq{\ref{lem:inverse}\\\ref{lem:red-state-on-upside-down-triangle}\\\ref{lem:pi-red-state-on-triangle}\\\ta}\input{./figures/\fig/\fig_02.tikz}
\eq{\ref{lem:param-triangle-scalar-product}\\\ref{lem:2-is-sqrt2-squared}\\\ref{lem:inverse}}\input{./figures/\fig/\fig_03.tikz}
\eq{\ref{lem:pi-red-state-on-triangle}}\input{./figures/\fig/\fig_04.tikz}
\eq{\ref{lem:inverse}}\input{./figures/\fig/\fig_05.tikz}
\end{align*}
Then, using $\tan(\pi/8)=\sqrt{2}-1$:
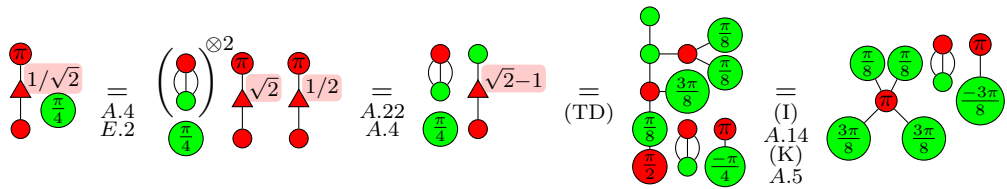
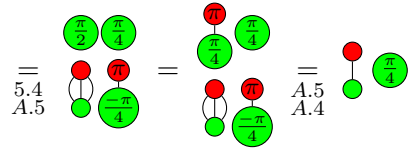
\def\fig{param-triangle-scalar-1_sqrt2}
\begin{align*}
\input{./figures/\fig/\fig_00.tikz}
\eq{\ref{lem:inverse}\\\ref{lem:param-triangle-scalar-product}}\input{./figures/\fig/\fig_01.tikz}
\eq{\ref{lem:pi-red-state-on-triangle}\\\ref{lem:inverse}}\input{./figures/\fig/\fig_02.tikz}
\eq{\td}\input{./figures/\fig/\fig_03.tikz}
\eq{\id\\\ref{lem:green-state-pi_2-is-red-state-minus-pi_2}\\\picom\\\ref{lem:multiplying-global-phases}}\input{./figures/\fig/\fig_04.tikz}\\
\eq{\ref{cor:big-scalar-equation}\\\ref{lem:multiplying-global-phases}}\input{./figures/\fig/\fig_05.tikz}
\hspace{-1em}\eq{}\input{./figures/\fig/\fig_06.tikz}
\eq{\ref{lem:multiplying-global-phases}\\\ref{lem:inverse}}\input{./figures/\fig/\fig_07.tikz}
\end{align*}
Again, the $\frac{\pi}{4}$-green node can be removed thanks to Lemma~\ref{lem:gn-inverse}. Finally:
\def\fig{hadamard-double-interpretation-t}
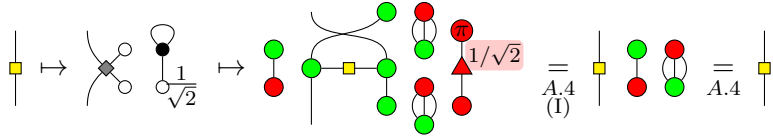
\begin{align*}
\input{./figures/\fig/\fig_00.tikz}~~\mapsto~~\input{./figures/\fig/\fig_01.tikz}~~\mapsto~~\input{./figures/\fig/\fig_02.tikz}
\eq{\ref{lem:inverse}\\\id}\input{./figures/\fig/\fig_03.tikz}
\eq{\ref{lem:inverse}}\input{./figures/\fig/\fig_04.tikz}
\end{align*}
Then we need to show $\zw\vdash D_1=D_2 \implies \zxt\vdash [D_1]_X=[D_2]_X$, and again, most of the work has already been done. We only need to prove it for rules $1c$, $2b$, $2c$, $3b$, $4a$, $4b$ and $6c$.
\item $1c$: Using~\ref{lem:triangle-looping},
\def\fig{rule-1c-proof-t}
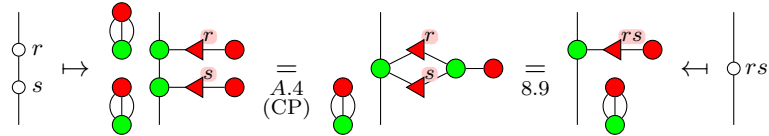
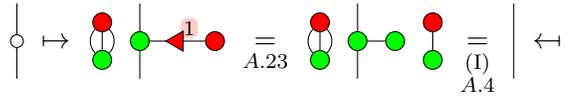
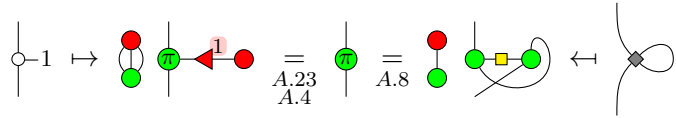
\begin{align*}
\input{./figures/\fig/\fig_00.tikz}~~\mapsto~~\input{./figures/\fig/\fig_01.tikz}
\eq{\ref{lem:inverse}\\\cp}\input{./figures/\fig/\fig_02.tikz}
\eq{\ref{lem:parallel-param-triangles}}\input{./figures/\fig/\fig_03.tikz}
~~\mapsfrom~~\input{./figures/\fig/\fig_04.tikz}
\end{align*}
\item $2b$:
\def\fig{rule-2b-proof}
\begin{align*}
\input{./figures/\fig/\fig_00.tikz}~~\mapsto~~\input{./figures/\fig/\fig_01.tikz}
\eq{\ref{lem:red-state-on-upside-down-triangle}}\input{./figures/\fig/\fig_02.tikz}
\eq{\id\\\ref{lem:inverse}}\input{./figures/\fig/\fig_03.tikz}
~~\mapsfrom~~\input{./figures/\fig/\fig_04.tikz}
\end{align*}
\item $2c$:
\def\fig{rule-2c-proof}
\begin{align*}
\input{./figures/\fig/\fig_00.tikz}~~\mapsto~~\input{./figures/\fig/\fig_01.tikz}
\eq{\ref{lem:red-state-on-upside-down-triangle}\\\ref{lem:inverse}}\input{./figures/\fig/\fig_02.tikz}
\eq{\ref{lem:h-loop}}\input{./figures/\fig/\fig_03.tikz}
~~\mapsfrom~~\input{./figures/\fig/\fig_04.tikz}
\end{align*}
\item $3b$ is directly given by lemma~\ref{lem:param-tri-distri}.
\item $4a$:
\def\fig{rule-4a-proof-t}
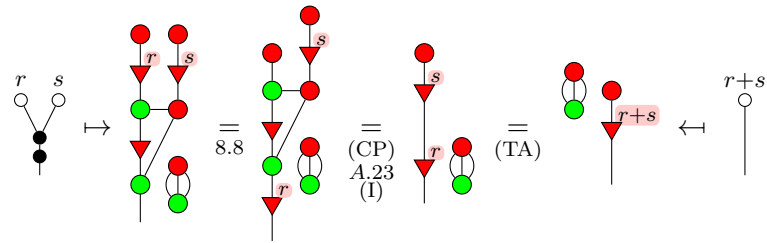
\begin{align*}
\input{./figures/\fig/\fig_00.tikz}~~\mapsto~~\input{./figures/\fig/\fig_01.tikz}
\eq{\ref{lem:C2-gen}}\input{./figures/\fig/\fig_02.tikz}
\eq{\cp\\\ref{lem:red-state-on-upside-down-triangle}\\\id}\input{./figures/\fig/\fig_03.tikz}
\eq{\ta}\input{./figures/\fig/\fig_04.tikz}
~~\mapsfrom~~\input{./figures/\fig/\fig_05.tikz}
\end{align*}
\item $4b$: using $\tan(0)=0$,
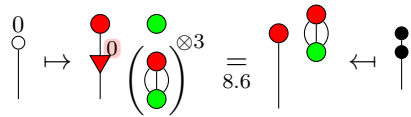
\def\fig{rule-4b-proof-t}
\begin{align*}
\input{./figures/\fig/\fig_00.tikz}~~\mapsto~~\input{./figures/\fig/\fig_01.tikz}
%\eq{\td}\tikzfigc{02}
%\eq{\id\\\cp\\\ref{lem:2-is-sqrt2-squared}\\\ref{lem:inverse}}\tikzfigc{03}
\eq{\ref{lem:param-triangle-0}}\input{./figures/\fig/\fig_04.tikz}
~~\mapsfrom~~\input{./figures/\fig/\fig_05.tikz}
\end{align*}
\item $6c$: again, it can be reduced to showing:
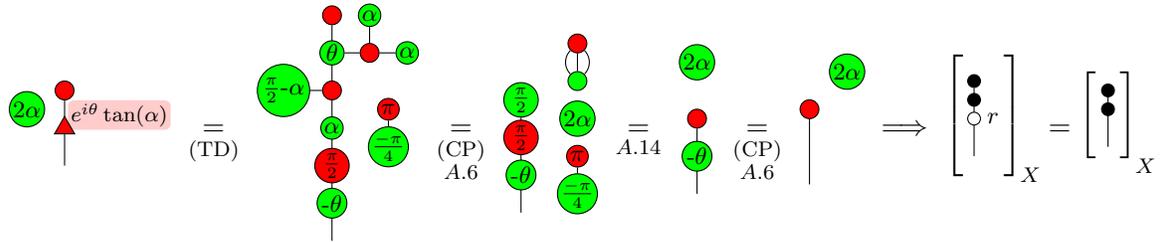
\def\fig{red-state-on-param-triangle}
\begin{align*}
\input{./figures/\fig/\fig_00.tikz}
\eq{\td}\input{./figures/\fig/\fig_01.tikz}
\eq{\cp\\\ref{lem:bicolor-0-alpha}}\input{./figures/\fig/\fig_02.tikz}
\eq{\ref{lem:green-state-pi_2-is-red-state-minus-pi_2}}\input{./figures/\fig/\fig_03.tikz}
\eq{\cp\\\ref{lem:bicolor-0-alpha}}\input{./figures/\fig/\fig_04.tikz}
\implies\input{./figures/\fig/\fig_05.tikz}=\input{./figures/\fig/\fig_06.tikz}
\end{align*}
where \begin{tikzpicture}
	\begin{pgfonlayer}{nodelayer}
		\node [style=gn] (0) at (0, -0) {$2\alpha$};
	\end{pgfonlayer}
\end{tikzpicture}
can be removed on both sides (Lemma~\ref{lem:gn-inverse}). The rule $6c$ for $r=1$ has been shown for Clifford+T.
\qed%\end{proof}
\end{itemize}

\end{document}